\documentclass[
  a4paper,
  runningheads,
  orivec
]{llncs}

\usepackage[hidelinks]{hyperref}

\usepackage{amsmath}
\usepackage{amssymb}
\usepackage{stmaryrd}
\usepackage{wasysym}
\usepackage{mathtools}
\usepackage{thmtools}

\usepackage{cleveref}

\usepackage[basic]{complexity}
\usepackage{proof}

\usepackage{csquotes}
\usepackage{xparse}
\usepackage{xspace}

\usepackage{cite}
\usepackage{booktabs}
\usepackage{subcaption}
\usepackage[inline]{enumitem}

\usepackage[pass]{geometry}

\AtEndEnvironment{example}{\def\squareforqed{$\lrcorner$}\qed}
\declaretheorem[name={Theorem~\ref{thm:extension-n-soundness}, part}]{thmSoundnessPart}

\let\oldxrightarrow\xrightarrow
\newcommand\myxrightarrow[2][]{
    \oldxrightarrow[{\raisebox{1.7ex-\heightof{$\scriptstyle#1$}}[0pt]{$\scriptstyle#1$}}]{{\raisebox{-0.2ex}[0pt][0pt]{$\scriptstyle#2$}}}}
\let\xrightarrow\myxrightarrow

\usepackage{tikz}
\usepackage{pgfplots}
\usetikzlibrary{
    automata,
    arrows.meta,
    calc,
    positioning,
    backgrounds,
    shapes.geometric,
    shapes.misc,
    plotmarks,
    patterns,
    decorations.pathreplacing,
    fit,
}
\tikzset {
  edge with arrow/.style = {
    ->,
    >=stealth,
    shorten >=1pt,
  },
  directed/.style = {
    edge with arrow,
    node distance=2cm,
    on grid,
    semithick,
    double distance=1.5pt,
  },
  automaton/.style = {
    directed,
    auto,	initial text={},
    state/.append style = {
      ellipse,
      inner sep = 0pt,
      minimum size = 10pt,
    },
  },
  timer graph/.style = {
    on grid,
    node distance = 15pt and 60pt,
  },
  timer/.style = {
    circle,
    fill = black,
    inner sep = 0pt,
    minimum size = 5pt,
  },
  set of timers/.style = {
    draw,
    ellipse,
    inner xsep = 12pt,
    inner ysep = 5pt,
  },
  mapping/.style = {
    edge with arrow,
    black,
  },
  apartness/.style = {
    densely dashed,
    red,
  },
  basis/.style = {
    fill = lightgray!70!white,
  },
}
\usepgfplotslibrary{fillbetween}
\pgfplotsset{
  compat=1.18
}

\usepackage[outline]{contour}
\contourlength{0.7pt}
\newcommand{\treeNodeLabel}[1]{\contour{white}{#1}} 
\renewcommand{\iff}{\Leftrightarrow}
\renewcommand{\implies}{\Rightarrow}

\NewDocumentCommand{\equivClass}{m}{\llbracket {#1} \rrbracket}

\makeatletter
\newcommand{\oset}[2]{{\mathop{#2}\limits^{\vbox to \ex@{\kern-\tw@\ex@
   \hbox{\scriptsize #1}\vss}}}}
\makeatother

\NewDocumentCommand{\complexity}{m}{\mathcal{O}\left(#1\right)}
\newcommand{\lsharp}{\ensuremath{L^{\#}}\xspace}
\newcommand{\lsharpMMT}{\ensuremath{L^{\#}_{\text{\MMT}}}\xspace}
\newcommand{\partto}{\rightharpoonup}
\newcommand{\timeout}[1]{\mathit{to}[#1]}
\newcommand{\toevents}[1]{\mathit{TO}[#1]}
\newcommand{\subsets}[1]{{\mathcal{P}}(#1)}  \newcommand{\nat}{{\mathbb N}}
\newcommand{\natplus}{\nat^{>0}}
\newcommand{\nnr}{{\mathbb R}^{\geq 0}}
\newcommand{\rplus}{{\mathbb R}^{>0}}

 \newcommand{\automaton}[1]{{\mathcal{#1}}}
\newcommand{\M}{\automaton{M}}

\newcommand{\N}{\automaton{N}}
\newcommand{\T}{\automaton{T}}

\newcommand{\zoneOf}[1]{\mathit{zone}(#1)}
\NewDocumentCommand{\hypothesis}{}{\automaton{H}}
\NewDocumentCommand{\tree}{}{\T}
\newcommand{\MMT}{MMT\xspace}
\newcommand{\MMTs}{MMTs\xspace}
\newcommand{\gMMT}{gMMT\xspace}
\newcommand{\gMMTs}{gMMTs\xspace}
\NewDocumentCommand{\equivalent}{}{\!\!\mathrel{\oset{time}{\approx}}\!\!}

\NewDocumentCommand{\symEquivalent}{}{\!\!\mathrel{\oset{sym}{\approx}}\!\!}
\NewDocumentCommand{\notSymEquivalent}{}{\!\!\mathrel{\oset{sym}{\not\approx}}\!\!}
\NewDocumentCommand{\ttequivalent}{}{\!\!\mathrel{\oset{tt}{\approx}}\!\!}
\NewDocumentCommand{\notttequivalent}{}{\!\!\mathrel{\oset{tt}{\not\approx}}\!\!}
\newcommand{\robust}{race-avoiding\xspace}
\newcommand{\Robust}{Race-avoiding\xspace}
\newcommand{\complete}{complete\xspace}
\newcommand{\good}{s-learnable\xspace}

\newcommand{\structural}{structural\xspace}
\newcommand{\behavioral}{behavioral\xspace}

\newcommand{\valuation}{\kappa}
\newcommand{\dom}[1]{{\textsf{dom}}(#1)} \newcommand{\ran}[1]{{\textsf{ran}}(#1)}  
\newcommand{\runs}[1]{\mathit{runs}(#1)} \newcommand{\run}[2]{{\mathit{run}}^{#1}(#2)} \newcommand{\length}[1]{\mathit{length}(#1)} \newcommand{\last}[1]{\mathit{last}(#1)} \newcommand{\cause}[2]{\mathit{cs}(#1,#2)} \newcommand{\causetwo}[2]{\mathit{cs}_2(#1,#2)}    \newcommand{\tinpw}[1]{\mathit{tiw}(#1)}   \newcommand{\conf}[1]{\cal{C}^#1} \newcommand{\symbtree}[1]{\mathit{symbtree}(#1)} \NewDocumentCommand{\tiwruns}{m O{}}{\mathit{tiwruns}^{#2}(#1)} \newcommand{\tOutputs}{\mathit{toutputs}} \NewDocumentCommand{\untimeRun}{m}{\mathit{untime}(#1)} 

\NewDocumentCommand{\tiw}{}{tiw\xspace}

\NewDocumentCommand{\tow}{}{tow\xspace}

\NewDocumentCommand{\siw}{}{sw\xspace}
\NewDocumentCommand{\siws}{}{sws\xspace}
\NewDocumentCommand{\symbolic}{m}{\mathtt{{#1}}}
\NewDocumentCommand{\toSymbolic}{m}{\overline{{#1}}}
\newcommand{\sinpw}[1]{\mathit{sw}(#1)} \newcommand{\timedtrace}[1]{\mathit{tt}(#1)} 
\newcommand{\silanguage}[1]{L_\mathit{sym}(#1)} 
\newcommand{\ttlanguage}[1]{L_\mathit{tt}(#1)}
\newcommand{\ofunction}[2]{\mathit{out}^{#1}(#2)} \newcommand{\outw}[1]{\mathit{ow}(#1)} \newcommand{\w}{\symbolic{w}}

\NewDocumentCommand{\actions}{m}{{A({#1})}}
\newcommand{\symbActions}{\mathtt{A}}
\NewDocumentCommand{\updates}{m}{{U({#1})}}

\NewDocumentCommand{\updateFunction}{}{\mathfrak{r}}

\NewDocumentCommand{\renamingMMT}{}{\mu}

\newcommand{\lengthOf}[1]{{\lvert #1 \rvert}}
\newcommand{\emptyword}{\varepsilon}

\NewDocumentCommand{\PER}{}{\equiv}

\NewDocumentCommand{\maxTime}{}{\Delta}

\NewDocumentCommand{\constraints}{}{\mathrm{cnstr}}

\NewDocumentCommand{\fractionalPart}{}{\mathrm{frac}}

\ExplSyntaxOn
\NewDocumentCommand{\copyRun}{
  O{m}
  m
  O{ p \c_math_subscript_token 0 \xrightarrow{w} p \c_math_subscript_token n}}
{\tl_if_empty:nTF{#2}{\mathit{read}^{#1}\c_math_subscript_token{#3}}{\mathit{read}^{#1}\c_math_subscript_token{#3}(#2)}}
\ExplSyntaxOff
\ExplSyntaxOn
\NewDocumentCommand{\replay}{
  O{m}
  m
  O{ p \c_math_subscript_token 0 \xrightarrow{w} p \c_math_subscript_token n}}
{\tl_if_empty:nTF{#2}{\mathit{replay}^{#1}\c_math_subscript_token{#3}}{\mathit{replay}^{#1}\c_math_subscript_token{#3}(#2)}}
\ExplSyntaxOff
\NewDocumentCommand{\activeTimers}{}{{\chi}}
\NewDocumentCommand{\enabled}{m O{}}{{\activeTimers_0^{#2}(#1)}}
\NewDocumentCommand{\funcSim}{}{\langle f, g \rangle}

\NewDocumentCommand{\apart}{}{\mathbin{\#}}
\NewDocumentCommand{\timerApart}{}{\mathbin{\rule{0pt}{5pt}^t\!\!\#}}

\NewDocumentCommand{\witness}{s}{\IfBooleanTF{#1}{\witness^*}{\vdash}}

\NewDocumentCommand{\basis}{O{}}{{\mathcal{B}^{\tree_{#1}}}}
\NewDocumentCommand{\frontier}{O{}}{{\mathcal{F}^{\tree_{#1}}}}
\NewDocumentCommand{\compatible}{O{}}{\mathit{compat}^{\tree_{#1}}}
\NewDocumentCommand{\explored}{}{{\mathcal{E}^\tree}}
\NewDocumentCommand{\matchingRun}{m m m}{{#1}^{{#2}}_{{#3}}}

\NewDocumentCommand{\Val}{m}{\mathsf{Val}({#1})} 

\NewDocumentCommand{\outputQ}{}{\ensuremath{\mathbf{OQ}}\xspace}
\NewDocumentCommand{\equivQ}{}{\ensuremath{\mathbf{EQ}}\xspace}
\NewDocumentCommand{\symOutputQ}{}{\ensuremath{\mathbf{OQ}^{\mathbf{s}}}\xspace}
\NewDocumentCommand{\symWaitQ}{}{\ensuremath{\mathbf{WQ}^{\mathbf{s}}}\xspace}
\NewDocumentCommand{\symEquivQ}{}{\ensuremath{\mathbf{EQ}^{\mathbf{s}}}\xspace}
\NewDocumentCommand{\yes}{}{\ensuremath{\mathbf{yes}}\xspace}
\NewDocumentCommand{\no}{}{\ensuremath{\mathbf{no}}\xspace}
\NewDocumentCommand{\APART}{}{\ensuremath{\mathrm{APART}}\xspace}
\NewDocumentCommand{\ACTIVE}{}{\ensuremath{\mathrm{ACTIVE}}\xspace}
\NewDocumentCommand{\DONE}{}{\ensuremath{\mathrm{DONE}}\xspace}

\newclass{\THREEEXP}{3EXP}
\newclass{\TWOEXP}{2EXP}

\NewDocumentCommand{\seismic}{}{\textbf{Seismic}\xspace}
\NewDocumentCommand{\promotion}{}{\textbf{Promotion}\xspace}
\NewDocumentCommand{\completion}{}{\textbf{Completion}\xspace}
\NewDocumentCommand{\minimizationActive}{}{\textbf{Active timers}\xspace}
\NewDocumentCommand{\minimizationCoTrans}{}{\textbf{WCT}\xspace}

\NewDocumentCommand{\foldFunction}{}{\boldsymbol{h}}

\NewDocumentCommand{\updateFig}{m m}{#1 \coloneq #2}
 
\def\orcidID#1{\smash{\href{http://orcid.org/#1}{\protect\raisebox{-1.25pt}{\protect\includegraphics{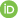}}}}}

\allowdisplaybreaks

\begin{document}

\newgeometry{hmargin = 2.5cm} 

\title{Active Learning of Mealy Machines with Timers\texorpdfstring{\thanks{
    Ga\"etan Staquet was a research fellow of the F.R.S.-FNRS\@.
    The research of B. Garhewal, F. Vaandrager, and
    G.A. P\'erez, was supported by NWO projects 612.001.852 ``GIRLS'' and OCENW.M.23.155 ``EVI'', and FWO project G0AH524N ``SynthEx''.}}{}}

\author{V\'eronique Bruy\`ere\inst{1}\orcidID{0000-0002-9680-9140}\and Bharat Garhewal\inst{2}\orcidID{0000-0003-4908-2863}\and Guillermo A. P\'erez\inst{3}\orcidID{0000-0002-1200-4952}\and Ga\"etan Staquet\inst{4}\orcidID{0000-0001-5795-3265}\and Frits W. Vaandrager\inst{2}\orcidID{0000-0003-3955-1910}}
\authorrunning{V. Bruy\`ere et al.}
\institute{University of Mons, Belgium\\
  \email{veronique.bruyere@umons.ac.be}
  \and Radboud University, The Netherlands\\
  \email{\{b.garhewal,f.vaandrager\}@cs.ru.nl}
  \and University of Antwerp -- Flanders Make, Belgium\\
  \email{guillermo.perez@uantwerpen.be}\and Centre Inria de l'Université de Rennes, France\\
  \email{gaetan.staquet@inria.fr}
}

\maketitle

\begin{abstract}
  We present the first algorithm for query learning 
  Mealy machines
  with timers in a black-box context.
  Our algorithm is an extension of the \lsharp algorithm of Vaandrager et
  al.\ to a timed setting.
  We rely on symbolic queries which empower us to reason on untimed
  executions while learning.
  Similarly to the algorithm for learning timed automata of Waga,
  these symbolic queries can be realized using finitely many concrete queries.
  Experiments with a prototype implementation show that our
  algorithm is able to efficiently learn realistic benchmarks.
\end{abstract}

\section{Introduction}
To understand and verify complex systems, we need accurate models that are  understandable for humans or can be analyzed fully automatically. Such models are typically not available for legacy software and for AI systems constructed from training data.
Model learning is a technology that potentially may fill this gap.
In this work, we consider a specific kind of model learning:
\emph{active automata learning}, which
is a black-box technique for constructing state machine models of software and hardware from information obtained through testing (i.e., providing inputs and observing the resulting outputs).
It has been successfully used in many applications, e.g.,
for spotting bugs in implementations of major network
protocols~\cite{RuiterP15,FiterauBrosteanJV16,FiterauBrosteanLPRVV17,FiterauBrosteanH17,FiterauBrosteanJMRSS20,FerreiraBDS21}.
We refer to~\cite{Vaandrager17,HowarS18} for surveys and further references.

Timing plays a crucial role in many applications.
However, extending model learning algorithms to a setting that incorporates
quantitative timing information turns out to be challenging.
Twenty years ago, the first papers on this subject were
published~\cite{GrinchteinJL04,MalerP04}, but we still do not have scalable
algorithms for a general class of timed models.
Consequently, in applications of model learning technology, timing issues still
need to be artificially suppressed.

Several authors have proposed active learning algorithms for the popular framework of \emph{timed automata} (TAs)~\cite{AlurD94}, which extends DFAs with clock variables.  Some of these proposals, for instance~\cite{GrinchteinJP06,GrinchteinJL10,HenryJM20} have never been implemented.
In recent years, however, several algorithms have been proposed and implemented that successfully learned realistic benchmark models.
A first line of work restricts to subclasses of TAs such as deterministic one-clock TAs~\cite{AnCZZZ20,XuAZ22}.
A second line of work explores synergies between active and passive learning algorithms.
Aichernig et al.~\cite{TapplerALL19,AichernigPT20}, for instance, employ a passive learning algorithm based on genetic programming to generate hypothesis models, which are subsequently refined using equivalence queries.
A major result was obtained recently by Waga~\cite{Waga23}, who presents an algorithm for active learning of (general) deterministic TAs and shows the effectiveness of the algorithm on various benchmarks.  Waga's algorithm is inspired by ideas of Maler \& Pnueli~\cite{MalerP04}. Based on the notion of elementary languages of~\cite{MalerP04},
Waga uses \emph{symbolic queries}
which are then performed using finitely many concrete queries. Notably, symbolic queries are also used for learning other families of automata,
such as register automata~\cite{GarhewalVHSLS20}.

A challenge for learning algorithms for TAs is the inference of the guards and resets that label transitions. Given those difficulties, Vaandrager
et al.~\cite{VaandragerBE21} propose to consider learning
algorithms for models using \emph{timers} instead of clocks, e.g., the class of
models defined by Dill~\cite{Dill89}. The value of a timer decreases when time advances, whereas the value of a clock increases. A timer can be set to integer values on transitions and may be stopped or times out (its value becomes 0) in later transitions. Each timeout triggers an observable output, allowing a learner to observe the occurrence of timeouts. In~\cite{VaandragerBE21}, the notion of Mealy machine with a \emph{single timer} (MM1T) is introduced, such that the absence of guards and invariants simplifies learning. A learner still has to determine which transitions (re)start the timer, but this no longer creates a combinatorial blow-up. If a transition sets a timer, then slight changes in the timing of this transition will trigger the corresponding changes in the timing of the resulting timeout, allowing a learner to identify the exact cause of each timeout.

Even if many realistic systems can be modeled as MM1Ts (e.g., the benchmarks in~\cite{VaandragerBE21} and the brick sorter and traffic controller examples in~\cite{DierlHKKLLM23}), the restriction to a single timer is a serious limitation. 
Kogel et al.~\cite{KogelKG23} propose Mealy machines with \emph{local timers} (MMLTs), where multiple timers are subject to carefully chosen constraints to enable efficient learning. Although quite interesting, the constraints of MMLTs are too restrictive for many applications (e.g., the FDDI protocol described in \Cref{appendix:fddi_model}). Also, any MMLT can be converted to an equivalent MM1T.
In~\cite{BruyerePSV23}, we explore a general extension of MM1Ts with \emph{multiple timers} (\MMTs), and show that for \MMTs that are \emph{race-avoiding}, the cause of a timeout event can be efficiently determined by ``wiggling'' the timing of input events.
Finally, \MMTs form a subclass of TAs, and thus, existing model checking algorithms and tools for TAs (such as UPPAAL\footnote{\url{https://uppaal.org/}}) can be used.

Our main contribution
is a learning algorithm for the \MMT model of~\cite{BruyerePSV23}, that is obtained by extending the \lsharp\ learning algorithm for Mealy machines of Vaandrager et al.~\cite{VaandragerGRW22} and by using the concept of symbolic queries as in~\cite{Waga23}. Extending the \lsharp algorithm to multiple timers is not an easy task; it requires several new ideas. As this approach is different from Angluin's $L^*$  framework~\cite{Angluin87}, used in other timed model learning algorithms, we cannot immediately apply existing ideas. Moreover, to the best of our knowledge, there is no notion of minimal \MMT, in contrast to the minimal representations of timed languages proposed in~\cite{BojanczykL12} and the minimal model for timed automata with integer resets of~\cite{DoveriGS24}. Experiments with a prototype implementation, written in Rust, show that our algorithm is able to efficiently learn realistic benchmarks.

\section{Mealy Machines with Timers}\label{sec:def}
A \emph{Mealy machine} is a variant of the classical finite automaton that associates an output with each transition.
It can be seen as a relation between input words and output words.
Mealy machines with
\emph{timers}~\cite{VaandragerBE21} can be used
to enforce timing constraints over the behavior of the relation, e.g., if we
send a message and do not receive the acknowledgment after \(d\) units of
time, we resend the message.

We fix two non-empty finite sets \(I\) of \emph{inputs} and \(O\) of \emph{outputs}.
A \emph{Mealy machine with timers} \(\M\) uses a finite set $X$ of \emph{timers} that can trigger transitions when they time out. Intuitively, changes in the state of the machine will be driven both by reading inputs and by ``timeouts''. To formalize this, it is convenient to define the set
\(\toevents{X} = \{\timeout{x} \mid x \in X\}\) of \emph{timeouts of \(X\)}.
We
write \(\actions{\M}\) for the set \(I \cup \toevents{X}\) of \emph{actions of
\(\M\)}:
reading an input 
or processing a timeout. 
The set \(\updates{\M} = (X \times \natplus) \cup \{\bot\}\) contains the
\emph{updates of \(\M\)}, where $(x, c)$ means that timer $x$ is started with
value $c$, and $\bot$ stands for no timer update.
We impose constraints on the shape of \MMTs, e.g.,
a timer \(x\) can time out only when it is active, to define
the timed semantics in a straightforward manner.

\begin{definition}[Mealy machine with timers]\label{def:MMT}
  A \emph{Mealy machine with timers} (MMT, for short) is a tuple
  \(\M = (X, Q, q_0, \activeTimers, \delta)\) where:
  \begin{itemize}
    \item
    \(X\) is a finite set of \emph{timers},
    \item
    \(Q\) is a finite set of states, with \(q_0 \in Q\) the \emph{initial state},
    \item
    \(\activeTimers : Q \to \subsets{X}\) assigns a
    set of \emph{active} timers to each state,
    and
    \item
    \(\delta : Q \times \actions{\M} \partto Q \times O \times \updates{\M}\)
    is a \emph{partial} transition function.
  \end{itemize}
  We write \(q \myxrightarrow[u]{i/o} q'\) if \(\delta(q, i) = (q', o, u)\).
  We require the following:
  \begin{enumerate}
    \item
    In the initial state, no timer is active, i.e.,
    \(\activeTimers(q_0) = \emptyset\).
    \item
    All active timers of the target state of a transition
    come from the source state, except at
    most one timer that may be started on the transition.
    That is,
    if \(q \myxrightarrow[\bot]{i/o} q'\), then
    \(\activeTimers(q') \subseteq \activeTimers(q)\), and if
    \(q \myxrightarrow[(x, c)]{i/o} q'\), then
    \(\activeTimers(q') \setminus \{x\} \subseteq \activeTimers(q)\).
    \item
    If timer $x$ times out, then $x$ was active in the source state, i.e.,
    if \(q \myxrightarrow[u]{\timeout{x}/o} q'\) then
    $x \in \activeTimers(q)$.
    Moreover, if \(u \neq \bot\), then \(u\) must be \((x, c)\) for some
    \(c \in \natplus\).
  \end{enumerate}
\end{definition}
When needed, we add a superscript to indicate which \MMT is
considered, e.g., \(Q^\M, q_0^\M\), etc.
Missing symbols in \(q \myxrightarrow[u]{i/o} q'\) are quantified existentially.
We say
a transition \(q \myxrightarrow[u]{} q'\) \emph{starts}
(resp.\ \emph{restarts}) a timer $x$ if $u = (x,c)$ and $x$ is inactive
(resp.\ active) in $q$.
We say that a transition\footnote{Technically, it is the arrival at a state where the timer is marked as inactive that stops the timer. However, we find it convenient to reason about the transitions leading to such a state as being responsible for stopping the timer. Also note that we could adapt our definition of MMTs so that transitions determine which, if any, timer must be stopped and this would be expressively equivalent: a \enquote{state-based} MMT can be obtained by constructing the Cartesian product of the state-space of the \enquote{transition-based} MMT and the set of possible \(\activeTimers\) functions.}
\(q \myxrightarrow{i} q'\) with
$i \ne \timeout{x}$ \emph{stops} timer \(x\) if $x$ is inactive in $q'$.
The requirement that a \(\timeout{x}\)-transition can only restart timer \(x\)
is without loss of generality and a choice made to simplify the presentation
(see \Cref{app:learning:hypo}).
An example of an \MMT is given in \Cref{fig:ex:MMT:good}.

A \emph{run} $\pi$ of $\M$ either consists of a single state $p_0$ or of a nonempty sequence of transitions
\(\pi ~=~  p_0 \myxrightarrow[u_1]{i_1/o_1} p_1 \myxrightarrow[u_2]{i_2/o_2} \dotsb
\myxrightarrow[u_n]{i_n/o_n} p_n\).
We denote by \(\runs{\M}\) the set of runs of \(\M\).
We often write \(q \myxrightarrow{i} {} \in \runs{\M}\) to highlight that \(\delta(q, i)\) is defined.
We lift the notation to words \(w = i_1 \dotsb i_n\) as usual:
\(p_0 \myxrightarrow{w} p_n \in \runs{\M}\) if there exists a run
\(p_0 \myxrightarrow{i_1} \dotsb \myxrightarrow{i_n} p_n \in \runs{\M}\).
Note that any run $\pi$ is uniquely determined by state $p_0$ and word $w$,
as \(\M\) is deterministic.
A run \(p_0 \myxrightarrow{w}\) is
\emph{\(x\)-spanning} (with \(x \in X\)) if it begins with a transition
(re)starting \(x\),
ends with a \(\timeout{x}\)-transition, and no
intermediate transition restarts or stops \(x\).

\begin{figure}[t]
  \centering
  \begin{tikzpicture}[
  automaton,
]
  \node [state, initial]          (q0)  {\(q_0\)};
  \node [state, right=45pt of q0] (q1)  {\(q_1\)};
  \node [state, right=45pt of q1] (q2)  {\(q_2\)};
  \node [state, right=70pt of q2] (q3)  {\(q_3\)};
  \node [state, right=80pt of q3] (q4)  {\(q_4\)};
  \node [state, below=45pt of q4] (q5) {\(q_5\)};

  \path
    (q0)  edge                  node [above=-2pt] {\(i/o\)}
                                node [below=-2pt] {\(\updateFig{x}{2}\)} (q1)
    (q1)  edge [loop below, min distance=4mm]
                                node [below=-2pt]
                                    {\(\timeout{x}/o, \updateFig{x}{2}\)} (q1)
          edge                  node [above=-2pt] {\(i/o'\)}
                                node [below=-2pt] {\(\updateFig{y}{3}\)} (q2)
    (q2)  edge                  node [above=-2pt] {\(\timeout{x}/o,
                                              \updateFig{x}{2}\)} (q3)
          edge [bend right=15]  node [below=-2pt] {\(i/o', \updateFig{x}{2}\)} (q3)
    (q3)  edge                  node [above=-2pt] {\(i/o', \updateFig{x}{2}\)} (q4)
    (q4)  edge                  node [near end, right=-2pt] {\(\timeout{x}/o,
                                                \updateFig{x}{2}\)} (q5)
          edge [loop right]     node {\(i/o', \updateFig{x}{2}\)} (q4)
    (q5)  edge [loop right]     node {\(i/o', \updateFig{x}{2}\)} (q5)
  ;

  \draw [rounded corners = 10]
    let
      \p{0} = (q0.-40),
      \p{3} = (q3)
    in
      (q3.south)
        -- (\x{3}, -1.2)
        -- (\x{0}, -1.2)
        node [above=-2pt, pos=0.2] {\(\timeout{y}/o, \bot\)} 
        -- (\p{0})
  ;

  \draw [rounded corners = 10]
    let
      \p{0} = (q0.-140),
      \p{4} = (q4.-45),
      \p{5} = (q5),
    in
      (\p{4})
        -- ($(\p{4}) + (0, -0.5)$)
        -- ($(\p{4}) + (2.5, -0.5)$)
        node [above=-2pt, midway] {\(\timeout{y}/o, \bot\)} 
        -- ($(\x{4}, \y{5}) + (2.5, -0.35)$)
        -- ($(\p{5}) + (0, -0.35)$)
        -- ($(\x{0}, \y{5}) + (0, -0.35)$)
        -- (\p{0})
  ;

  \draw [rounded corners = 10]
    let
      \p{0} = (q0.-90),
      \p{1} = (q1),
      \p{5} = (q5),
    in
      (q5.west)
        -- (\x{1}, \y{5})
        node [pos=0.17, above=-2pt] {\(\timeout{y}/o, \bot\)}
        -- (\x{0}, \y{5})
        -- (\p{0})
  ;

  \draw [rounded corners = 10]
    let
      \p{s} = (q3.-45),
      \p{t} = (q5.135),
      \p{m1} = ($(\p{s}) + (0, -0.35)$),
      \p{m2} = ($(\x{t}, \y{s}) + (0, -0.35)$),
    in
      (\p{s}) -- (\p{m1})
              -- (\p{m2})
              node [midway, above=-2pt] {\(\timeout{x}/o, \updateFig{x}{2}\)}
              -- (\p{t})
  ;
\end{tikzpicture}
  \caption{An \MMT with $\activeTimers(q_0) = \emptyset$,
  $\activeTimers(q_1) = \{x\}$, and $\activeTimers(q) = \{x,y\}$ for all other $q$.}\label{fig:ex:MMT:good}
\end{figure}

\subsection{Timed Semantics}\label{subsec:semantics}

The semantics of an \MMT $\M$ is defined via an infinite-state labeled
transition system describing all possible configurations and transitions
between them.
A \emph{valuation} is a partial function $\valuation \colon X \partto \nnr$ that
assigns nonnegative real numbers to timers.
For $Y \subseteq X$, we write $\Val{Y}$ for the set of all valuations $\valuation$
with $\dom{\valuation} = Y$.
A \emph{configuration} of $\M$ is a pair $(q, \valuation)$ where $q \in Q$ and
$\valuation \in \Val{\chi(q)}$.
The \emph{initial configuration} is the pair $(q_0, \valuation_0)$ where
$\valuation_0 = \emptyset$ since $\chi(q_0) = \emptyset$.
If $\valuation \in \Val{Y}$ is a valuation in which all timers from $Y$ have a
value of at least $d \in \nnr$, then $d$ units of time may \emph{elapse}.
We write $\valuation - d \in\Val{Y}$ for the resulting valuation that satisfies
$(\valuation -d)(x) = \valuation(x) -d$, for all $x \in Y$. 
If the valuation $\valuation$ contains a value $\valuation(x) = 0$ for some
timer $x$, then $x$ may \emph{time out}.
We define the transitions between configurations
\((q, \valuation), (q', \valuation')\) as follows:
\begin{description}
  \item[delay transition]
  \((q, \valuation) \myxrightarrow{d} (q, \valuation - d)\), with
  \(\valuation(x) \geq d\) for every \(x \in \activeTimers(q)\),
  \item[discrete transition]
  \((q, \valuation) \myxrightarrow[u]{i/o} (q', \valuation')\), with
  \(q \myxrightarrow[u]{i/o} q'\) a transition, \(\valuation'(y) = \valuation(y)\)
  for all \(y \in \activeTimers(q')\) except that
  \(\valuation'(x) = c\) 
  if \(u = (x, c)\).
  Moreover, if \(i = \timeout{x}\) then
  \(\valuation(x)=0\)
  and we call it a \emph{timeout transition};
  otherwise, we call it an \emph{input transition}.
\end{description}

A \emph{timed run} of $\M$ is a sequence of configuration transitions such that
delay and discrete transitions alternate, beginning and ending with a delay
transition.
We say a configuration $(q,\kappa)$ is \emph{reachable} if there is a timed run
$\rho$ that starts with the initial configuration and ends with $(q,\kappa)$.
The \emph{untimed projection} of \(\rho\), noted
\(\untimeRun{\rho}\), is the run obtained by omitting the valuations and
delay transitions of \(\rho\).
A run \(\pi\) is said \emph{feasible} if there is a timed
run \(\rho\) such that \(\untimeRun{\rho} = \pi\).

A \emph{timed word} over a set \(\Sigma\) 
is an alternating sequence of delays from $\nnr$
and symbols from \(\Sigma\), 
 such that it starts and ends with a delay.
The length of a timed word \(w\), noted \(\lengthOf{w}\), is the number of
symbols of $\Sigma$ in \(w\).
Note that, when \(\Sigma = \actions{\M}\), a timed run reading a timed word \(w\)
is uniquely determined by its first configuration and \(w\).
We thus write $(p,\kappa) \myxrightarrow{w}$ for a
timed run.
A timed run \(\rho\) is called \emph{\(x\)-spanning} (with \(x \in X\)) if
\(\untimeRun{\rho}\) is \(x\)-spanning.

\begin{example}
  Let \(\M\) be the \MMT of \Cref{fig:ex:MMT:good}.
  The timed run reading the timed word \(0.5 \cdot i \cdot 1 \cdot i
  \cdot 1 \cdot \timeout{x} \cdot 2 \cdot \timeout{y} \cdot 0\) and its
  untimed projection are:
  \begin{align*}
    (q_0, \emptyset)
    & {\myxrightarrow{0.5}} (q_0, \emptyset)
    {\myxrightarrow[(x, 2)]{i/o}} (q_1, x {=} 2)
    {\myxrightarrow{1}} (q_1, x {=} 1)
    {\myxrightarrow[(y, 3)]{i/o'}} (q_2, x {=} 1, y {=} 3)
{\myxrightarrow{1}} (q_2, x {=} 0, y {=} 2)\\
    &\myxrightarrow[(x, 2)]{\timeout{x}/o} (q_3, x {=} y {=} 2)
    \myxrightarrow{2} (q_3, x {=} y {=} 0)
    \myxrightarrow[\bot]{\timeout{y}/o} (q_0, \emptyset)
    \myxrightarrow{0} (q_0, \emptyset),
    \\
    \pi ={} & q_0
    \myxrightarrow[(x, 2)]{i/o} q_1
    \myxrightarrow[(y, 3)]{i/o'} q_2
    \myxrightarrow[(x, 2)]{\timeout{x}/o} q_3
    \myxrightarrow[\bot]{\timeout{y}/o} q_0.
  \end{align*}
  Hence, \(\pi\) is feasible,
  unlike the run \(q_0 \myxrightarrow{i \cdot i \cdot \timeout{y}}\), as the value
  of \(y\) in \(q_2\) is always greater than the value of \(x\), no matter the
  chosen delays.
  Observe that \(q_1 \myxrightarrow{i \cdot \timeout{x} \cdot \timeout{y}}\) is
  \(y\)-spanning,
  while \(q_0 \myxrightarrow{i \cdot i \cdot \timeout{x}}\) is \(x\)-spanning.

  As \((q_2, x = 0, y = 2)\) is reachable, we say that
  \(x\) is \emph{enabled} in \(q_2\), i.e., it is possible to observe the timeout
  of \(x\) in \(q_2\) along some timed run.
  However, \(y\) is not enabled in \(q_2\), as it is impossible to reach a
  configuration \((q_2, x = \cdot, y = 0)\).
\end{example}
Formally, we write \(\enabled{q}\) for the set of all enabled timers of \(q\),
i.e.,
\[
  \enabled{q} = \{x \in \activeTimers(q) \mid \exists (q_0, \emptyset)
  \myxrightarrow{w} (q, \valuation) : \valuation(x) = 0 \}.
\]
If \(q\) has at least one enabled timer then, just by waiting in \(q\) for long
enough, we can force one timer to reach the value zero.
A desirable property is for all such behaviors to have a corresponding timeout
transition, which is the case for \(\M\).
We thus say that an \MMT \(\N\) is \emph{complete} if for all
\(q \in Q\) and all \(i \in I \cup \toevents{\enabled{q}}\) we have \(q \myxrightarrow{i} {} \in \runs{\N}\).\footnote{Notice that we only require the timeout transitions to be defined for enabled timers.}

\subsection{Symbolic Semantics}\label{sec:MMT:equivalence}

Our learning algorithm is based on an untimed, symbolic semantics that abstracts from timing delays and timer names.
The idea is that, rather than the name of a timer in a $\timeout{x}$ event, we record a pointer to the preceding transition that (re)started this timer, together with the value to which the timer was set in this preceding transition.
Consider, for instance, the run
\[
\pi  = q_0 \myxrightarrow[(x, 2)]{i} q_1
\myxrightarrow[(x, 2)]{\timeout{x}} q_1
\myxrightarrow{\timeout{x}} q_1
\]
of the \MMT of \Cref{fig:ex:MMT:good}.  
Then the \emph{symbolic input word} of $\pi$ is $\symbolic{w} = i \cdot \timeout{2, 1} \cdot \timeout{2, 2}$, as 
the first action of \(\pi\) is the input \(i\), 
the second action is a timeout that is caused by setting the timer to $2$ in the first transition, and 
the third action is a timeout that is caused by setting the timer to $2$ in the second transition.
More generally, consider a run
\[
	\pi =  p_0 \myxrightarrow[u_1]{i_1/o_1} p_1 \myxrightarrow[u_2]{i_2/o_2} \dotsb
	\myxrightarrow[u_n]{i_n/o_n} p_n 
\]
with $p_0 = q_0^{\M}$.  For \(k \in \{1, \dotsc, n\}\), let $\pi_k$ be the prefix of run $\pi$ up to state $p_k$.
Suppose that timer $x$ is active in the last state of $\pi_k$, for some $k>0$. Then we define
$\cause{\pi_k}{x}$  as the pair $(c, j)$, where $j$ is the index of the last transition (re)starting \(x\), and $c$ is the value to which $x$ is set in this transition:
\[
	\cause{\pi_k}{x} = \begin{cases}
		(c, k) 					& \mbox{if } u_k=(x,c),\\
		\cause{\pi_{k-1}}{x}    & \mbox{otherwise.}
	\end{cases}
\]
Now we associate with $\pi$ a \emph{symbolic word} (sw, in short)
$\sinpw{\pi} =\symbolic{i_1} \dotsb \symbolic{i_n}$ over alphabet \(\symbActions = I \cup \toevents{\natplus\times\natplus}\) as follows.
For every \(k \in \{1, \dotsc, n\}\),
if \(i_k \in I\) then \(\symbolic{i_k} = i_k\), and 
if \(i_k = \timeout{x}\) then
\(\symbolic{i_k} = \timeout{\cause{\pi_{k-1}}{x}}\).
Similarly, we associate an \emph{output word} $\outw{\pi}$ with $\pi$, consisting of the sequence of output actions $o_1 \cdots o_n$ occurring in run $\pi$.
We write $\silanguage{\M}$ to denote the \emph{symbolic input language} of $\M$, defined as the set of symbolic input words associated with feasible runs of $\M$ starting from $q_0^{\M}$.

Notice that, since MMTs are deterministic, a feasible run \(\pi\) can be
denoted via the word \(w = i_1 \dotsb i_n\), and vice-versa.
Thus, for every word \(w\) labeling a feasible run \(\pi\), we also associate a
sw to \(w\) as \(\toSymbolic{w} = \sinpw{\pi}\).
Moreover, since each transition of an MMT (re)starts at most one timer,
at most one run can be associated with any symbolic input word.
For instance, the symbolic input word
  \(\symbolic{w} = i \cdot i \cdot \timeout{2, 1} \cdot \timeout{3, 2}\) induces
  the run \(q_0 \myxrightarrow[(x, 2)]{i} q_1 \myxrightarrow[(y, 3)]{i} q_2
  \myxrightarrow{\timeout{x}} q_3 \myxrightarrow{\timeout{y}} q_0\) in \(\M\).
If $\symbolic{w} \in\silanguage{\M}$, we write $\ofunction{\M}{\symbolic{w}}$ for the output word $\outw{\pi}$ of the unique run $\pi$ of $\M$ with $\sinpw{\pi} = \symbolic{w}$.

We now define a notion of \emph{symbolic equivalence} between two \MMTs.
\begin{definition}[Symbolic equivalence]\label{def:equivalence:symbolic}
  Two \complete \MMTs \(\M\) and \(\N\) are \emph{symbolically equivalent}, noted \(\M \symEquivalent \N\), 
  if $\silanguage{\M} = \silanguage{\N}$ and, for each $\symbolic{w} \in \silanguage{\M}$, $\ofunction{\M}{\w} = \ofunction{\N}{\w}$.
\end{definition}
As outputs and updates at the start of spanning runs are the same,
symbolic equivalence implies \emph{timed equivalence}, i.e., the MMTs produce the same timed output words on all timed input words
(see \Cref{app:equivalence}).

\subsection{Learning framework}

As usual for active learning algorithms, we rely on
Angluin's framework~\cite{Angluin87}:
A teacher knows an \MMT \(\M\),
and a learner can query the teacher to obtain
knowledge about \(\M\).
The set of \MMTs $\M$ that we consider for our learning algorithm must be \emph{\good} (the \emph{s} stands for symbolically), i.e., 
$\M$ is \complete and every run of \(\M\) is feasible.
The \MMT of \Cref{fig:ex:MMT:good} is \good.
From any \complete \MMT,
one can construct an \good \MMT  that is symbolically equivalent, by using
\emph{zones}
(akin to the homonymous concept for
timed automata, see~\cite{HandbookModelChecking}).
See \Cref{proof:lemma:gootMMT} for details.

\begin{restatable}{lemma}{goodMMT}\label{lemma:goodMMT}
  For any \complete \MMT \(\M\), there is an \good \MMT
  \(\N \symEquivalent \M\).
\end{restatable}

We now define the \emph{queries} the learner uses to gather knowledge about \(\M\),
the \MMT of the teacher.
For classical Mealy
machines~\cite{VaandragerGRW22,ShahbazG09},
there are two queries: \emph{output queries} providing the sequence of
outputs for a given input word, and \emph{equivalence queries} asking whether
a hypothesis \(\hypothesis\) is correct.
If it is not, a counterexample is returned, i.e., a word \(w\) inducing different
outputs in \(\hypothesis\) and in \(\M\).
In this work, we need to adapt those queries to encode the timed behavior
induced by the timers of \(\M\).
As two \MMTs do not use the same timers in general, we rely on \siws,
and adapt our 
queries accordingly.
In order to deal with timed behavior, we also need a new type of query, 
called a \emph{wait query}.
Kogel et al.~introduced a similar type of query, called timer query, in~\cite{KogelKG23}.

\begin{definition}[Symbolic queries]\label{def:queries:symbolic}
  The learner uses three symbolic queries, parametrized by either a symbolic input word $\symbolic{w}$ or a complete MMT $\hypothesis$:
  \begin{description}
    \item[\(\symOutputQ(\symbolic{w})\)] If $\symbolic{w} \in\silanguage{\M}$, then this query returns $\ofunction{\M}{\symbolic{w}}$.
    \item[\(\symWaitQ(\symbolic{w})\)]
   If $\w \in\silanguage{\M}$, then this query returns the set of all pairs \((c, j)\) with $\w ~\timeout{c, j} \in\silanguage{\M}$.
    \item[\(\symEquivQ(\hypothesis)\)]
    This query returns \yes if \(\hypothesis \symEquivalent \M\),
    and otherwise a symbolic input word $\w$ witnessing the non-equivalence: either 
    $\w \in\silanguage{\hypothesis}\setminus\silanguage{\M}$, or
    $\w \in\silanguage{\M}\setminus\silanguage{\hypothesis}$, or
    $\w \in\silanguage{\hypothesis}\cap\silanguage{\M}$ and
    $\ofunction{\hypothesis}{\w} \neq \ofunction{\M}{\w}$.
  \end{description}
\end{definition}

$\symOutputQ$ and $\symEquivQ$ are analogous to regular output and equivalence queries 
for Mealy machines, while $\symWaitQ$ provides, 
for each timer $x$ enabled at the end of the run induced by the symbolic word, 
the transition which last (re)started $x$ and the constant $c$ with which $x$ was (re)started.
(See \Cref{ex:learning:use_queries} below.)

Similarly to the TA learning algorithm of Waga~\cite{Waga23}, these symbolic queries can be performed via
concrete output and equivalence queries. This is possible under the assumption that \(\M\) is 
\emph{race-avoiding}~\cite{BruyerePSV23}, i.e., $\M$
allows runs to be observed deterministically, in the sense that
any feasible run \(\pi\) is the untimed projection of a run \(\rho\) where
all delays are non-zero
and there are no two timers that time out at the same time along \(\rho\).\footnote{We highlight that we do not require this to hold for any timed run
\(\rho\). Importantly, if \(\M\) is race-avoiding, then we can always change the
delays in the runs to avoid having concurrent timeouts.}
Not all \MMTs are race-avoiding and a decidable characterization of those \MMTs that are is given in~\cite{BruyerePSV23}.
\MMTs with one timer~\cite{VaandragerBE21} and Mealy machines with local timers~\cite{KogelKG23} (translated as \MMTs) are all race-avoiding, this is also the case for all benchmarks in \Cref{sec:implementation}.

\begin{restatable}{lemma}{symbolicQueries}\label{lemma:queries}
  For \robust MMTs, the three symbolic queries can be realized
  via a polynomial number of concrete output and equivalence queries.
\end{restatable}

For intuition, for all symbolic queries, we can construct a system of linear temporal constraints whose solutions can be mapped to concrete input words. Technically, the system also depends on the tree we use to structure the result of queries we have made. (See \Cref{proof:lemma:queries} for a proof.)

\section{Learning algorithm}\label{sec:learning}

We now present our learning algorithm for \MMTs, called \lsharpMMT.
Given the hidden \good \MMT \(\M\) we want to learn, this algorithm gradually builds a tree-shaped \MMT $\tree$ that stores the observations obtained by membership and wait queries. It then constructs a hypothesis \MMT $\hypothesis$ obtained from $\tree$ whose certain states (the basis) are the states of $\hypothesis$ and certain transitions are redirected to the basis. An example of \MMTs $\tree$ and $\hypothesis$ is given in \Cref{fig:ex:counterexample} for the \MMT $\M$ of \Cref{fig:ex:MMT:good}, such that the equivalence query fails. The main difficulties of the learning process are: \emph{(i)} $\tree$ uses its own timers without knowing those of $\M$ and for which it is necessary to discover when two distinct timers of $\tree$ represent the same timer of $\M$ (see the concept of timer matching and functional simulation below); \emph{(ii)} the basis of $\tree$ used to define the states of $\hypothesis$ must be formed of states that correspond to distinct states of $\M$ despite partial knowledge of its timers (see the concept of apartness and \Cref{thm:extension-n-soundness} below).
Before moving on to the details of the \lsharpMMT algorithm, let us state its complexity 
(see \Cref{app:learning:termination} for a proof).
Observe that the same bounds hold with concrete queries, by \Cref{lemma:queries}, and that the complexity becomes polynomial when \(\lengthOf{X^\M}\) is fixed.

\begin{figure}[t]
  \centering
  \begin{tikzpicture}[
  automaton,
]
  \node [state, initial, basis]                   (t0)  {\treeNodeLabel{\(t_0\)}};
  \node [state, right=40pt of t0, basis]          (t1)  {\treeNodeLabel{\(t_1\)}};
  \node [state, above right=7pt and 78pt of t1]   (t2)  {\(t_2\)};
  \node [state, below right=7pt and 78pt of t1]   (t3)  {\(t_3\)};
  \node [state, right=58pt of t2]                 (t4)  {\(t_4\)};
  \node [state, right=58pt of t3]                 (t5)  {\(t_5\)};

  \path
    (t0)  edge                  node [above] {\(i/o\)}
                                node [below] {\(\updateFig{x_1}{2}\)}         (t1)
    (t1)  edge                  node [below=-2pt] {\(i/o', \bot\)}            (t3)
          edge                  node [sloped] {\(\timeout{x_1}/o,
                                              \updateFig{x_1}{2}\)}           (t2)
    (t2)  edge                  node [above=-2pt] {\(\timeout{x_1}/o, \bot\)} (t4)
    (t3)  edge                  node [above=-2pt] {\(\timeout{x_1}/o, \bot\)} (t5)
  ;
\end{tikzpicture}
\begin{tikzpicture}[
  automaton,
  every loop/.append style = {
    min distance = 1mm,
  },
]
  \node [state, initial]          (t0)  {\(t_0\)};
  \node [state, right=55pt of t0] (t1)  {\(t_1\)};

  \path
    (t0)  edge              node [above=-1pt] {\(i/o\)}
                            node [below]
                              {\(\updateFig{y_1}{2}\)}  (t1)
    (t1)  edge [loop right] node {\(i/o', \bot\)}                         (t1)
          edge [loop above] node [above=-1pt] {\(\timeout{y_1}/o,
                                  \updateFig{y_1}{2}\)} (t1)
  ;
\end{tikzpicture}
  \caption{On the left, a tree $\tree$ from which the hypothesis \MMT on the
  right is constructed.
  Basis states of $\tree$ have a gray background.}\label{fig:ex:counterexample}
\end{figure}

\begin{restatable}{theorem}{termination}\label{thm:learning:termination}
  The \lsharpMMT algorithm terminates
  and returns an \MMT \(\N \symEquivalent \M\) s.t. $\N$ has a size
  polynomial in \(\lengthOf{Q^\M}\) and factorial in \(\lengthOf{X^\M}\). 
    Its running time and number of symbolic queries of \lsharpMMT
  are polynomial in \(\lengthOf{Q^\M}, \lengthOf{I}\) and the length of the
  longest counterexample returned by the teacher, factorial in \(\lengthOf{X^\M}\).
\end{restatable}

\subsection{Observation tree}\label{sec:learning:tree}

We describe the main data structure of \lsharpMMT:
a modification of the \emph{observation tree} used for
\lsharp~\cite{VaandragerGRW22}.
We impose that this tree \(\tree\) is a tree-shaped \MMT
whose every run is feasible.
Each state \(q\) of \(\tree\) has its own timer \(x_q\) that can only be started
by the incoming transition of \(q\), and may be restarted only by a
\(\timeout{x_q}\)-transition.
Due to the tree-shape nature of $\tree$, we can impose strict constraints on the sets of
active and enabled timers of its states, as described in the next definition.
We only provide the main ideas, see \Cref{app:tree} for details.

\begin{definition}[Observation tree]\label{def:tree}
An MMT $\M$ is \emph{tree shaped} if for every state $q \in Q^{\M}$ there is a unique run, denoted $\run{\M}{q}$, that starts in $q^{\M}_0$ and ends in $q$.
  An \emph{observation tree} is a tree-shaped \MMT
  \(\tree = (X, Q, q_0, \activeTimers, \delta)\) such that
  \(X = \{x_q \mid q \in Q \setminus \{q_0\}\}\),
  every run of \(\tree\) is feasible,
  \begin{itemize}
    \item
    for all \(p \myxrightarrow[(x, c)]{i} q\) with \(i \in I\), we have \(x = x_q\),
    \item
    for all \(q \in Q, x \in X\), we have \(x \in \activeTimers(q)\) if and only if
    there is an \(x\)-spanning run traversing \(q\), and
    \item
    for all \(q \in Q, x \in X\), we have \(x \in \enabled{q}\) if and only if
    \(q \myxrightarrow{\timeout{x}} {} \in \runs{\tree}\).
  \end{itemize}
\end{definition}

\begin{figure}[t]
  \centering
  \begin{tikzpicture}[
  automaton,
]
  \node [state, initial, basis]                         (t0)  {\treeNodeLabel{\(t_0\)}};
  \node [state, right=58pt of t0, basis]                (t1)  {\treeNodeLabel{\(t_1\)}};
  \node [state, above right=8pt and 82pt of t1]         (t2)  {\(t_2\)};
  \node [state, below right=8pt and 82pt of t1, basis]  (t3)  {\treeNodeLabel{\(t_3\)}};
  \node [state, above right=8pt and 90pt of t2]         (t4)  {\(t_4\)};
  \node [state, above right=8pt and 90pt of t3]         (t5)  {\(t_5\)};
  \node [state, below right=8pt and 90pt of t3]         (t6)  {\(t_6\)};
  \node [state, right=70pt of t4]                       (t7)  {\(t_7\)};
  \node [state, right=70pt of t5]                       (t8)  {\(t_8\)};
  \node [state, right=70pt of t6]                       (t9)  {\(t_9\)};
  \node [state, below=15pt of t9]                       (t10) {\(t_{10}\)};

  \path
    (t0)  edge                  node {\(i/o, \updateFig{x_1}{2}\)}              (t1)
    (t1)  edge                  node [sloped, '] {\(i/o', \updateFig{x_3}{3}\)} (t3)
          edge                  node [sloped] {\(\timeout{x_1}/o,
                                                \updateFig{x_1}{2}\)}           (t2)
    (t2)  edge                  node [sloped, above=-2pt]
                                              {\(\timeout{x_1}/o, \bot\)}       (t4)
    (t3)  edge                  node [sloped, above=-2pt] {\(\timeout{x_1}/o,
                                                \updateFig{x_1}{2}\)}           (t5)
          edge                  node [sloped, below=-2pt]
                                                {\(i/o', \updateFig{x_6}{2}\)}  (t6)
    (t5)  edge                  node [above=-2pt] {\(\timeout{x_3}/o, \bot\)}   (t8)
    (t6)  edge                  node [above=-2pt] {\(\timeout{x_6}/o, \bot\)}   (t9)
  ;

  \draw [rounded corners = 10pt]
    let
      \p{s} = (t5.45),
      \p{t} = (t7.-180),
    in
      (\p{s}) -- ($(\x{s}, \y{t}) + (0.1, 0)$)
              -- (\p{t})
              node [above=-2pt, midway] {\(\timeout{x_1}/o, \bot\)}
  ;

  \draw [rounded corners = 10pt]
    let
      \p{s} = (t6.-43),
      \p{t} = (t10.180),
    in
      (\p{s}) -- ($(\x{s}, \y{t}) + (0.2, 0)$)
              -- (\p{t})
              node [above=-2pt, midway] {\(\timeout{x_3}/o, \bot\)}
  ;
\end{tikzpicture}
  \caption{Sample observation tree (we write \(x_i\) instead of \(x_{t_i}\) for all states $t_i$) with
  \(\activeTimers(t_1) = \activeTimers(t_2) = \{x_1\}\),
  \(\activeTimers(t_3) = \activeTimers(t_5) = \{x_1, x_3\}\),
  \(\activeTimers(t_6) = \{x_3, x_6\}\), and
  $\activeTimers(t) = \emptyset$ for the other states $t$.
  }\label{fig:ex:tree}
\end{figure}

We explain
how to use \symOutputQ and \symWaitQ to gradually grow \(\tree\) on an example.

\begin{example}\label{ex:learning:use_queries}
  Let \(\M\) be the \MMT of \Cref{fig:ex:MMT:good} and
  \(\tree\) be the observation tree of \Cref{fig:ex:tree}, except that
  \(t_3 \myxrightarrow{i} {} \notin \runs{\tree}\), i.e.,
  the subtree rooted at \(t_6\) is not present in the tree.
  We construct that subtree, via \symOutputQ and \symWaitQ.

  First, we create the \(t_3 \myxrightarrow{i}\) transition.
  Let \(w = i \cdot i\) be the unique word such that
  \(t_0 \myxrightarrow{w} t_3\) and \(\symbolic{w} = \toSymbolic{w} = i \cdot i\)
  be the corresponding \siw.
  As \(\M\) is \good (and, thus, \complete), it follows that
  \(q_0^\M \myxrightarrow{\symbolic{w} \cdot i} {} \in \runs{\M}\), i.e.,
  \(\symbolic{w} \cdot i \in\silanguage{\M}\). 
  So, we can call \(\symOutputQ(\symbolic{w} \cdot i)\), which returns
  \(o \cdot o' \cdot o'\).
  Hence, we can create the new transition \(t_3 \myxrightarrow[\bot]{i/o'} t_6\) in $\tree$.
  We initially use a \(\bot\) update (seen as a sort of wildcard) as, for now,
  we do not
  have any information about the potential update of the corresponding
  transition (here, \(q_2 \myxrightarrow[(x, 2)]{i} q_3\)) in \(\M\).
  Later on, we may discover that the transition must start a timer, in
  which case \(\bot\) will be replaced by an actual update.

  We then perform a symbolic wait query in \(t_6\), i.e., call
  \(\symWaitQ(\symbolic{w} \cdot i)\), which returns the set
  \(\{(2, 3), (3, 2)\}\) meaning that the second (resp.\ third)  transition of the run
  \(q_0^\tree \myxrightarrow{\symbolic{w} \cdot i}\) must (re)start a timer
  with constant 3 (resp.\ constant 2).
  So, the \(\bot\) of the newly created transition is replaced by \((x_6, 2)\)
  (as the label of the transition is an input).
  It remains to create the \(\timeout{x_3}\)- and \(\timeout{x_6}\)-transitions
  from $t_6$ by performing two symbolic output queries.
  We call   \(\symOutputQ(\symbolic{w} \cdot i \cdot \timeout{2})\) and \(\symOutputQ(\symbolic{w} \cdot i \cdot \timeout{3})\)
  and create the transitions.
  We obtain the tree of \Cref{fig:ex:tree}.
\end{example}

\subsubsection{Functional simulation.}
Each state of an observation tree can be \emph{mapped} to a state of the hidden
\(\M\), by adapting the notion of \emph{functional simulation} of
\lsharp~\cite{VaandragerGRW22} (see details in \Cref{app:tree:simulation}).
We have a function \(f : Q^\tree \to Q^\M\) that preserves the initial state and the behavior of transitions.
In addition, since \(\tree\) and \(\M\) may use \emph{different} timers,
we have a function \(g : X^\tree \partto X^\M\) that maps timers of \(\tree\) to timers of \(\M\).
We lift \(g\) to actions such that \(g(i) = i\) for every \(i \in I\) and
\(g(\timeout{x}) = \timeout{g(x)}\) for every \(x \in \dom{g}\).
We require, for all \(q, q' \in Q^\tree, x, y \in X^\tree, i \in I\) and $o \in O$:
\begin{enumerate}
  \item 
  $f$ preserves the initial state, i.e., \(f(q_0^\tree) = q_0^\M\);
  \item
  any timer that is active in \(q\) has a corresponding active timer in
  \(f(q)\), i.e.,
  if \(x \in \activeTimers^\tree(q)\), then \(g(x) \in \activeTimers^\M(f(q))\);
  \item
  two distinct timers that are active in \(q\) are mapped to two distinct timers in \(\M\), i.e.,
  if \(x \neq y\) and both are active in \(q\), then \(g(x) \neq g(y)\);
  \item
  for each \(q \myxrightarrow[]{i/o} q'\) in \(\tree\),
  there is a corresponding transition 
  \(f(q) \myxrightarrow{g(i)/o} f(q')\) in \(\M\);
  \item 
  using rule (4) we may construct, for each run $\pi$ of $\tree$, a corresponding run $\langle f, g \rangle(\pi)$ of $\M$.  We require that the causality of timeouts is preserved, i.e., for each \(q \myxrightarrow{\timeout{x}/o} q'\) in \(\tree\) with $q$ the last state of $\pi$,
  $\cause{\pi}{x} = \cause{\langle f,g \rangle(\pi)}{g(x)}$.
\end{enumerate}
One can show the following properties from the above constraints.

\begin{lemma}
  Let \(\M\) be an \good \MMT,
  \(\tree\) be an observation tree, and \(f\) and \(g\) be the functions
  described above.
  Then, for any state \(q \in Q^\tree\), we have
  \(\lengthOf{\activeTimers^\tree(q)} \leq \lengthOf{\activeTimers^\M(f(q))}\)
  and for all \(x \in \enabled{q}[\tree]\), it holds that
  \(g(x) \in \enabled{f(q)}[\M]\).
\end{lemma}

We say that a state \(q\) of \(\tree\) is \emph{explored} once \(\symWaitQ(q)\)
has been called.
By definition of wait queries, it means that every enabled timer
of \(f(q)\) is identified in \(q\), i.e.,
\(\lengthOf{\enabled{q}} = \lengthOf{\enabled{f(q)}}\).
We thus obtain a \emph{one-to-one correspondence} between \(\enabled{q}[\tree]\)
and \(\enabled{f(q)}[\M]\).
Define \(\explored\) as the maximal set of explored states of \(\tree\).
During the learning process, we will ensure that a wait query is performed only on states that are successors of an explored state to keep $\explored$ as a subtree of $\tree$.

\begin{example}
  In the tree of \Cref{fig:ex:tree}, the explored states are
  \(t_0\), \(t_1\), \(t_2\), \(t_3\), \(t_5\), and \(t_6\).
  In \Cref{ex:learning:use_queries}, \(t_0\) to \(t_3\), and \(t_5\)
  were already explored and formed a subtree.
  The wait query over \(t_6\) made it an explored state,
  i.e., \(\explored\) is still a subtree.
\end{example}

\subsubsection{Apartness.}
A key aspect of \lsharp for regular Mealy machines is the notion of \emph{apartness} indicating which states of an observation tree 
have different behaviors~\cite{VaandragerGRW22}: states \(p, p'\) are \emph{apart} if they
have different output responses to the same input word.
Here, we face a more complex situation as we need to handle the fact that different timers in \(\tree\) can represent the same timer in \(\M\). Before defining apartness in the \MMT  setting, we need to introduce two concepts.

For two states $p,p'$ of $\tree$, we define a \emph{matching} $m : \activeTimers^\tree(p) \partto \activeTimers^\tree(p')$ as an injective partial function. This matching is meant to indicate that the timers
\(x\) and \(m(x)\) could represent the same timer in \(\M\). By abuse of notation, it is denoted as $m : p \leftrightarrow p'$. A matching is \emph{maximal} if $m$ is total or surjective, i.e., $m$ embeds $\activeTimers^\tree(p)$ in $\activeTimers^\tree(p')$ or $m^{-1}$ embeds $\activeTimers^\tree(p')$ in $\activeTimers^\tree(p)$. 
Recall that if two distinct timers \(x\) and \(y\) are both active in
the same state of $\tree$, it must hold that \(g(x) \neq g(y)\), i.e., they correspond to different timers in \(\M\) (by functional simulation). Hence, we say that a matching \(m\) is \emph{valid} if for all \(x \in \dom{m}\) and \(q \in Q^\tree\), we \emph{do not} have \(x, m(x) \in \activeTimers^\tree(q)\).

Second, given a run \(\pi = p_0 \myxrightarrow{i_1} p_1 \myxrightarrow{i_2} \dotsb \myxrightarrow{i_n} p_n\) and a matching $m : p_0 \leftrightarrow p'_0$, we want to check whether it is possible to \enquote{read} the same word $i_1 \ldots i_n$ from $p'_0$ while extending $m$ during this reading. If this is possible, this unique run, denoted \(\copyRun{p'_0}[\pi]\), is equal to
\(\pi' = p'_0 \myxrightarrow{i'_1} p'_1 \myxrightarrow{i'_2} \dotsb
\myxrightarrow{i'_n} p'_n\) such that:
\begin{enumerate}[label=\emph{(\roman*)}]
  \item
  If $i_j \in I$, then $i'_j = i_j$. 
  \item
  If \(i_j = \timeout{x}\) for some \(x \in X^\tree\), recall that 
  \(x \in \activeTimers^\tree(p_0)\) or
  \(x = x_{p_k}\) for some \(k\).
  Then, \(i'_j\) is either
  \(\timeout{m(x)}\), or \(\timeout{x_{p'_k}}\) with the same \(k\).
\end{enumerate}
When \(\copyRun{p'_0}[\pi]\) exists,
we therefore have a matching between $\pi$ and $\pi'$, denoted \(\matchingRun{m}{\pi}{\pi'} : \pi \leftrightarrow \pi'\), which extends $m$ such that 
\(\matchingRun{m}{\pi}{\pi'} = m \cup
\{x_{p_k} \mapsto x_{p'_k} \mid k \leq n\}\).

While two states of a Mealy machine are apart when they have different output responses to the same input word, we have five different apartness cases in the \MMT setting. An illustrative example is given after the definition.

\begin{definition}[Apartness]\label{def:apartness} 
  For two states \(p_0, p'_0\) and a matching $m : p_0 \leftrightarrow p'_0$, we say that $p_0,p'_0$ are \emph{\(m\)-apart}, denoted by
  \(p_0 \apart^m p'_0\), if there are
  \(\pi = p_0 \myxrightarrow{i_1} \dotsb \myxrightarrow[u]{i_n/o} p_n\) and
  \(\copyRun{p'_0}[\pi] = \pi' =
  p'_0 \myxrightarrow{i'_1} \dotsb \myxrightarrow[u']{i'_n/o'} p'_n\)
  with $\matchingRun{m}{\pi}{\pi'} : \pi \leftrightarrow \pi'$, and 
  \begin{description}
    \item [Structural apartness] either
    \( \matchingRun{m}{\pi}{\pi'}\) is invalid, or
    \item [Behavioral apartness] one of the following cases holds:
  \begin{gather}
    o \neq o'
    \tag{outputs}\label{eq:apartness:outputs}
    \\
    u = (x, c) \land u' = (x', c') \land c \neq c'
    \tag{constants}\label{eq:apartness:constants}
    \\
    p_n, p'_n \in \explored \land
      \lengthOf{\enabled{p_n}[\tree]} \neq \lengthOf{\enabled{p'_n}[\tree]}
    \tag{sizes}\label{eq:apartness:sizesEnabled}
    \\
    p_n, p'_n \in \explored \land \exists x {\in} \dom{\matchingRun{m}{\pi}{\pi'}}
    : \left(
      x \in \enabled{p_n}[\tree] \iff
      \matchingRun{m}{\pi}{\pi'}(x) \notin \enabled{p'_n}[\tree]
    \right)
    \tag{enabled}\label{eq:apartness:enabled}
  \end{gather}\end{description}
  The word $w = i_1 \dotso i_n$ is called a \emph{witness}
  of \(p_0 \apart^{m} p'_0\), denoted \(w \witness p_0 \apart^m p'_0\).
\end{definition}

\begin{example}\label{ex:tree:apartness}
  Let \(\tree\) be the  tree of \Cref{fig:ex:tree} and
  \(\pi =
  t_0 \myxrightarrow[(x_1, 2)]{i/o} t_1 \myxrightarrow{\timeout{x_1}/o} t_2
  \in \runs{\tree}\).
  Let us show that \(\copyRun[\emptyset]{t_3}[\pi]\) exists (\(\emptyset\) is
  the empty matching).
  As the first action in $\pi$ is \(i\), we take the transition
  \(t_3 \myxrightarrow[(x_6, 2)]{i/o'} t_6\). As the second action in $\pi$ is \(\timeout{x_1}\) with \(x_1\) a fresh timer, we retrieve the corresponding fresh timer
  $x_6$ in the new run.
  Hence,
  \(\pi' = \copyRun[\emptyset]{t_3}[\pi] =
  t_3 \myxrightarrow[(x_6, 2)]{i/o'} t_6 \myxrightarrow{\timeout{x_6}/o} t_9\) and $\matchingRun{m}{\pi}{\pi'} : \pi \leftrightarrow \pi'$ with $\matchingRun{m}{\pi}{\pi'} =  \emptyset \cup \{x_1 \mapsto x_6\}$. Let us now show some cases of apartness.
  The states $t_0,t_3$ are $\emptyset$-apart as the first transition of \(\pi\) outputs \(o\) but the first transition of
  \(\pi'\) outputs \(o' \neq o\), i.e., 
  \(i \witness t_0 \apart^{\emptyset} t_3\) by~\eqref{eq:apartness:outputs}. The states $t_1,t_6$ are also $\emptyset$-apart as
   \(t_1, t_6 \in \explored\)
  and
  \(\lengthOf{\enabled{t_1}} = 1 \neq \lengthOf{\enabled{t_6}} = 2\), i.e.,
  \(\emptyword \witness t_1 \apart^{\emptyset} t_6\)
  by~\eqref{eq:apartness:sizesEnabled}.
  Now, with $m : x_1 \mapsto x_3$, let
  \(\sigma =
  t_1 \myxrightarrow[(x_3, 3)]{i/o'}
  t_3 \myxrightarrow[]{\timeout{x_1}} t_5\)
  and
  \(\sigma' = \copyRun[m]{t_3}[\pi'] =
  t_3 \myxrightarrow[(x_6, 2)]{i/o'}
  t_6 \myxrightarrow[]{\timeout{x_3}} t_{10}\). We have a structural apartness as $\matchingRun{m}{\rho}{\rho'} = m \cup \{x_3 \mapsto x_6\}$ is invalid ($\activeTimers^\tree(t_6) = \{x_3,x_6\}$).
  We also have \(i \witness t_1 \apart^{m} t_3\) by~\eqref{eq:apartness:constants}. 
\end{example}

Observe that any extension \(m'\) of \(m\) is such that
\(w \witness p \apart^{m'} p'\), i.e., taking a larger matching does not break
the apartness,
as \(\copyRun{p'}[p \myxrightarrow{w}] = \copyRun[m']{p'}[p \myxrightarrow{w}]
= p' \myxrightarrow{w'}\).
Moreover, we claim that the definition of apartness is reasonable:
when \(p \apart^m p'\), then \(f(p) \neq f(p')\)
(the two states are really distinct) or
\(g(x) \neq g(m(x))\) for some \(x\)
(\(m\) and \(g\) do not agree).
\Cref{app:tree:extension_soundness} gives a proof.

\begin{restatable}{theorem}{extensionNSoundness}\label{thm:extension-n-soundness}
  Let \(\tree\) be an observation tree for an \good \MMT 
  with
  functional simulation \(\funcSim\), \(p, p' \in Q^\tree\), and
  \(m,m' : p \leftrightarrow p'\)
  matchings.
  Then,
  \begin{itemize}
  \item \(w \witness p \apart^m p' \land m \subseteq m'
  ~\implies~
  w \witness p \apart^{m'} p'\), and
  \item \(p \apart^m p' ~\implies~ f(p) \neq f(p') \lor \exists x \in \dom{m} :
  g(x) \neq g(m(x))\).
  \end{itemize}
\end{restatable}

The second implication of \Cref{thm:extension-n-soundness} is the one we leverage in our algorithm below. The reverse implication
is false as the learned \MMT can be smaller than that of the teacher (see the \MMT in  \Cref{fig:ex:learning:hypothesis} learned from the \MMT of \Cref{fig:ex:MMT:good}).

\subsection{Hypothesis construction}\label{sec:learning:hypo}

In this section, we provide the construction of a hypothesis 
\(\hypothesis\) from \(\tree\).
In short, we extend the observation tree such that some conditions are satisfied and we 
define a subset of \(Q^\tree\), called the \emph{basis}, that forms the set of
states of \(\hypothesis\).
Similar to \lsharp~\cite{VaandragerGRW22}, we then
\enquote{fold} the tree; that is, some
transitions \(q \myxrightarrow{} r\) must be redirected to some state \(p\) of the basis.
For \MMTs, we also need to map \emph{every} timer active in \(r\) to an
active timer of \(p\).
Formally, we define:
\begin{itemize}
  \item
  The \emph{basis} \(\basis\) is a subtree of \(Q^\tree\) such that
  \(q_0^\tree \in \basis\) and \(p \apart^m p'\) for any
  \(p \neq p' \in \basis\) and maximal matching \(m : p \leftrightarrow p'\).
  By \Cref{thm:extension-n-soundness}, we thus know that \(f(p) \neq f(p')\) or
  \(g(x) \neq g(m(x))\) for some \(x \in \dom{m}\).
  As we have this for every maximal \(m\), we \emph{conjecture}
  that \(f(p) \neq f(p')\).
  We may be wrong, i.e., \(f(p) = f(p')\) but we need a matching that is currently unavailable, due to unknown active timers which will be discovered later.

  \item
  The \emph{frontier} \(\frontier \subseteq Q^\tree\) is the set of immediate
  non-basis successors of basis states.
  We say \(p \in \basis\) and \(r \in \frontier\) are \emph{compatible}
  under a maximal matching \(m\) if \(\lnot (p \apart^m r)\).
  We write \(\compatible(r)\) for the set of all such pairs \((p, m)\).
  Hence \(\compatible(r)\) indicates all the possible ways to redirect a transition $q \rightarrow r$ to some state $p \in \basis$ together with an adequate maximal matching $m$.
\end{itemize}
During the learning algorithm, the tree is extended such that a \complete
\MMT $\hypothesis$ can be constructed from \(\tree\).
We will ensure that:
\begin{enumerate}[label=\emph{(\Alph*)}]
  \item\label{item:tree:basis:explored}each basis and frontier state is explored, i.e.
  \(\basis \cup \frontier \subseteq \explored\), in order to discover timers as
  fast as possible,
  \item\label{item:tree:basis:complete}the basis is \emph{complete}, i.e., \(p \myxrightarrow{i}\) is
  defined for every \(i \in I \cup \toevents{\enabled{p}[\tree]}\),
  and
  \item\label{item:tree:basis:timers}for every \(r \in \frontier\), \(\compatible(r) \neq \emptyset\) and
  \(\lengthOf{\activeTimers^\tree(p)} = \lengthOf{\activeTimers^\tree(r)}\)
  for every \((p, m) \in \compatible(r)\).
  The last constraint imposes the same number of active timers for both $p$ and $r$ since folding them using $(p,m)$ will merge them.
\end{enumerate}
These constraints are obtained via \symOutputQ and \symWaitQ, as we now illustrate.

\begin{example}\label{ex:learning:basis}
  In order to simplify the explanations, we assume from now on that a call to
  \(\symOutputQ(\symbolic{w \cdot i})\) with $\symbolic{i} \in I$
  automatically adds the corresponding transition to \(\tree\), and that a call to
  \(\symWaitQ(\symbolic{w})\) automatically calls \(\symOutputQ(\symbolic{w} \cdot \timeout{j})\),
  for every \(\timeout{j}\) deduced from the wait query, modifies
  updates accordingly, and adds the new explored states to \(\explored\).
  Moreover, we let
  \(\symOutputQ(q, i)\) denote \(\symOutputQ(\symbolic{w} \cdot \symbolic{i})\)
  and \(\symWaitQ(q)\) denote \(\symWaitQ(\symbolic{w})\) with \(\symbolic{w}\)
  such that \(q_0^\tree \myxrightarrow{\symbolic{w}} q \in \runs{\tree}\).
  
  Let \(\tree\) be the observation tree of \Cref{fig:ex:tree} for the \MMT \(\M\) of \Cref{fig:ex:MMT:good}.
  One can check that \(t_0, t_1\), and \(t_3\) are all pairwise apart under any
  maximal matching.
  So, \(\basis = \{t_0, t_1, t_3\}\) (highlighted in gray in the figure) and
  \(\frontier = \{t_2, t_5, t_6\}\).
  Moreover, we have \(\compatible(t_2) = \{(t_1, x_1 \mapsto x_1),
  (t_3, x_1 \mapsto x_1)\}\), and \(\compatible(t_5) = \compatible(t_6)
  = \emptyset\).
  We thus \emph{promote} \(t_6\) by moving it from the frontier to the basis,
  i.e., \(\basis = \{t_0, t_1, t_3, t_6\}\).
  To satisfy~\ref{item:tree:basis:complete}, we call
  \(\symOutputQ(t_6, i)\) and add a new transition
  \(t_6 \myxrightarrow[\bot]{i/o'} t_{11}\).
  We call \(\symWaitQ(t_9), \symWaitQ(t_{10})\), and \(\symWaitQ(t_{11})\),
  which yield \(\activeTimers(t_9) = \enabled{t_9} = \{x_3\}\),
  \(\activeTimers(t_{10}) = \emptyset\), and
  \(\activeTimers(t_{11}) = \enabled{t_{11}} = \{x_3, x_{11}\}\).
  Hence, we get~\ref{item:tree:basis:explored} with
  \(\frontier = \{t_2, t_5, t_9, t_{10}, t_{11}\}\). 

Observe that~\ref{item:tree:basis:timers} is not satisfied, due to \((t_3, m) \in \compatible(t_2)\) with $m : x_1 \mapsto x_1$ but \(\lengthOf{\activeTimers(t_3)} = 2\) while \(\lengthOf{\activeTimers(t_2)} = 1\).  To resolve this issue, we \emph{replay} the run \(\pi = t_3 \myxrightarrow{\timeout{x_1} \timeout{x_3}} t_8\) (witnessing that \(x_3\)  is active in \(t_3\) and eventually times out) from the state \(t_2\). That is, we extend the tree $\tree$ in order to create $\copyRun{t_2}[\pi]$ while extending $m$, if this is possible (see apartness paragraph). As \(t_2 \myxrightarrow{\timeout{x_1}} t_4\) already exists, let us replay \(t_5 \myxrightarrow{\timeout{x_3}}\) from \(t_4\). We call \(\symWaitQ(t_4)\) to learn \(\enabled{t_4} = \activeTimers(t_4)  = \{x_1\}\). We can stop the creation of $\copyRun{t_2}[\pi]$. Indeed,  as $\enabled{t_5} = \{x_1,x_3\}$, \(\enabled{t_4} = \{x_1\}\), we conclude that \(\timeout{x_1} \witness t_3 \apart^m t_2\) by~\eqref{eq:apartness:sizesEnabled}. Hence, $(t_3,m)$ is no longer a compatible pair and it is removed from $\compatible(t_2)$.
  
In general, replaying a run when~\ref{item:tree:basis:timers} is not satisfied may have two results due to newly defined transitions: either finding a new active timer in a state, or a new apartness pair (as we illustrated). In any case, the set of compatible pairs of a state gets reduced. See \Cref{app:learning:replay} for more details.
\end{example}

We now explain, via an example, how to construct a hypothesis \(\hypothesis\)
such that \(Q^\hypothesis = \basis\).
The idea is to pick a \((p, m) \in \compatible(r)\) for each
frontier state \(r\).
Then, the unique transition \(q \myxrightarrow{i} r\) in \(\tree\) becomes
\(q \myxrightarrow{i} p\) in \(\hypothesis\), i.e., $p$ and $r$ are merged. 
We also globally rename the timers according to \(m\).
\begin{figure}[t]
  \centering
  \begin{tikzpicture}[
  automaton,
  new/.style = {
  },
]
  \node [state, initial, basis, initial distance=7pt]   (t0)  {\treeNodeLabel{\(t_0\)}};
  \node [state, basis, right=40pt of t0]                (t1)  {\treeNodeLabel{\(t_1\)}};
  \node [state, above right=9pt and 70pt of t1]        (t2)  {\(t_2\)};
  \node [state, basis, below right=9pt and 70pt of t1] (t3)  {\treeNodeLabel{\(t_3\)}};
  \node [state, above right=9pt and 80pt of t2]        (t4)  {\(t_4\)};
  \node [state, above right=9pt and 80pt of t3]        (t5)  {\(t_5\)};
  \node [state, basis, below right=9pt and 80pt of t3] (t6)  {\treeNodeLabel{\(t_6\)}};
  \node [state, right=65pt of t4]                       (t7)  {\(t_7\)};
  \node [state, right=65pt of t5]                       (t8)  {\(t_8\)};
  \node [state, basis, right=65pt of t6]           (t9)  {\treeNodeLabel{\(t_9\)}};
  \node [state, below=17pt of t9]                       (t10) {\(t_{10}\)};
  \node [state, below=17pt of t10]                 (t11) {\(t_{11}\)};
  \node [state, right=66pt of t9]                  (t12) {\(t_{12}\)};
  \node [state, below=17pt of t12]                 (t13) {\(t_{13}\)};
  \node [state, below=17pt of t13]                 (t14) {\(t_{14}\)};
  \node [state, above=17pt of t12]                (t15) {\(t_{15}\)};
  \node [state, above=31pt of t15]                (t16) {\(t_{16}\)};
  \node [state, new, above=17pt of t7]                  (t17) {\(t_{17}\)};

  \path
    (t0)  edge                  node [above] {\(i/o\)}
                                node [below] {\(\updateFig{x_1}{2}\)}       (t1)
    (t1)  edge                  node [sloped, '] {\(i/o', \updateFig{x_3}{3}\)} (t3)
          edge                  node [sloped] {\(\timeout{x_1}/o,
                                              \updateFig{x_1}{2}\)}             (t2)
    (t2)  edge                  node [sloped] {\(\timeout{x_1}/o, \bot\)}       (t4)
    (t3)  edge                  node [sloped] {\(\timeout{x_1}/o,
                                              \updateFig{x_1}{2}\)}             (t5)
          edge                  node [sloped, '] {\(i/o', \updateFig{x_6}{2}\)} (t6)
    (t5)  edge                  node [above=-2pt] {\(\timeout{x_3}/o, \bot\)}   (t8)
    (t6)  edge                  node [above=-2pt] {\(\timeout{x_6}/o, \bot\)}   (t9)
    (t9)  edge                  node [above=-2pt] {\(\timeout{x_3}/o, \bot\)}   (t12)
    (t11) edge                  node [above=-2pt] {\(\timeout{x_3}/o, \bot\)}   (t14)
    (t15) edge                  node              {\(\timeout{x_3}/o, \bot\)} (t16)
  ;

  \draw [rounded corners = 10pt]
    let
      \p{s} = (t5.45),
      \p{t} = (t7.-180),
    in
      (\p{s}) -- ($(\x{s}, \y{t}) + (0.1, 0)$)
              -- (\p{t})
              node [above=-2pt, midway] {\(\timeout{x_1}/o, \bot\)}
  ;

  \draw [rounded corners = 10pt]
    let
      \p{s} = (t6.-43),
      \p{t} = (t10.180),
    in
      (\p{s}) -- ($(\x{s}, \y{t}) + (0.2, 0)$)
              -- (\p{t})
              node [above, midway] {\(\timeout{x_3}/o, \bot\)}
  ;

  \draw [rounded corners = 10pt]
    let
      \p{s} = (t6.-90),
      \p{t} = (t11.180),
      \p{m} = (\x{s}, \y{t}),
    in
      (\p{s}) -- (\p{m})
              -- (\p{t})
              node [midway, above=-1pt] {\(i/o', \updateFig{x_{11}}{2}\)}
  ;

  \draw[rounded corners = 10pt]
    let
      \p{s} = (t11.45),
      \p{t} = (t13.-160),
      \p{m} = ($(\x{s}, \y{t}) + (0.3, 0)$),
    in
      (\p{s}) -- (\p{m})
              -- (\p{t})
              node [midway, above=-2pt] {\(\timeout{x_{11}}/o, \bot\)}
  ;

  \draw[rounded corners = 10pt]
    let
      \p{s} = (t9.45),
      \p{t} = (t15.-160),
      \p{m} = ($(\x{s}, \y{t}) + (0.3, 0)$),
    in
      (\p{s}) -- (\p{m})
              -- (\p{t})
              node [midway, above=-2pt] {\(i/o', \bot\)}
  ;

  \draw[new, rounded corners = 10pt]
    let
      \p{s} = (t4.50),
      \p{t} = (t17.-160),
      \p{m} = ($(\x{s}, \y{t}) + (0.3, 0)$),
    in
      (\p{s}) -- (\p{m})
              -- (\p{t})
              node [midway, above=-2pt] {\(\timeout{x_1}/o, \bot\)}
  ;
\end{tikzpicture}
\begin{tikzpicture}[
  automaton,
  every loop/.append style = {
    min distance = 1mm,
  },
]
  \node [state, initial]      (t0)  {\(t_0\)};
  \node [state, right=50pt of t0]  (t1)  {\(t_1\)};
  \node [state, right=50pt of t1]  (t3)  {\(t_3\)};
  \node [state, right=80pt of t3]  (t6)  {\(t_6\)};
  \node [state, right=70pt of t6]  (t9)  {\(t_9\)};

  \path
    (t0)  edge              node [above] {\(i/o\)}
                            node [below=-2pt] {\(\updateFig{y_1}{2}\)}      (t1)
    (t1)  edge              node [above] {\(i/o'\)}
                            node [below=-2pt] {\(\updateFig{y_2}{3}\)}      (t3)
          edge [loop above] node {\(\timeout{y_1}/o, \updateFig{y_1}{2}\)}  (t1)
    (t3)  edge              node [below=-2pt] {\(i/o', \updateFig{y_1}{2}\)}(t6)
          edge [bend left=10]  node [pos=0.46, above=-2pt] {\(\timeout{y_1}/o,
                                            \updateFig{y_1}{2}\)}           (t6)
    (t6)  edge              node {\(\timeout{y_1}/o, \bot\)}                (t9)
          edge [loop above] node {\(i/o', \updateFig{y_1}{2}\)}             (t6)
    (t9)  edge [loop above] node {\(i/o', \bot\)}                           (t9)
  ;

  \draw [rounded corners=10pt]
    let
      \p{e} = ($(t0.south east) + (0, -0.5)$),
      \p{s} = ($(t6)$)
    in
    (t6)  -- (\x{s}, \y{e})
          node [midway, right] {\(\timeout{y_2}/o, \bot\)}
          -- (\p{e})
          -- (t0.south east)
  ;
  \draw [rounded corners=10pt]
    let
      \p{e} = ($(t0) + (0, -0.8)$),
      \p{s} = ($(t9)$)
    in
    (t9)  -- (\x{s}, \y{e})
          node [midway, right] {\(\timeout{y_2}/o, \bot\)}
          -- (\p{e})
          -- (t0)
  ;
\end{tikzpicture}
  \caption{On top, an observation tree from which the hypothesis \MMT at the
  bottom is constructed, with \(y_1 = \equivClass{x_1}_\PER\) and \(y_2 = \equivClass{x_3}_\PER\).
  Basis states are highlighted with a gray background.}\label{fig:ex:learning:hypothesis}
\end{figure}
\begin{example}\label{ex:learning:hypothesis}
  Let \(\M\) be the \MMT of \Cref{fig:ex:MMT:good} and
  \(\tree\) be the observation tree of \Cref{fig:ex:learning:hypothesis},
  with
  \(\basis = \{t_0, t_1, t_3, t_6, t_9\}\) and
  \(\frontier = \{t_2, t_5, t_{10}, t_{11}, t_{12}, t_{15}\}\).
  Constraints~\ref{item:tree:basis:explored} to~\ref{item:tree:basis:timers} are satisfied and
  \begin{align*}
    \compatible(t_2) &= \{(t_1, x_1 \mapsto x_1)\}
    \\
    \compatible(t_5) &= \{(t_6, x_6 \mapsto x_1, x_3 \mapsto x_3)\}
    \\
    \compatible(t_{10}) = \compatible(t_{12}) &= \{(t_0, \emptyset)\}
    \\
    \compatible(t_{11}) &= \{(t_6, x_6 \mapsto x_{11}, x_3 \mapsto x_3)\}
    \\
    \compatible(t_{15}) &= \{(t_9, x_3 \mapsto x_3)\}
  \end{align*}

  We construct $\hypothesis$ with $Q^\hypothesis = \basis$.
  While defining the transitions \(q \myxrightarrow{} q'\) is easy when
  \(q, q' \in \basis\),
  we have to redirect the transition to some basis state when \(q' \in \frontier\).
  To do so, we first define a map \(\foldFunction : \frontier \to \basis\) (to redirect transitions into the basis),
  and an equivalence relation \(\PER\) over the set of
  active timers of the basis and the frontier (to declare which timers are equal).
  For each \(r \in \frontier\), we pick \((p, m) \in \compatible(r)\),
  define \(\foldFunction(r) = p\), and add \(x \PER m(x)\)
  for every \(x \in \dom{m}\) (and compute the symmetric and transitive closure
  of \(\PER\)).
  Here,
  \begin{gather*}
    \foldFunction(t_2) = t_1,
    \foldFunction(t_5) = \foldFunction(t_{11}) = t_6,
    \foldFunction(t_{10}) = \foldFunction(t_{12}) = t_0,
    \foldFunction(t_{15}) = t_9,
    \\
    \text{\(x_1 \PER x_6 \PER x_{11}\), and
    \(x_3 \PER x_3\).}
  \end{gather*}
  We check that the constructed relation $\PER$ is \emph{valid}, i.e., it does not lead to an undesirable situation $x \PER y$ and \(x, y \in \activeTimers^\tree(q)\) for some $q$ (in which case, we restart again by picking some different pair \((p, m) \in \compatible(r)\)).
  In this example, $\equiv$ is valid, and
  we construct \(\hypothesis\) by copying the transitions starting from a
  basis state (while folding the tree when required), except that a timer \(x\)
  is replaced by its equivalence class \(\equivClass{x}_\PER\).
  \Cref{fig:ex:learning:hypothesis} gives the resulting \(\hypothesis\), which is symbolically equivalent to $\M$.
\end{example}

We highlight that it is \emph{not} always possible to construct a valid \(\PER\). However, in this case, we can still construct a \emph{generalized} \MMT, in which the
transitions can arbitrarily rename the active timers.
The size of that machine is  also \(\lengthOf{\basis}\), and
from it a classical \MMT can be constructed of size
\(n! \cdot \lengthOf{\basis}\), with
\(n = \max_{p \in \basis} \lengthOf{\activeTimers^\tree(p)}\)
(see \Cref{app:learning:hypo}).
We observed that a valid $\PER$ can always be constructed for all our benchmarks.

\subsection{Main loop}\label{sec:learning:algo}

We finally give the main loop of \lsharpMMT.
We initialize \(\tree\) to only contain \(q_0^\tree\), \(\basis = \explored
= \{q_0^\tree\}\), and \(\frontier = \emptyset\).
The main loop is split into two parts:
\begin{description}
  \item[Refinement loop]
  The \emph{refinement loop} extends the tree to obtain the
  conditions~\ref{item:tree:basis:explored} to~\ref{item:tree:basis:timers}
  (see \cref{sec:learning:hypo} page~\pageref{item:tree:basis:explored}),
  by performing
  the following operations, in this order, until no more
  changes are possible:
  \begin{description}
    \item[\seismic]
    If we discover a new active timer in a basis state, then it may
    be that \(\lnot (q \apart^m q')\) for some \(q, q' \in \basis\) and
    maximal \(m\), due to the new timer.
    To avoid this, we reset the basis back to \(\{q_0^\tree\}\), as soon as a new
    timer is found, without removing states from \(\tree\).

    \item[\promotion]
    If \(\compatible(r)\) is empty for some frontier state \(r\), then we
    know that \(q \apart^m r\) for every \(q \in \basis\) and maximal matching
    \(m : q \leftrightarrow r\).
    Hence, we promote \(r\) to the basis.

    \item[\completion]
    If an \(i\)-transition is missing from some basis state \(p\), we complete
    the basis with that transition.

    \item[\minimizationActive]
    We ensure
    \(\lengthOf{\activeTimers^\tree(p)} = \lengthOf{\activeTimers^\tree(r)}\)
    for every \((p, \cdot) \in \compatible(r)\).
  \end{description}

\medskip
  \item[Hypothesis and equivalence]
  We call \(\symEquivQ(\hypothesis)\) with \(\hypothesis\) a hypothesis.
  If the answer is \yes, we return \(\hypothesis\).
  Otherwise, we process the counterexample, as we now explain in an example.
\end{description}

\begin{example}
  Let \(\tree\) be the observation tree of \Cref{fig:ex:counterexample} with
  \(\basis = \{t_0, t_1\}\), \(\frontier = \{t_2, t_3\}\), and
  \(\compatible(t_2) = \compatible(t_3) = \{(t_1, x_1 \mapsto x_1)\}\).
  \Cref{fig:ex:counterexample} also gives the \MMT constructed from \(\tree\),
  which is not symbolically equivalent to the \MMT of \Cref{fig:ex:MMT:good},
  with \(\symbolic{w} = i \cdot i \cdot \timeout{1} \cdot \timeout{2}\) as
  a counterexample.
  We extend \(\tree\) such that \(\symbolic{w}\) can be read in \(\tree\).
  That is, we call \(\symWaitQ(t_5)\) and discover that the transition from
  \(t_1\) to \(t_3\) must start the timer \(x_3\).
  Hence, \(t_3\) has no compatible state anymore (i.e., we found a new apartness
  pair) and gets promoted.
  After completing the tree and performing \symWaitQ on the frontier states,
  we get the tree from \Cref{fig:ex:tree}.
\end{example}

In our example, simply adding the symbolic word provided by the teacher was
enough to discover a new apartness pair (meaning that the hypothesis can no
longer be constructed).
However, there may be cases where we need to \emph{replay}
(see \Cref{ex:learning:basis}) a part of the newly added run
\(t_0 \myxrightarrow{\symbolic{w}}\): split the run into
\(t_0 \myxrightarrow{\symbolic{u}} t \myxrightarrow{\symbolic{v}}\) with
\(t \in \frontier\) and replay \(t \myxrightarrow{\symbolic{v}}\) from the state
compatible with \(t\) that was selected to build the hypothesis.
Repeat this principle until a compatible set is reduced.

\section{Implementation and Experiments}\label{sec:implementation}

We have implemented \lsharpMMT as an open-source tool.\footnote{See:
\url{https://gitlab.science.ru.nl/bharat/mmt_lsharp} and
Zenodo~\cite{zenodo-artifact}.}
Note that our prototype implementation only uses symbolic queries, with no explicit translation into concrete ones.
As we do not yet have a timed conformance testing algorithm for checking equivalence
between a hypothesis and the teacher's \MMT, we utilize a BFS algorithm to check for symbolic equivalence between the two \MMTs.\footnote{That is, seeking a difference in behavior in the 
product of the two \MMTs.} 

We evaluated the performance of our tool on a selection of both real and synthetic benchmarks.
We use the AKM, TCP and Train benchmarks from~\cite{VaandragerBE21}, and the CAS, Light and PC benchmarks from~\cite{AichernigPT20}.
These have also been used for experimental evaluation by~\cite{VaandragerBE21,Waga23,KogelKG23} and can be described as Mealy machines with a single timer (MM1Ts).
We introduce two additional benchmarks with two timers: a model of an FDDI station, and the \MMT of~\Cref{fig:ex:MMT:good}.
We refer to \Cref{appendix:fddi_model} for details on our FDDI benchmark.
We did not include the FDDI 2 process benchmark from~\cite{Waga23},
as our implementation cannot (yet) handle the corresponding generalized MMTs.
Finally, we learned instances of the Oven and WSN MMLTs benchmarks from~\cite{KogelKG23}.
We modified the timing parameters to generate smaller
MM1Ts.
For each experiment, we record the number of $\symOutputQ$, $\symWaitQ$,
$\symEquivQ$, and the time taken to finish the experiment.
Note, in practice, a $\symWaitQ$, in addition to returning the list of timeouts and
their constants, also provides the outputs of the 
timeout transitions.
This is straightforward, as a $\symWaitQ$ must necessarily trigger the timeouts in
order to observe them.
Thus, we do not count the $\symOutputQ$ associated with a $\symWaitQ$.\footnote{An upper bound on the number of \(\symOutputQ\) can be obtained by
adding \(\lengthOf{X} \lengthOf{\symWaitQ}\) to the values for
\(\lengthOf{\symOutputQ}\).}
{\begin{table}[t]
      \footnotesize
      \caption{Experimental Results.}\label{tbl:expr}
      \centering
      \begin{tabular}{lrrrrrrr} \toprule
            Model                    & $|Q|$ & $|I|$ & $|X|$ & $|\symWaitQ|$ & $|\symOutputQ|$ & $|\symEquivQ|$ & Time[ms]  \\
            \specialrule{\cmidrulewidth}{0pt}{0pt}
            AKM                      & 4     & 5     & 1                 & 22            & 35              & 2              & 684        \\
            CAS                      & 8     & 4     & 1                 & 60            & 89              & 3              & 1344       \\
            Light                    & 4     & 2     & 1                 & 10            & 13              & 2              & 302        \\
            PC                       & 8     & 9     & 1                 & 75            & 183             & 4              & 2696       \\
            TCP                      & 11    & 8     & 1                 & 123           & 366             & 8              & 3182     \\
            Train                    & 6     & 3     & 1                 & 32            & 28              & 3              & 1559        \\
            MMT of Fig.~\ref{fig:ex:MMT:good} & 3     & 1     & 2                 & 11            & 5               & 2              & 1039       \\
            FDDI 1-station           & 9     & 2     & 2                 & 32            & 20              & 1              & 1105       \\
            Oven                     & 12    & 5     & 1                 & 907           & 317             & 3              & 9452        \\
            WSN                      & 9     & 4     & 1                 & 175           & 108             & 4              & 3291       \\
            \bottomrule
      \end{tabular}
\end{table}
 }
The results are given in \Cref{tbl:expr}.
Here, we highlight that the presented data is only meant to indicate that our algorithm is able to learn these models within a reasonable amount of time. The table also gives insight into which benchmarks require more or less (symbolic) queries to be learned with our algorithm.

Comparison with other learning algorithms for timed systems is complicated.
First, we need to convert the numbers of symbolic \lsharpMMT queries to
concrete queries.
As our prototype does not implement the concrete queries, we can only use the theoretical polynomial bounds given in \Cref{lemma:queries}. Note that for MM1Ts, each symbolic query can be implemented using a single concrete query (see Lemma 3 in~\cite{VaandragerBE21}).
For the FDDI protocol using two timers, we need at most 1282 concrete queries in total while around 118200 queries are used in~\cite{Waga23}.
Second, several algorithms presented in the literature learn TAs~\cite{Waga23,AichernigPT20,XuAZ22,AnCZZZ20}.
Typically, these models have different numbers of states and transitions than an MMT model:  Mealy machines tend to be more compact than TAs, but the use of timers may lead to more states than a TA encoding.
Therefore we cannot just compare numbers of queries.
As a third complication,  equivalence queries can be implemented in different ways, which may affect the total number of queries required for learning.
Finally, concerning models close to our \MMTs, MMLTs~\cite{KogelKG23} can be converted to equivalent MM1Ts~\cite{VaandragerBE21}, but this may blow up the number of states.
Since \lsharpMMT learns the MM1Ts, it is less efficient than the MMLT learner of~\cite{KogelKG23} which learns the more compact MMLT representations, or than the learner of~\cite{VaandragerBE21} specially designed for MM1Ts.  However, \lsharpMMT can handle a larger class of models, as it can handle multiple timers.

\section{Future work}
A major challenge in \lsharpMMT is to infer the update on transitions.
This would become easier if we know in advance that timers can only be started
by specific inputs, akin to what is done in
\emph{event recording automata} (ERA)~\cite{AlurFH99} for TAs.
Although, as discussed in~\cite{VaandragerBE21}, the restrictions of ERAs
make it hard to capture the timing behavior of standard network protocols. It
would be interesting to study the theoretical complexity of learning a system
that can be modelled by ERAs as well as \MMTs.
A more interesting trail would be to allow transitions to start multiple timers,
instead of a single one, which would permit more complex models (such as those resulting of parallel composition of simpler models) to be learned.
It would also be interesting to explore the theory of generalized MMTs.

\bibliographystyle{splncs04}
\bibliography{abbreviations,references}

\clearpage

\appendix

\section{More details on equivalence of MMTs}\label{app:equivalence}

In this appendix, we give more details about equivalences of \MMTs.
We first define the classical notion of \emph{timed bisimulation} equivalence, and prove that the symbolic equivalence introduced in \Cref{sec:MMT:equivalence} refines timed bisimulation equivalence.
Moreover, we give a counterexample for the other direction.
That is, we prove that timed bisimulation equivalence does not imply symbolic equivalence.
Next we introduce \emph{timed trace} equivalence and show that timed bisimulation equivalence implies timed trace equivalence.
Finally we prove that, for a subclass of \emph{race avoiding} MMTs, timed trace equivalence implies symbolic equivalence, which means that for this subclass all three proposed notions of equivalence coincide.

\subsection{Timed bisimulation equivalence}

We write $\conf{\M}$ to denote the set of configurations of an MMT $\M$.  A \emph{timed bisimulation} between MMTs $\M$ and $\N$ is a relation over $\conf{\M}\times\conf{\N}$ that describes how $\M$ and $\N$ may simulate each others behavior.

\begin{definition}[Timed bisimulation]
	A \emph{timed bisimulation} between MMTs $\M$ and $\N$ is a relation $R \subseteq \conf{\M}\times\conf{\N}$ that satisfies $(q_0^{\M}, \emptyset )\; R \; (q_0^{\N}, \emptyset )$ and, for
	all $C, C' \in\conf{\M}$, $D, D'\in \conf{\N}$, $d \in\nnr$, $i \in I$, $o \in O$, $x \in X^{\M}$ and $y \in X^{\N}$, $C \; R \; D$ implies
	\begin{enumerate}
		\item 
		if  $C \myxrightarrow{d} C'$ then there exists a $D'$ such that $D \myxrightarrow{d} D'$ and $C'\; R \; D'$,
		\item 
		if  $C \myxrightarrow{i/o} C'$ then there exists a $D'$ such that $D \myxrightarrow{i/o} D'$ and $C'\; R \; D'$,
		\item 
		if $C \myxrightarrow{\timeout{x}/o} C'$ then there exist $D'$ and $y$ such that $D \myxrightarrow{\timeout{y}/o} D'$ and $C'\; R \; D'$,
		\item 
		if $D \myxrightarrow{d} D'$ then there exists a $C'$ such that $C \myxrightarrow{d} C'$ and $C'\; R \; D'$, 
		\item 
		if $D \myxrightarrow{i/o} D'$ then there exists a $C'$ such that $C \myxrightarrow{i/o} C'$ and $C'\; R \; D'$, and
		\item 
		if $D \myxrightarrow{\timeout{y}/o} D'$ then there exist $C'$ and $x$ such that $C \myxrightarrow{\timeout{x}/o} C'$ and $C'\; R \; D'$.
	\end{enumerate}
	We say that $\M$ and $\N$ are \emph{timed bisimulation equivalent} or \emph{bisimilar}, and write \(\M \equivalent \N\), if there exists a timed bisimulation between $\M$ and $\N$.
\end{definition}

The next lemma follows directly from the definitions.

\begin{lemma}\label{lemma: timed bisimulation is an equivalence}
	$\equivalent$ is an equivalence relation on MMTs.
\end{lemma}

\subsection{Timed equivalence does not refine symbolic equivalence}%
\label{app:counterexample}
Consider the \MMT $\M$ of \Cref{fig:timedNotSym}.
We make \(\M\) complete by adding \(q \xrightarrow[\bot]{i/o} q_6\) for every
missing transition.
Observe that \(O = \{o, o_1, o_2\}\) and that \(q_4 \xrightarrow{\timeout{z}}\)
outputs \(o_1\) while \(q_7 \xrightarrow{\timeout{y}}\) outputs \(o_2\).
\begin{figure}[ht!]
	\centering
	\begin{tikzpicture}[
  automaton,
  state/.append style = {
    minimum size = 15pt,
  },
]
  \node [state, initial]        (q0)  {\(q_0\)};
  \node [state, right=of q0]    (q1)  {\(q_1\)};
  \node [state, right=of q1]    (q2)  {\(q_2\)};
  \node [state, right=of q2]    (q3)  {\(q_3\)};
  \node [state, above right=of q3]    (q4)  {\(q_4\)};
  \node [state, below right=of q4]    (q5)  {\(q_5\)};
  \node [state, right=of q5]    (q6)  {\(q_6\)};
  \node [state, below right=of q3]    (q7)  {\(q_7\)};

  \path
    (q0)  edge    node {\(i/o, \updateFig{x}{1}\)}        (q1)
    (q1)  edge    node {\(i/o, \updateFig{y}{1}\)}        (q2)
    (q2)  edge    node {\(i/o, \updateFig{z}{1}\)}        (q3)
    (q3)  edge    node {\(\timeout{y}/o, \bot\)}          (q4)
          edge    node [left] {\(\timeout{z}/o, \bot\)}   (q7)
    (q4)  edge    node {\(\timeout{z}/o_1, \bot\)}        (q5)
    (q5)  edge    node {\(\timeout{x}/o, \bot\)}          (q6)
    (q7)  edge    node [right] {\(\timeout{y}/o_2, \bot\)}(q5)
  ;
\end{tikzpicture}
  \caption{An \MMT with
		\(\activeTimers(q_0) = \activeTimers(q_6) = \emptyset\),
		\(\activeTimers(q_1) = \activeTimers(q_5) = \{x\}\),
		\(\activeTimers(q_2) = \activeTimers(q_7) = \{x, y\}\),
		\(\activeTimers(q_3) = \{x, y, z\}\),
		\(\activeTimers(q_4) = \{x, z\}\).
		Every missing transition \(q \xrightarrow[u]{i/\omega} p\) to obtain a complete
		\MMT is such that \(p = q_6, \omega = o\), and \(u = \bot\).
	}\label{fig:timedNotSym}
\end{figure}
Moreover, consider a copy $\N$ of \(\M\) in which \(o_1\) and \(o_2\) are swapped.
Let us argue that \(\M \equivalent \N\) but \(\M \notSymEquivalent \N\), starting
with the latter.
We write \(q_j^\M\) and \(q_j^\N\) to distinguish the states of \(\M\)
and \(\N\).

Let \(\symbolic{w} = i \cdot i \cdot i \cdot \timeout{1,2} \cdot \timeout{1,3}\) b
a symbolic word, inducing the following runs:
\begin{align*}
	&q_0^\M \xrightarrow[(x, 1)]{i/o}
	q_1^\M \xrightarrow[(y, 1)]{i/o}
	q_2^\M \xrightarrow[(z, 1)]{i/o}
	q_3^\M \xrightarrow[\bot]{\timeout{y}/o}
	q_4^\M \xrightarrow[\bot]{\timeout{z}/o_1}
	q_5^\M
	\\
	&q_0^\N \xrightarrow[(x, 1)]{i/o}
	q_1^\N \xrightarrow[(y, 1)]{i/o}
	q_2^\N \xrightarrow[(z, 1)]{i/o}
	q_3^\N \xrightarrow[\bot]{\timeout{y}/o}
	q_4^\N \xrightarrow[\bot]{\timeout{z}/o_2}
	q_5^\N.
\end{align*}
Hence, \(\M \notSymEquivalent \N\) as the last pair of transitions has
different outputs.

So, it remains to show that \(\M \equivalent \N\). Observe that
MMTs $\M$ and $\N$ share the same states, and for each state $q_i$ also the set ${\cal C}_i$ of reachable configurations is the same, where
\begin{eqnarray*}
	{\cal C}_0 & = & \{ (q_0, \emptyset) \}\\
	{\cal C}_1 & = & \{ (q_1, \kappa) \mid \kappa \models x \leq 1\}\\
	{\cal C}_2 & = & \{ (q_2, \kappa) \mid \kappa \models x \leq y \leq 1\}\\
	{\cal C}_3 & = & \{ (q_3, \kappa) \mid \kappa \models x \leq y \leq z \leq 1\}\\
	{\cal C}_4 & = & \{ (q_4, \kappa) \mid \kappa \models x =0 \wedge z \leq 1\}\\
	{\cal C}_5 & = & \{ (q_5, \kappa) \mid \kappa \models x =0 \}\\
	{\cal C}_6 & = & \{ (q_6, \emptyset) \}\\
	{\cal C}_7 & = & \{ (q_7, \kappa) \mid \kappa \models x = y = 0 \}
\end{eqnarray*}
We define a bisimulation $R$ between configurations of $\M$ and $\N$ as the identity relation for all configurations except two, which are paired crosswise: 
\begin{eqnarray*}
	R & = & \{ (C, C) \mid C \in {\cal C}_0 \cup {\cal C}_1 \cup {\cal C}_2 \cup {\cal C}_3 \cup {\cal C}_5 \cup {\cal C}_6  \cup \{ (q_4, \kappa) \mid \kappa \models x =0 < z \leq 1\} \} \cup\\
	& & \{ ((q_4, x=z=0), (q_7, x=y=0)) \} \cup  \{ ((q_7, x=y=0), (q_4, x=z=0)) \}
\end{eqnarray*}
It is routine to verify that $R$ is a timed bisimulation: for each pair of identical configurations in $R$, each outgoing transition from one configuration is simulated by exactly the same transition in the other one, except for configuration $(q_3, x=y=z=0)$, where the outgoing transition to $q_4$ is simulated by the outgoing transition to $q_7$, and vice versa.
Moreover, the outgoing $\timeout{z}$-transition from $q_4$ is simulated by the outgoing $\timeout{y}$-transition from $q_7$, and vice versa, thus ensuring that both $o_1$-transitions are matched, as well as both $o_2$-transitions.
We thus conclude that \(\M \equivalent \N\).

\subsection{Symbolic equivalence refines timed bisimulation equivalence}

In order to prove that symbolic equivalence refines timed bisimulation, a key step is the construction of a \emph{symbolic tree unfolding}, a tree-shaped MMT $\symbtree{\M}$ based on the symbolic semantics of a given MMT $\M$.
The states of $\symbtree{\M}$ are just the symbolic input words of $\M$.
The timers of $\symbtree{\M}$ are the positive natural numbers: a timer $j$ can only be started in the $j$-th transition of a run (and is thus never restarted).
From every state $\w$ there are transitions to all its one input extensions. 

Note that formally, $\symbtree{\M}$ is not an MMT because it has infinitely many states and timers.  In addition, $\timeout{x}$-transitions start a timer different from $x$.  We imposed the finiteness restrictions on MMTs to ensure termination of our learning algorithm. The restriction that a $\timeout{x}$-transition may only restart $x$, has been introduced mainly to simplify some definitions and proofs. The definitions of the timed semantics, symbolic semantics and timed bisimulation do not depend on MMTs being finite, or a $\timeout{x}$-transition only restarting $x$.  Clearly, timed bisimulation is also an equivalence for these more general MMTs.  If we show that symbolic equivalence refines timed bisimulation for the more general MMTs, it will certainly also hold for the more restrictive ones.
The situation is analogous to one in real analysis where, as observed by the French mathematician Jacques Hadamard, \enquote{The shortest path between two truths in the real domain passes through the complex domain.}\footnote{See \url{https://homepage.divms.uiowa.edu/~jorgen/hadamardquotesource.html}.}

The next simple lemma is needed to show that the construction of $\symbtree{\M}$ is well-defined.

\begin{lemma}\label{lemma symbolic input language}
	Let $\M$ be an \MMT.
	Suppose $\silanguage{\M}$ contains words
	$\w \; \w' \timeout{c, j}$ and $\w \; \w'' \timeout{c', j}$ with
 $\length{\w} = j$.
	Then $c = c'$.
\end{lemma}
\begin{proof}
	Words $\w \; \w' \timeout{c, j}$ and $\w \; \w'' \timeout{c', j}$ correspond to unique runs $\pi'$ and $\pi''$ of $\M$, which have a shared common prefix $\pi$.
	The timeouts at the end of runs $\pi'$ and $\pi''$ are triggered by a timer that started by the last transition of run $\pi$.
	Both $c$ and $c'$ equal the value to which this timer was set, and are therefore equal.
\qed\end{proof}

\begin{definition}\label{def symb tree}
	Let $\M$ be an \MMT.
	Then the \emph{symbolic tree unfolding} of $\M$, denoted $\symbtree{\M}$, is the tuple \((X, Q, q_0, \activeTimers, \delta)\), where
	\begin{itemize}
		\item 
		$X = \natplus$,
		\item 
		$Q = \silanguage{\M}$,
		\item 
		$q_0 = \epsilon$,
		\item 
		\(\delta : Q \times \actions{\M} \partto Q \times O \times \updates{\M}\) is specified as follows.
		Let $\w \in Q$ and $i \in \actions{\M}$.
		If $i \in I$ then $\delta(\w, i)$ is defined iff $\w\; i \in Q$.
		If $i = \timeout{j}$, for some $j \in X$, then $\delta(\w, i)$ is defined iff there is some $c \in \natplus$ such that $\w\; \timeout{c, j} \in Q$.
		Suppose that $\delta(\w, i)$ is defined and equal to $(\w', o, u)$. Then:
		\begin{itemize}
			\item 
		If $i \in I$ then $\w' = \w \; i$.
		Otherwise, if $i = \timeout{j}$, then $\w' = \w \; \timeout{c, j}$, for some $c$.
		Note that in this case $c$ is uniquely determined by Lemma~\ref{lemma symbolic input language}.
		\item 
		$o = \last{\ofunction{\M}{\w'}}$, where function \emph{last} returns the last element of a nonempty sequence.
		\item 
		Let $\length{\w'} = j$.  If there exists $\w''$ and $c$ such that $\w' \; \w'' \timeout{c, j} \in Q$ then $u = (j, c)$.
		Note that in this case $c$ is uniquely determined by Lemma~\ref{lemma symbolic input language}.
		Otherwise, $u = \perp$.
	\end{itemize}
	\item 
	Suppose $\w \in Q$ with $\length{\w} = j$.
	Then a timer $k$ is active in $\w$, that is, $k \in \activeTimers(\w)$, iff $k \leq j$ and $Q$ contains a state
	$\w \; \w' \timeout{c, k}$, for some $\w'$ and $c$.
	\end{itemize}
\end{definition}

Whereas, for each timer that is started in $\symbtree{\M}$, there exists a feasible run in which this timer expires, \MMT $\M$ may contain \emph{ghost timers}, which never expire because another timer will always timeout first and stop them.
Ghost timers do not affect the observable behavior of an \MMT, but we need to actually prove this in order to show that $\symbtree{\M}$ and $\M$ are bisimilar.

\begin{definition}
	Let $\M$ be an \MMT with feasible run $\pi$ that starts in the initial state and ends in state $q$. A timer $x \in \activeTimers(q)$ is \emph{live} after $\pi$ iff there exists a feasible run $\pi'$ of $\M$ that extends $\pi$ in which $x$ expires before being stopped or restarted.
	A timer that is not live is called a \emph{ghost timer}.
\end{definition}

\begin{lemma}\label{lemma cause two}
	Let $\M$ be an \MMT. 
	For $\pi$ a feasible run of $\M$ that starts in the initial state and ends in state $q$.
    For $x \in \activeTimers(q)$, let $\causetwo{\pi}{x}$ denote the second element of the pair $\cause{\pi}{x}$.
    Then $\causetwo{\pi}{\cdot}$ is a bijection between the live timers of $\activeTimers(q)$ and the set $\activeTimers(\w)$ of timers of state $\w =  \sinpw{\pi}$ of MMT $\symbtree{\M}$.
\end{lemma}

We are now prepared to prove the existence of a timed bisimulation between MMTs $\symbtree{\M}$ and $\M$.

\begin{lemma}\label{lemma:symbolic tree equivalent}
	Let $\M$ be an \MMT. Then $\M\equivalent\symbtree{\M}$.
\end{lemma}
\begin{proof}
We define a relation $R$ between configurations of $\M$ and $\symbtree{\M}$ $=$ $\N$ as follows.
Configurations $(q, \kappa)$ and $(\w, \lambda)$, of $\M$ and $\N$, respectively, are related by $R$ iff $\M$ has a timed run $\rho$ with $\untimeRun{\rho} = \pi$, first configuration $(q_0^{\M}, \emptyset)$, and final configuration $(q, \kappa)$ such that
$\sinpw{\pi} = \w$ and, for each timer $x \in \activeTimers(q)$ that is live after $\pi$, $\kappa(x) = \lambda(\causetwo{\pi}{x})$. It is routine to verify that $R$ satisfies the conditions of a timed bisimulation.
Clearly, 
$(q_0^{\M}, \emptyset )\; R \; (\epsilon, \emptyset )$, as the timed run $\rho = (q_0^{\M}, \emptyset ) \; 5 \; (q_0^{\M}, \emptyset )$ end in
configuration $(q_0^{\M}, \emptyset )$, $\untimeRun{\rho} = q_0^{\M}$,
$\sinpw{q_0^{\M}} = \epsilon$, and the condition on timer values vacuously holds since there are no timers.
Now suppose that $C \; R \; D$, where $C = (q, \kappa)$ and $D = (\w, \lambda)$.
\begin{enumerate}
	\item 
	Suppose  $C \myxrightarrow{d} C'$.
	Then $\kappa(x) \geq d$, for all timers $x \in \activeTimers(q)$,
	and $C'= (q, \kappa - d)$.
	Since for each timer $x \in \activeTimers(q)$ that is live after $\pi$, $\kappa(x) = \lambda(\causetwo{\pi}{x})$, and since by Lemma~\ref{lemma cause two},
	$\causetwo{\pi}{\cdot}$ is a bijection between the timers of $\activeTimers(q)$ that are live after $\pi$ and $\activeTimers(\w)$,
	$\lambda(j) \geq d$, for all timers $j \in \activeTimers(\w)$.
	Therefore $D \myxrightarrow{d} D'$, with $D' = (\w, \lambda - d)$.
	It is routine to check that $C'\; R \; D'$.
	\item 
	Suppose  $D \myxrightarrow{d} D'$.
	Then $\lambda(j) \geq d$, for all timers $j \in \activeTimers(\w)$,
	and $D'= (\w, \lambda - d)$.
	Since for each timer $x \in \activeTimers(q)$ that is live after $\pi$, $\kappa(x) = \lambda(\causetwo{\pi}{x})$, and since by Lemma~\ref{lemma cause two},
	$\causetwo{\pi}{\cdot}$ is a bijection between the live timers of $\activeTimers(q)$ and $\activeTimers(\w)$,
	$\kappa(x) \geq d$, for all live timers $x \in \activeTimers(q)$.
	Also for all ghost timers $x \in \activeTimers(q)$,
	$\kappa(x) \geq d$: otherwise we could pick the ghost timer $y$ with minimal value, say $e < d$, and advance time with $e$ and let $y$ timeout.  But by definition of a ghost timer, this is not possible.
	Therefore $C \myxrightarrow{d} C'$, with $C' = (q, \kappa - d)$.
	It is routine to check that $C'\; R \; D'$.
\end{enumerate}
The other cases are similar.
\qed\end{proof}

As a direct corollary of Lemma~\ref{lemma:symbolic tree equivalent}, we conclude that symbolic equivalence implies (timed bisimulation) equivalence.

\begin{corollary}\label{Cy:symEquivalent:equivalent}
  Let \(\M\) and \(\N\) be two \complete \MMTs.
  Then \(\M \symEquivalent \N\) implies \(\M \equivalent \N\).
\end{corollary}

\begin{proof}
Assume \(\M \symEquivalent \N\).
Since the symbolic trees of $\M$ and $\N$ are directly constructed from their symbolic semantics, which are the same by our assumption,
$\symbtree{\M} = \symbtree{\N}$.
By Lemma~\ref{lemma:symbolic tree equivalent}, $\M \equivalent\symbtree{\M}$ and $\N \equivalent\symbtree{\N}$.
Therefore, by Lemma~\ref{lemma: timed bisimulation is an equivalence},
$\M\equivalent\N$.
\qed\end{proof}

\subsection{Timed trace equivalence}
Timed bisimulations preserve inputs, outputs and delays, but abstract all $\timeout{x}$ actions, for $x \in X$, into a single $\mathit{to}$ action.
Analogously, we may associate a timed word over $(I \cup \{ \mathit{to}\}) \times O$ to each timed run $\rho$ by deleting all configurations and abstracting all $\timeout{x}$ actions into $\mathit{to}$.
We call this the \emph{timed trace} of $\rho$, and denote it as $\timedtrace{\rho}$.
We write $\ttlanguage{\M}$ for the set of all timed traces of timed runs of an MMT $\M$ starting from $(q_0^{\M}, \emptyset)$.

\begin{definition}[Timed trace equivalence]\label{def:equivalence:tt}
	Two \MMTs \(\M\) and \(\N\) are \emph{timed trace equivalent}, noted \(\M \ttequivalent \N\), 
	if $\ttlanguage{\M} = \ttlanguage{\N}$.
\end{definition}

Trivially, timed bisimulation implies timed trace equivalence: if $\M$ and $\N$ are bisimilar, then we may construct, for each timed run of $\M$, a corresponding timed run of $\N$ with the same timed trace (and vice versa).

\begin{lemma}\label{lemma:symEquivalent:equivalent}
	Let \(\M\) and \(\N\) be two \MMTs.
	Then \(\M \equivalent \N\) implies \(\M \ttequivalent \N\).
\end{lemma}

Due to the nondeterminism introduced by the $\mathit{to}$-transitions,
timed trace equivalence does not imply timed bisimulation equivalence.  A counterexample can be constructed by slightly adjusting the example of \Cref{app:counterexample}.

\subsection{\Robust MMTs}\label{proof:lemma:queries:raceAvoiding}

In this subsection, we will establish that timed trace equivalence implies symbolic equivalence for the subclass of \emph{\robust} MMTs. 
Whenever there is zero delay between two transitions in a timed run, we say that we have a \emph{race}. 
An \MMT is \robust~\cite{BruyerePSV23} if
every feasible run is the untimed projection of a timed run in which all delays are non-zero and at most one timer may time out in any configuration.

\begin{definition}
An MMT $\M$ is \emph{\robust} if any feasible run
\(\pi = p_0 \xrightarrow{i_1} p_1 \xrightarrow{i_2} \dotsb \xrightarrow{i_n} p_n\)
with \(p_0 = q_0\) is the untimed projection of a \emph{race-free} timed run, i.e., a timed run
\(\rho = (p_0, \emptyset) \xrightarrow{d_1}
(p_0, \emptyset) \xrightarrow{i_1}
(p_1, \valuation_1) \xrightarrow{d_2}
\dotsb \xrightarrow{i_n}
(p_n, \valuation_n) \xrightarrow{d_{n+1}}
(p_n, \valuation_n - d_{n+1})\)
such that:
\begin{itemize}
	\item
	all delays are non-zero: \(d_j > 0\) for any \(j \in \{1, \dotsc, n + 1\}\),
	\item
	at most one timeout is possible at any time:
	in any \((\valuation_j - d_{j+1})\) and \(x \in \activeTimers(p_j)\) with
	\(j \in \{1, \dotsc, n - 1\}\), we have
	\((\valuation_j - d_{j+1})(x) = 0\) if and only if \(i_{j+1} = \timeout{x}\),
	and
	\item
	no timer times out in $\valuation_n - d_{n+1}$:
	$(\valuation_n - d_{n+1})(x) > 0$ for all $x \in \activeTimers(p_n)$.
\end{itemize}
\end{definition}
The notion of \robust machine is introduced in~\cite{BruyerePSV23} with a \THREEEXP{} algorithm to decide whether a machine is \robust.

We are particularly interested in a specific type of timed runs that we call \emph{transparent}, since for these timed runs the symbolic trace is already fully determined by their timed trace.

\begin{definition}
	Let $\tau$ be a timed word over $(I \cup \{ \mathit{to} \}) \times O$.
	The \emph{starting time} of an occurrence of a symbol from $I \cup \{ \mathit{to} \}$ in $\tau$ is the sum of all the delays from $\nnr$ that precede it.
	We say that $\tau$ is \emph{transparent} if, for all occurrences of symbols from $I$, the fractional part of their starting time is different.
	A timed run $\rho$ of an MMT is \emph{transparent} if its timed trace $\timedtrace{\rho}$ is transparent.
\end{definition}

\begin{example}\label{ex:transparent}
Consider the following timed word:
\begin{eqnarray*}
	\tau & = & 1.1 ~ i_1/o_1 ~ 1.1 ~ i_2/o_2 ~ 2.4 ~ i_3/o_3 ~ 0.6 ~ \mathit{to}/o_4 ~ 3.4 ~ \mathit{to}/o_5 ~ 1.6 ~ \mathit{to}/o_6
\end{eqnarray*}
This timed word is transparent since the first input $i_1$ has starting time $1.1$, the second input $i_2$ has starting time $2.2$,
the third input $i_3$ has starting time $4.6$, and the fractional parts of these starting times ($0.1$, $0.2$ and $0.6$, resp.) are all different.
\end{example}

In a timed run, an input may trigger a subsequent timeout, which in turn may trigger another timeout, etc.
Following~\cite{BruyerePSV23}, we refer to the set of an input transition and all the subsequent timeouts that are triggered by it as a \emph{block}.
Since timers may only be set to an integer value, the starting times of all transitions in a block have the same fractional part.
This means that in the transparent timed trace of such a run we may, for each timeout symbol, easily identify the preceding event that triggered it: we just take the most recent preceding event with a starting time that has the same fractional part. We can also compute the value to which the timer was set as the sum of all the time delays between the timeout and the preceding event that triggered it.
Hence we may associate to each transparent timed trace $\tau$ a symbolic input word $\sinpw{\tau}$ by (1) deleting all outputs and time delays, and (2) appending to each timeout the index of the preceding event that started it, and the amount of time (i.e., the sum of all the delays) that has passed since this event.

\begin{example}
	Consider again the timed word $\tau$ from \Cref{ex:transparent}.
	The first timeout in $\tau$ has starting time $5.2$, and is therefore caused by a timer that was set to 2 on the second input.
	The second timeout in $\tau$ has starting time $8.6$, and is thus caused by a timer that was set to 4 on the third input.
	The last timeout in $\tau$ has starting time $10.2$, and is triggered by a timer that was set to 5 on the first timeout. Hence the corresponding symbolic input word $\sinpw{\tau}$ is 
	\begin{eqnarray*}
		\w & = &  i_1  ~ i_2  ~ i_3  ~ \mathit{to}(3,2)  ~ \mathit{to}(4,3)  ~ \mathit{to}(5,4).
	\end{eqnarray*}
\end{example}

The next lemma asserts that the symbolic word extracted from the timed trace of a timed run by the above procedure equals the symbolic word of the corresponding untimed run, as defined in \Cref{sec:MMT:equivalence}.

\begin{lemma}\label{lemma:commuting2}
	Suppose $\rho$ is a transparent timed run of an MMT $\M$.
	Then $\sinpw{\timedtrace{\rho}} = \sinpw{\untimeRun{\rho}}$.
\end{lemma}

In a race-avoiding MMT there exists, for each feasible run, a corresponding transparent timed run.

\begin{lemma}\label{lemma: transparent}
	Suppose $\pi$ is a feasible run of a race-avoiding MMT $\M$.
	Then $\M$ has a transparent, race-free timed run $\rho$ with
	$\untimeRun{\rho} = \pi$.
  Moreover, one can construct \(\rho\) from \(\pi\) in polynomial time.
\end{lemma}
\begin{proof}
  The first part of the lemma, i.e., the existence of a transparent, race-free
  timed run follows easily from the fact that \(\pi\) is feasible and \(\M\)
  is \robust.
  Indeed, as \(\pi\) is feasible, there necessarily exists a timed run \(\rho'\)
  such that \(\untimeRun{\rho'} = \pi\), by definition.
  Furthermore, since \(\M\) is \robust, we can assume that \(\rho'\) is race-free.
  It remains to discuss the transparent aspect.
	Using the properties of a race-free run, and in particular that each input transition is preceded and followed by a nonzero delay, we may slightly adjust the timing of the input transitions in $\rho'$ to ensure that the fractional parts of their starting times are different.
  That is, we can turn the race-free timed run \(\rho'\) into a
  transparent, race-free timed run \(\rho\).

  We now prove the second part of the lemma in a constructive way by arguing that
  \(\pi\) yields a system of linear constraints over the delays of \(\rho\).
  As \(\pi\) is feasible and \(\M\) \robust, a solution will naturally always
  exist and suffices to construct a transparent and race-free timed run.

  Let
  \[
    \pi = p_0 \xrightarrow[u_1]{i_1/o_1}
    p_1 \xrightarrow[u_2]{i_2/o_2}
    \dotsb \xrightarrow[u_n]{i_n/o_n}
    p_n \in \runs{\M}
  \]
  with \(p_0 = q_0\) (i.e., we start from the initial state), and
  \[
    \rho =
    (p_0, \emptyset) \xrightarrow{d_1}
    (p_0, \emptyset) \xrightarrow[u_1]{i_1/o_1}
    (p_1, \valuation_1) \xrightarrow{d_2}
    \dotsb \xrightarrow[u_n]{i_n/o_n}
    (p_n, \valuation_n) \xrightarrow{d_{n+1}}
    (p_n, \valuation_n - d_{n+1})
  \]
  be a transparent, race-free timed run such that \(\untimeRun{\rho} = \pi\).
  We can make the following observations on the delays:
  \begin{itemize}
    \item
    Each delay \(d_j\) is in \(\rplus\), as \(\rho\) is race-free.
    \item
    If the sub-run \(p_{j-1} \xrightarrow{i_{j} \dotsb i_k} p_k\)
    is \(x\)-spanning, i.e., the first transition starts \(x\) at value \(c\),
    \(i_k\) is the timeout of \(x\), and none of the intermediate transitions
    restarts or stops \(x\), then the sum of the delays \(d_{j+1}\) to \(d_k\)
    must be exactly \(c\).
    Indeed, we must let enough time elapse for \(x\) to reach zero.
    Mathematically, \(\sum_{\ell = j+1}^k d_\ell = c\).
    \item
    As \(\rho\) is race-free,
    for all \(j \in \{1, \dotsc, n - 1\}\) and \(x \in \activeTimers(p_j)\),
    \(x\) times out if and only if the next action is \(\timeout{x}\).
    Formally,
    \((\valuation_j - d_{j+1})(x) = 0\) if and only if \(i_{j+1} = \timeout{x}\).
    Likewise, for all \(x \in \activeTimers(p_n)\),
    $(\valuation_n - d_{n+1})(x) \neq 0$.
    \item
    For any \(j\) such that \(u_j = (x, c)\) and there is no \(k > j\) such that
    \(i_k = \timeout{x}\), then either \(x\) is restarted or stopped by some
    transition, or the last action \(i_n\) is read before \(c\) units of time
    elapse, since \(\rho\) is race-free.
    \begin{itemize}
      \item
      In the first case, let \(k > j\) such that \(i_k \neq \timeout{x}\) and
      \(p_{k - 1} \xrightarrow{i_k}\) restarts or stops \(x\).
      Then, the sum of the delays \(d_{j+1}\) to \(d_k\) must be strictly less than
      \(c\), i.e., \(\sum_{\ell = j+1}^k d_\ell < c\).

      \item
      In the second case (so, $x \in \activeTimers(p_n)$ and $x$ does not time out after waiting $d_{n+1}$),
      the sum of the delays \(d_{j+1}\) to \(d_{n+1}\) must be strictly less than \(c\),
      i.e., \(\sum_{\ell = j+1}^{n+1} d_\ell < c\).
    \end{itemize}
    \item
    Since \(\rho\) is transparent, for every input \(i_j \in I\), it must be that
    the fractional part of the sum of the delays up to \(d_j\) is unique.
    Observe that, if the transition reading \(i_j\) starts a timer \(x\), by the
    other constraints, we will immediately obtain that the subsequent timeouts
    of \(x\) will all share the same fractional part, since we start a timer at
    a natural value. It is thus sufficient to focus on the inputs.

    Formally, let us write \(\fractionalPart(c)\) to denote the fractional part
    $c - \lfloor c \rfloor$ of \(c \in \nnr\). Then, for any
    \(j \neq k \in \{1, \dotsc, n\}\) such that \(i_j, i_k \in I\),
    \(\fractionalPart(\sum_{\ell = 0}^j d_\ell) \neq
    \fractionalPart(\sum_{\ell = 0}^k d_\ell)\).
  \end{itemize}
  Observe that these constraints are all linear, except the last.
  Moreover, if we consider the delays \(d_j\) as \emph{variables}, one can still
  gather the constraints and use them to find a value for each \(d_j\).
  We denote by \(\constraints(\pi)\) the set of constraints for \(\pi\) over the
  variables representing the delays.
  By the arguments given above, a solution always exists, whenever \(\pi\) is
  feasible and \(\M\) is \robust.
  That is, given \(\pi\), one can compute appropriate delays and construct a
  transparent, race-free timed run \(\rho\) such that \(\untimeRun{\rho} = \pi\).

  Finally, it remains to explain how to compute a solution in polynomial time.
  One can build a system of linear constraints from the first four items above.
  Using known results from linear programming~\cite{Jiang0WZ21}, a solution can
  be found in polynomial time. Since this yields a timed run that is race-free,
  we can slightly adjust the delays before each input to satisfy the last item.
  This can also be performed in polynomial time in the length of \(\pi\).
\qed\end{proof}

To each timed word $\tau$ we may associate a sequence of outputs $\outw{\tau}$ by just discarding all time delays and inputs from $I \cup \{\mathit{to}\}$. We then obtain the following (trivial) variant of \Cref{lemma:commuting2}:

\begin{lemma}\label{lemma:commuting1}
	Suppose $\rho$ is a timed run of an MMT $\M$.
	Then $\outw{\timedtrace{\rho}} = \outw{\untimeRun{\rho}}$.
\end{lemma}

We are now prepared to prove the main result of this subsection.

\begin{lemma}\label{lemma: timed equivalence implies symbolic equivalence}
Let \(\M\) and \(\N\) be two \robust \MMTs.
Then \(\M \ttequivalent \N\) implies \(\M \symEquivalent \N\).	
\end{lemma}
\begin{proof}
Suppose \(\M \ttequivalent \N\).
Suppose further that $\w \in \silanguage{\M}$.
Then $\M$ has a feasible run $\pi$ that starts in $q_0^{\M}$ with $\sinpw{\pi} = \symbolic{w}$.
By \Cref{lemma: transparent}, $\M$ has a transparent, race-free timed run $\rho$ with $\untimeRun{\rho} = \pi$.
Let $\tau = \timedtrace{\rho}$.
Then $\tau$ is transparent and $\tau \in \ttlanguage{\M}$.
Since \(\M \ttequivalent \N\), also $\tau \in \ttlanguage{\N}$.
Let $\rho'$ be a timed run of $\N$ that starts in $q_0^{\N}$ with $\tau = \timedtrace{\rho'}$.
Then, since $\tau$ is transparent, $\rho'$ is transparent as well.
Let $\pi'= \untimeRun{\rho'}$. Then $\pi'$ is a feasible run of $\N$ that starts in $q_0^{\N}$. Using \Cref{lemma:commuting2} twice, we derive
\begin{eqnarray*}
	\sinpw{\pi'} & = & \sinpw{\untimeRun{\rho'}} = \sinpw{\timedtrace{\rho'}} = \sinpw{\tau} = \sinpw{\timedtrace{\rho}} = \sinpw{\untimeRun{\rho}} = \sinpw{\pi} = \w.
\end{eqnarray*}
Hence $\w \in \silanguage{\N}$. Using \Cref{lemma:commuting1} twice, we derive
\begin{multline*}
	 \ofunction{\N}{\w} = \outw{\pi'} = \outw{\untimeRun{\rho'}} = \outw{\timedtrace{\rho'}} = \outw{\tau} = \outw{\timedtrace{\rho}} = \outw{\untimeRun{\rho}} = \outw{\pi}
	 \\ =\ofunction{\M}{\w}.
\end{multline*}
Since $\w$ was chosen arbitrarily, we conclude
$\silanguage{\M} \subseteq \silanguage{\N}$.
By a symmetric argument, we can prove
$\silanguage{\N} \subseteq \silanguage{\M}$.
Hence $\silanguage{\M} = \silanguage{\N}$, as required.
Moreover, for each $\w \in \silanguage{\M}$, $\ofunction{\M}{\w} = \ofunction{\N}{\w}$.
Hence \(\M \symEquivalent \N\).
\end{proof}

The final lemma of this section is important for the definition of concrete equivalence queries: if the hypothesis and the SUT are not timed equivalent, then a helpful teacher may always provide a transparent timed trace as witness of the inequivalence.

\begin{lemma}\label{lemma: timed inequivalence implies transparent witness}
	Let \(\M\) and \(\N\) be two \robust \MMTs.
	If \(\M \notttequivalent \N\) then there exists a transparent timed trace $\tau$ in the symmetric difference of 	$\ttlanguage{\M}$ and $\ttlanguage{\N}$.
\end{lemma}
\begin{proof}
	The proof is by contradiction.
	Assume that \(\M \notttequivalent \N\) and all transparent timed traces of $\M$ and $\N$ are contained in $\ttlanguage{\M} \cap \ttlanguage{\N}$.
	Since \(\M \notttequivalent \N\) there exists a (nontransparent) timed trace $\tau$ in the symmetric difference of 	$\ttlanguage{\M}$ and $\ttlanguage{\N}$.
	Without loss of generality (due to symmetry), we may assume $\tau \in \ttlanguage{\M} \setminus \ttlanguage{\N}$.
	Let $\rho$ be a timed run with $\tau = \timedtrace{\rho}$, and let $\pi = \untimeRun{\rho}$ be the feasible run of $\M$ associated to $\rho$.
	Since $\M$ is race-avoiding, there exists (by \Cref{lemma: transparent}) a transparent timed run $\rho'$ of $\M$ with $\untimeRun{\rho'} = \pi$. Let $\tau' = \timedtrace{\rho'}$.
	Then, by our assumption, $\tau' \in \ttlanguage{\N}$.
	Let $\rho''$ be a timed run with $\tau' = \timedtrace{\rho''}$, and let $\pi' = \untimeRun{\rho''}$ be the feasible run of $\N$ associated to $\rho''$.
	Observe that runs $\pi$ and $\pi'$ are rather similar since by \Cref{lemma:commuting2} $\sinpw{\pi} = \sinpw{\tau'} = \sinpw{\pi'}$, and by \Cref{lemma:commuting1} $\outw{\pi} = \outw{\tau'} = \outw{\pi'}$. So $\pi$ and $\pi'$ have exactly the same inputs, the same outputs, and the same timeouts with the same causes.
	
	Now suppose that $\pi$ has a prefix $\pi_0$ that ends with a  transition that starts a timer that does not expire in $\pi$ but expires in another feasible run $\pi_1$  of $\M$ that extends $\pi_0$.
	Then, using again that all transparent timed traces of $\M$ are contained in $\ttlanguage{\M} \cap \ttlanguage{\N}$, we may infer that the corresponding transition of $\pi'$ sets a timer (which does not expire in $\pi'$) to exactly the same value as $\pi$.
	In fact, if the timer started at the end of $\pi_0$ remains active until some point $\pi_2$ with $\pi_0 \leq \pi_2 \leq \pi$ (in the sense that that timer expires in a feasible run $\pi_3$ of $\M$ that extends $\pi_2$) then we may show that the corresponding timer in $\pi'$ remains active just as long.
	Symmetrically, we may show that each timer that is started in $\N$ by the final transition of some prefix $\pi'_0$ of $\pi'$ and expires in some feasible extension $\pi'_1$ of $\pi'_0$, is set to the same value by the corresponding transition in $\M$ of $\pi$, and remains alive equally long.
	
	Thus, $\pi$ and $\pi'$ impose the same constraints on the
	variables representing the delays, and even though
 \(\constraints(\pi)\) and \(\constraints(\pi')\) may be different, they allow for timed runs with exactly the same timed traces.
 In particular, $\tau$ is a timed trace of $\N$. 
 This contradiction completes the proof of the lemma.
	
\end{proof}
 
\section{Proof of \texorpdfstring{\Cref{lemma:goodMMT}}{Lemma 1}}\label{proof:lemma:gootMMT}

\NewDocumentCommand{\downwardZ}{m}{{#1}\!\downarrow}
\NewDocumentCommand{\restrictionZ}{m m}{{#1} \lceil {#2}}
\NewDocumentCommand{\assignmentZ}{m m}{{#1}[#2]}
\NewDocumentCommand{\timeoutZ}{m m}{\timeout{{#1}, {#2}}}
\NewDocumentCommand{\zoneMMT}{}{\automaton{Z}}

\goodMMT*

In order to prove this lemma, we first properly adapt the notion of \emph{zones}
from timed automata~\cite{Alur99,HandbookModelChecking} to \MMTs.\footnote{The concept of \emph{region} from timed automata was adapted to automata with timers in~\cite{BruyerePSV23}.}
Given a \complete \MMT \(\M\), we show that
the \MMT constructed from the reachable zones of \(\M\) is \good.

\subsection{Zones}

Let \(X\) be a set of timers.
A \emph{zone \(Z\) over \(X\)} is a set of valuations over \(X\), i.e.,
\(Z \subseteq \Val{X}\), described by the following grammar:
\[
  \phi =
    x < c \mid
    x \leq c \mid
    c < x \mid
    c \leq x \mid
    x - y < c \mid
    x - y \leq c \mid
    \phi_1 \land \phi_2
\]
with \(x, y \in X\) and \(c \in \nat\).
It may be that a zone is empty.
For the particular case \(X = \emptyset\), we have that \(\Val{X} = \{\emptyset\}\), meaning that a zone $Z$ is either the zone \(\{\emptyset\}\) or the empty zone.

Given a zone \(Z\) over \(X\), a set \(Y \subseteq X\), a timer \(x\) (that does
not necessarily belong to \(X\)), and a constant \(c \in \natplus\), we define
the following operations:
\begin{itemize}
  \item
  The \emph{downward closure} of \(Z\) where we let some time elapsed in all
  valuations of \(Z\).
  That is, we obtain all valuations that can be reached from \(Z\) by 
  waiting (delays cannot exceed the smallest value to avoid
  going below zero):
  \[
    \downwardZ{Z} = \{\valuation - d \mid
      \valuation \in Z, d \leq \min_{y \in \dom{\valuation}} \valuation(y)
    \}.
  \]
  \item
  The \emph{restriction} of \(Z\) to \(Y \neq \emptyset\) where,
  for every valuation of \(Z\),
  we discard values associated with timers in \(X \setminus Y\), i.e., we only
  keep the timers that are in \(Y\):
  \[
    \restrictionZ{Z}{Y} = \{
      \valuation' \in \Val{Y}
      \mid
      \exists \valuation \in Z, \forall y \in Y : \valuation'(y) = \valuation(y)
    \}.
  \]
  If \(Y\) is empty, then we define the restriction as:
  \[
    \restrictionZ{Z}{\emptyset} = \begin{cases}
      \emptyset & \text{if \(Z = \emptyset\)}\\
      \{\emptyset\} & \text{otherwise.}
    \end{cases}
  \]
  \item
  The \emph{assignment} of \(x\) to \(c\) in the zone $Z$ over $X$.
  Either \(x\) is already in \(X\) in which case we simply overwrite the value
  of \(x\) by \(c\), or \(x\) is not in \(X\) in which case we \enquote{extend}
  the valuations of \(Z\) by adding \(x\):
  \[
    \assignmentZ{Z}{x = c} = \{
      \valuation' \in \Val{X \cup \{x\}}
      \mid
      \valuation'(x) = c \land
      \exists \valuation \in Z, \forall y \in X \setminus \{x\} :
      \valuation'(y) = \valuation(y)
    \}.
  \]
  \item
  The \emph{timeout of \(x\)} in \(Z\) where we keep valuations of \(Z\) s.t.\ \(x \in X\) times out:
  \[
    \timeoutZ{Z}{x} = \{
      \valuation \in Z \mid \valuation(x) = 0
    \}.
  \]
\end{itemize}
If $Z$ is a zone, then \(\downwardZ{Z}\), \(\restrictionZ{Z}{Y}\), \(\assignmentZ{Z}{x = c}\), and \(\timeoutZ{Z}{x}\) are again zones~\cite{HandbookModelChecking}.
Observe that \(\downwardZ{Z}\) and \(\timeoutZ{Z}{x}\) are zones over \(X\)
(when \(x \in X\)), \(\restrictionZ{Z}{Y}\) is a zone over \(Y\), and
\(\assignmentZ{Z}{x = c}\) is a zone over \(X \cup \{x\}\).

\subsection{Zone \MMT}

Given a \complete \MMT \(\M\), we
explain how to construct its \emph{zone \MMT} that we denote \(\zoneOf{\M}\).
The states of $\zoneOf{\M}$ are pairs $(q,Z)$ where $q$ is a state of $\M$ and
$Z$ is a zone included in $\Val{\chi(q)}$.
The idea to construct $\zoneOf{\M}$ is to start from the pair
\((q_0^\M, \{\emptyset\})\) and explore every outgoing transition of \(q_0^\M\).
In general, we want to define the outgoing transitions of the current pair \((q, Z)\).
To do so, we consider the outgoing transitions of \(q\) in the \complete machine $\M$.
For every \(q \xrightarrow[u]{i} q'\) with \(i \in I\), we reproduce the same
transition in \(\zoneOf{\M}\) (as it is always possible to trigger an
input transition).
That is, we define \((q, Z) \xrightarrow[u]{i} (q', Z')\) with a zone \(Z'\) that
depends on the update: if \(u = \bot\), then \(Z'\) is obtained by the restriction
of \(Z\) to the active timers of \(q'\); if \(u = (x, c)\), then we first assign
\(x\) to \(c\).
In both cases, we also let time elapse, i.e., we always compute the downward
closure.
Finally, we perform the same idea with every \(q \xrightarrow{\timeout{x}}\) such that $x \in \activeTimers(q)$,
except that we first only consider the valuations of \(Z\) where \(x\) is zero.
If \(\timeoutZ{Z}{x}\) is empty, we do not define the transition.
Hence, in this way, we construct the states $(q,Z)$ of $\zoneOf{\M}$ such that
$Z \neq \emptyset$ and that are reachable from its initial state \((q_0^\M, \{\emptyset\})\).

More formally, let \(\M = (X^\M, Q^\M, q_0^\M, \activeTimers^\M, \delta^\M)\) be
a \complete \MMT.
We define the MMT
\(\zoneMMT = (X^\zoneMMT, Q^\zoneMMT, q_0^\zoneMMT, \activeTimers^\zoneMMT,
\delta^\zoneMMT)\) with:
\begin{itemize}
  \item
  \(X^\zoneMMT = X^\M\),
  \item
  \(Q^\zoneMMT = \{(q, Z) \mid q \in Q^\M, Z \subseteq \Val{\chi^\M(q)} \land
    Z \neq \emptyset\}\),
  \item
  \(q_0^\zoneMMT = (q_0^\M, \{\emptyset\})\),
  \item
  For any \((q, Z) \in Q^\zoneMMT\), we define
  \(\activeTimers^\zoneMMT((q, Z)) = \activeTimers^\M(q)\), i.e., we simply copy
  the active timers of \(q\),
  \item
  Let \((q, Z) \in Q^\zoneMMT\) and \(q \xrightarrow[u]{i/o} q'\) be a transition
  of \(\M\).
  We define
  \[
    Z' = \begin{cases}
      Z & \text{if \(u = \bot\) and \(i \in I\)}
      \\
      \assignmentZ{Z}{x = c} & \text{if \(u = (x, c)\) and \(i \in I\)}
      \\
      \timeoutZ{Z}{x} & \text{if \(u = \bot\) and \(i = \timeout{x}\)}
      \\
      \assignmentZ{\left(\timeoutZ{Z}{x}\right)}{x = c} &
        \text{if \(u = (x, c)\) and \(i = \timeout{x}\).}
    \end{cases}
  \]
  Then, if \(Z' \neq \emptyset\), we restrict \(Z'\) to \(\activeTimers^\M(q')\)
  and let time elapse.
  That is, we define
  \[
    \delta^\zoneMMT((q, Z), i) =
    \left(
      \left(
        q',
        \downwardZ{\left(\restrictionZ{Z'}{\activeTimers^\M(q')}\right)}
      \right), o, u
    \right).
  \]
\end{itemize}
The \MMT $\zoneOf{\M}$ is then the MMT $\zoneMMT$ restricted to its reachable
states.
Observe that the set of actions of \(\zoneOf{\M}\) is the set of
actions of \(\M\), i.e., \(\actions{\zoneOf{\M}} = \actions{\M}\).

Let us argue that $\zoneOf{\M}$ has finitely many states and is
well-formed.
We will prove later that it is also \complete.

\begin{lemma}\label{lem:soundZone}
Let $\M$ be a complete \MMT.
Then, $\zoneOf{\M}$ has finitely many states and is well-formed.
\end{lemma}
\begin{proof}
The zone \MMT is clearly well-formed because its transitions mimics the
transitions of $\M$ and for any \((q, Z) \in Q^\zoneMMT\), we have
\(\activeTimers^\zoneMMT((q, Z)) = \activeTimers^\M(q)\).

By construction, the states $(q,Z)$ of $\zoneOf{\M}$ are such that $Z$ is a zone
over $\chi^\M(q)$, described as a finite conjunction of constraints of the
shape \(x \bowtie c\) or \(x - y \bowtie c\), with \(\bowtie \in
\{<, \leq, \geq, >\}\) and \(c \in \nat\).
For each timer \(x\), let \(c_x\) be the maximal constant appearing on an
update (re)starting \(x\).
Since the value of a timer can only decrease, it is clear that we will never
reach a zone where \(x > c_x\).
Moreover, as the value of a timer must remain at least zero at any time, we also
have a lower bound.
In other words, we know that each timer \(x\) will always be confined between
zero and \(c_x\).
From the shape of the constraints and these bounds, it follows immediately that
there are finitely many zones.
Hence, $\zoneOf{\M}$ has finitely many states.
\qed\end{proof}

We now prove the required properties to show \Cref{lemma:goodMMT} where the announced \MMT $\N$ is the zone \MMT of $\M$:
\begin{itemize}
  \item
  Both \MMTs $\M$ and $\zoneOf{\M}$ can read the same timed words.
  That is, for any state \(q \in Q^\M\), it holds that \((q_0^\M, \emptyset)
  \xrightarrow{w} (q, \valuation)\) if and only if
  \(((q_0^{\M}, \{\emptyset\}), \emptyset) \xrightarrow{w} ((q, Z), \valuation)\) for
  some zone \(Z\).
  See \Cref{prop:zoneMMT:timedruns}.

  \item
  \(\zoneOf{\M}\) is \complete, by \Cref{prop:zoneMMT:complete}.

  \item
  A run reading \(w\) in \(\M\) is feasible if and only if the run reading \(w\)
  in \(\zoneOf{\M}\) is also feasible, by
  \Cref{cor:zoneMMT:feasible:iff}.
  
  \item
  \(\M\) and \(\zoneOf{\M}\) are symbolically equivalent, by
  \Cref{prop:zoneMMT:equivalent}.

  \item
  Any run of \(\zoneOf{\M}\) is feasible, by \Cref{prop:zones:feasible}.
\end{itemize}

\begin{proposition}\label{prop:zoneMMT:timedruns}
  Let $\zoneOf{\M}$ be the zone MMT of some \complete \MMT \(\M\).
  Then, for every state \(q \in Q^\M\), valuation
  \(\valuation \in \Val{\activeTimers^\M(q)}\), and timed word \(w\),
  \[
    (q_0^\M, \emptyset) \xrightarrow{w} (q, \valuation) \text{ in } \M
    \quad \iff \quad
    ((q_0^\M, \{\emptyset\}), \emptyset) \xrightarrow{w} ((q, Z), \valuation) \text{ in } \zoneOf{\M}
  \]
  for some zone \(Z\) over \(\activeTimers^\M(q)\) such that \(\valuation \in Z\).
\end{proposition}
\begin{proof}
  We focus on the \(\implies\) direction.
  The other direction can be obtained with similar arguments.
  Let \(q \in Q^\M, \valuation \in \Val{\activeTimers^\M(q)}\), and \(w\) be a timed word
  such that \((q_0^\M, \emptyset) \xrightarrow{w} (q, \valuation)\).
  We show that there exists a zone \(Z\) over \(\activeTimers^\M(q)\) such that
  \(\valuation \in Z\) and \(((q_0^\M, \{\emptyset\}), \emptyset) \xrightarrow{w}
  ((q, Z), \valuation)\).
  We proceed by induction over the length of \(w\).

  \textbf{Base case:} \(\lengthOf{w} = 0\), i.e., \(w = d\) with \(d \in \nnr\).
  Since no timer is active, it is clear that we have the runs
  \begin{align*}
    (q_0^\M, \emptyset) \xrightarrow{d} (q_0^\M, \emptyset)
    &&\text{and}&&
    ((q_0^\M, \{\emptyset\}), \emptyset) \xrightarrow{d}
    ((q_0^\M, \{\emptyset\}), \emptyset),
  \end{align*}
  and \(\emptyset \in \{\emptyset\}\).

  \textbf{Induction step:} let \(k \in \nat\) and assume the implication is
  true for every timed word of length \(k\).
  Let \(w = w' \cdot i \cdot d\) of length \(k + 1\), i.e., \(\lengthOf{w'} = k\),
  \(i \in \actions{\M}\), and \(d \in \nnr\).
  Then, we have
  \[
    (q_0^\M, \emptyset) \xrightarrow{w'}
    (p, \lambda) \xrightarrow[u]{i}
    (q, \valuation) \xrightarrow{d}
    (q, \valuation - d).
  \]
  This implies that
  \(d \leq \min_{y \in \activeTimers^\M(q)} \valuation(y)\).\footnote{We recall
  that $\min_{y \in \activeTimers^\M(q)} \valuation(y) = +\infty$ when
  $\activeTimers^\M(q) = \emptyset$.}
  By induction hypothesis, we have that
  \[
    ((q_0^\M, \{\emptyset\}), \emptyset) \xrightarrow{w'} ((p, Z_p), \lambda)
  \]
  such that \(\lambda \in Z_p\).
  It is then sufficient to show that we have
  \[
    ((p, Z_p), \lambda) \xrightarrow[u]{i}
    ((q, Z), \valuation) \xrightarrow{d}
    ((q, Z), \valuation - d)
  \]
  with $\valuation - d \in Z$.

  By construction of the zone MMT and as \(p \xrightarrow{i} q\) is defined in
  \(\M\), the \(i\)-transition from \((p, Z_p)\) to
  \((q, Z)\) is defined if and only if \(Z\) is not empty and
  \[
    Z = \begin{cases}
      \downwardZ{\left(
        \restrictionZ{Z_p}{\activeTimers^\M(q)}
      \right)}
      & \text{if \(i \in I\) and \(u = \bot\)}
      \\
      \downwardZ{\left(
        \restrictionZ{\left(
          \assignmentZ{Z_p}{x = c}
        \right)}{\activeTimers^\M(q)}
      \right)}
      & \text{if \(i \in I\) and \(u = (x, c)\)}
      \\
      \downwardZ{\left(
        \restrictionZ{\left(
          \timeoutZ{Z_p}{x}
        \right)}{\activeTimers^\M(q)}
      \right)}
      & \text{if \(i = \timeout{x}\) and \(u = \bot\)}
      \\
      \downwardZ{\left(
        \restrictionZ{\left(
          \assignmentZ{\left(
            \timeoutZ{Z_p}{x}
          \right)}{x = c}
        \right)}{\activeTimers^\M(q)}
      \right)}
      & \text{if \(i = \timeout{x}\) and \(u = (x, c)\).}
    \end{cases}
  \]
  Since \(Z_p\) is not empty (as \(\lambda \in Z_p\)) and
  \((p, \lambda) \xrightarrow{i}\) can be triggered
  (meaning that \(\lambda(x) = 0\) if \(i = \timeout{x}\)), we have that
  \(Z\) is also not empty.
  Hence, \(((p, Z_p), \lambda) \xrightarrow{i} ((q, Z), \valuation)\) is
  well-defined and can be triggered.

  Let us show that \(\valuation \in Z\).
  By definition of a timed run, \(\valuation \in \Val{\activeTimers^\M(q)}\).
  Moreover, \(Z\) is a zone over \(\activeTimers^\M(q)\).
  We know that \(\lambda \in Z_p\) and \(\valuation\) is constructed from
  \(\lambda\) by discarding the values for timers that are stopped by the
  discrete transition and, maybe, (re)starting a timer.
  Since \(Z\) is constructed using the same operations, it follows that
  \(\valuation \in Z\).

  Finally, we process the delay \(d\).
  We already know that \(d \leq \min_{y \in \activeTimers^\M(q)} \valuation(y)\).
  Hence, it is feasible to wait \(d\) units of time from \(((q, Z), \valuation)\).
  Moreover, as \(Z\) is already its downward closure, we still have that
  \(\valuation - d \in Z\).
  
  We get the implication.
  One can show the other direction using similar arguments, by definition
  of $\zoneOf{\M}$.
\qed\end{proof}

From the (proof of the) previous proposition, we can easily obtain
that \(\zoneOf{\M}\) is \complete, if \(\M\) is \complete,
Indeed, the construction of \(\zoneOf{\M}\) immediately copies input-transitions,
while timeout-transitions are reproduced when the timer is enabled in the state,
meaning that \(\zoneOf{\M}\) is \complete.

\begin{corollary}\label{prop:zoneMMT:complete}
  Let \(\M\) be a \complete \MMT and \(\zoneOf{\M}\) be its zone \MMT.
  Then, \(\zoneOf{\M}\) is \complete.
\end{corollary}

From the previous proposition, we also conclude that any feasible run of \(\M\)
can be reproduced in \(\zoneOf{\M}\) and vice-versa.
Recall that \(\actions{\M} = \actions{\zoneOf{\M}}\).

\begin{corollary}\label{cor:zoneMMT:feasible:iff}
  Let \(\M\) be a \complete \MMT and \(\zoneOf{\M}\) be its zone \MMT.
  Then, for all words \(w \in \actions{\M}^*\), the following items are equivalent
  \begin{itemize}
    \item
    \(q_0^\M \xrightarrow{w} q\) is in \(\runs{\M}\) and is feasible
    \item
    \((q_0^\M, \{\emptyset\}) \xrightarrow{w} (q, Z)\) is in
    \(\runs{\zoneOf{\M}}\) and is feasible for some zone \(Z\).
  \end{itemize}
\end{corollary}
\begin{proof}
  Let \(q_0^\M \xrightarrow{w} q\) be a feasible run of \(\M\).
  Then, there exists a timed run \((q_0^\M, \emptyset) \xrightarrow{v}
  (q, \valuation)\) of \(\M\) such that \(w\) and \(v\) use the same actions.
  By \Cref{prop:zoneMMT:timedruns}, it follows that
  \(((q_0^\M, \{\emptyset\}), \emptyset) \xrightarrow{v} ((q, Z), \valuation)\)
  is a timed run of \(\zoneOf{\M}\)
  for some zone \(Z\) such that \(\valuation \in Z\).
  As \(v\) and \(w\) use the same actions, the run
  \((q_0^\M, \{\emptyset\}) \xrightarrow{w} (q, Z)\) is a feasible run of
  \(\zoneOf{\M}\).
  The other direction holds with similar arguments.
\qed\end{proof}

Let us now move towards proving that \(\M \symEquivalent \zoneOf{\M}\).

\begin{proposition}\label{prop:zoneMMT:equivalent}
  Let \(\M\) be a \complete \MMT.
  Then, \(\M \symEquivalent \zoneOf{\M}\).
\end{proposition}
\begin{proof}
  We have to show that for every symbolic word
  \(\symbolic{i_1} \dotsb \symbolic{i_n}\) over \(I \cup \toevents{\natplus}\):
  \begin{itemize}
    \item  $q_0^\M \xrightarrow[u_1]{\symbolic{i_1}/o_1} q_1
    \dotsb \xrightarrow[u_n]{\symbolic{i_n}/o_n} q_n$ is a feasible run in $\M$
    iff $q_0^{\zoneOf{\M}} \xrightarrow[u'_1]{\symbolic{i_1}/o'_1}
    q'_1 \dotsb \xrightarrow[u'_n]{\symbolic{i_n}/o'_n} q'_n$ is a feasible run
    in $\zoneOf{\M}$.
    \item Moreover,
    \begin{itemize}
        \item $o_j = o'_j$ for all $j \in \{1, \dotsc, n\}$, and
        \item $q_j \xrightarrow{\symbolic{i_j}} \dotsb \xrightarrow{\symbolic{i_k}}
        q_k$ is spanning $\implies u_j = (x,c) \land u'_j = (x',c') \land c = c'$.
    \end{itemize}
  \end{itemize}

Let \(\symbolic{w} = \symbolic{i_1} \dotsb \symbolic{i_n}\) be a symbolic word such that
  $q_0^\M \xrightarrow{\symbolic{w}} q_n$ is a feasible run of $\M$.
  Hence, there exists $w = i_1 \dotsb i_n$ such that $\toSymbolic{w} = \symbolic{w}$ and $q_0^\M \xrightarrow[u_1]{i_1/o_1} q_1 \dotsb \xrightarrow[u_n]{i_n/o_n} q_n$ is a feasible run of $\M$. By \Cref{cor:zoneMMT:feasible:iff}, it follows that
  $(q_0^{\M}, \{\emptyset\}) \xrightarrow[u'_1]{i_1/o'_1}
  (q_1, Z_1) \dotsb \xrightarrow[u'_n]{i_n/o'_n} (q_n, Z_n)$ is a
  feasible run of $\zoneOf{\M}$.
  By construction of \(\zoneOf{\M}\), we have that \(o_j = o'_j\)
  and \(u_j = u'_j\) for every \(j\). Therefore, $\toSymbolic{i'_1 \dotsb i'_n} = \symbolic{w}$ and $(q_0^{\M}, \{\emptyset\}) \xrightarrow{\symbolic{w}} (q_n, Z_n)$ is a feasible run of $\M$.
  Hence, the direction from \(\M\) to \(\zoneOf{\M}\) holds.
  The other direction follows with the same arguments.
  We thus conclude that \(\M \symEquivalent \zoneOf{\M}\).
\qed\end{proof}

Finally, we show that any run of $\zoneOf{\M}$ is feasible.

\begin{proposition}\label{prop:zones:feasible}
  Let \(\M\) be a complete \MMT.
  Then, any run of \(\zoneOf{\M}\) is feasible.
\end{proposition}
\begin{proof}
  As all states of $\zoneOf{\M}$ are reachable, we can restrict the proof to runs starting at the initial state of $\zoneOf{\M}$. Let us prove that for any run \(\pi = (q_0^\M, \{\emptyset\}) \xrightarrow{w} (q, Z)\) of $\zoneOf{\M}$, for any $\valuation \in Z$, there exists a timed run
  \[
    \rho = ((q_0^\M, \{\emptyset\}), \emptyset) \xrightarrow{v}
    ((q, Z), \valuation)
  \]
  such that \(\untimeRun{\rho} = \pi\). We prove this property by induction over \(n = |w|\).

  \textbf{Base case:} \(n = 0\), i.e., \(w = \emptyword\).
  Let \(\pi = (q_0^\M, \{\emptyset\}) \xrightarrow{\emptyword}
  (q_0^\M, \{\emptyset\})\).
  It is clear that there exists \(\rho = ((q_0^\M, \{\emptyset\}), \emptyset) \xrightarrow{d}
  ((q_0^\M, \{\emptyset\}), \emptyset)\) and \(\emptyset \in \{\emptyset\}\) for
  any \(d \in \nnr\). And we have $\untimeRun{\rho} = \pi$.

  \textbf{Induction step:}
  Let \(k \in \nat\) and assume the proposition holds for every word of
  length \(k\).
  Let \(w\) of length \(k + 1\), i.e., we can decompose \(w = w' \cdot i\)
  with \(i \in \actions{\M}\) and \(\lengthOf{w'} = k\).
  We show that, if \(\pi = (q_0^\M, \{\emptyset\}) \xrightarrow{w} (q, Z)\) is a
  run and $\valuation$ is a valuation in $Z$, there exists a timed run \(\rho =
  ((q_0^\M, \{\emptyset\}), \emptyset) \xrightarrow{v} ((q, Z), \valuation)\)
  such that \(\untimeRun{\rho} = \pi\).

  Assume that \(\pi\) is a run of $\zoneOf{\M}$ and
  let \(\pi' = (q_0^\M, \{\emptyset\}) \xrightarrow{w'} (p, Z_p)\).
  Observe that \(\pi'\) is a sub-run of \(\pi\).
  By the induction hypothesis, we know that for any \(\lambda \in Z_p\), there exists a timed run
  \(\rho' = ((q_0^\M, \{\emptyset\}), \emptyset) \xrightarrow{v'} ((p, Z_p), \lambda)\)
  such that \(\untimeRun{\rho'} = \pi'\).
  Hence, let us focus on the last transition $(p,Z_p) \xrightarrow{i} (q,Z)$ of $\pi$. By construction of $\zoneOf{\M}$,
  given $(q,Z)$ and $\valuation \in Z$, we deduce that there exists $d \in \nnr$ and $\lambda \in Z_p$ such that
   \[
    ((p, Z_p), \lambda) \xrightarrow{i}
    ((q, Z), \valuation + d) \xrightarrow{d}
    ((q, Z), \valuation).
  \]
  Thus, by the induction hypothesis with this $\lambda$, it follows that we have the timed run
  \[
    \rho = ((q_0^\M, \{\emptyset\}), \emptyset) \xrightarrow{v'}
    ((p, Z_p), \lambda) \xrightarrow{i}
    ((q, Z), \valuation+ d) \xrightarrow{d}
    ((q, Z), \valuation).
  \]
  such that \(\untimeRun{\rho} = \pi\).
\qed\end{proof}

\subsection{Proof of \texorpdfstring{\Cref{lemma:goodMMT}}{Lemma 1}}

We are now ready to prove \Cref{lemma:goodMMT} which we repeat one more time.

\goodMMT*

\begin{proof}
  Let \(\M\) be a \complete \MMT and
  \(\zoneOf{\M}\) be its zone \MMT.
  By \Cref{prop:zoneMMT:complete} and
  \Cref{prop:zones:feasible,prop:zoneMMT:equivalent},
  \(\zoneOf{\M}\) is \complete, any run of \(\zoneOf{\M}\) is
  feasible, and \(\M \symEquivalent \zoneOf{\M}\).
  Hence, \(\zoneOf{\M}\) is \good and satisfies the lemma.
\qed\end{proof}
 
\section{Proof of \texorpdfstring{\Cref{lemma:queries}}{Lemma 2}}\label{proof:lemma:queries}

\symbolicQueries*

We first define the concrete output and
equivalence queries in \Cref{app:concrete queries}.
Next, in \Cref{proof:lemma:queries:queries}, we prove the lemma. 

\subsection{Concrete queries}\label{app:concrete queries}
The \emph{concrete} queries are defined in terms of the timed semantics of the \MMT model, and are an adaptation of the queries used for Mealy machines~\cite{VaandragerGRW22,ShahbazG09}:
one type of queries requests the output in response to a given input, and another asks whether a hypothesis is correct.

In the case of a timed system $\M$, it is natural to describe the inputs provided to the system as a \emph{timed input word} $w$, i.e., a timed word over the alphabet $I$.  A learner then observes the timed trace of a timed run of $\M$ triggered by $w$.
To each timed run $\rho$ of $\M$, we may associate a unique timed input words $\tinpw{\rho}$ by just keeping the inputs and the (accumulated) delays between them.
For instance, the following run of the \MMT of \Cref{fig:ex:MMT:good}:
\begin{eqnarray*}
\rho & = & (q_0, \emptyset){\myxrightarrow{0.5}} (q_0, \emptyset)
	{\myxrightarrow[(x, 2)]{i/o}} (q_1, x {=} 2)
	{\myxrightarrow{1}} (q_1, x {=} 1)
	{\myxrightarrow[(y, 3)]{i/o'}} (q_2, x {=} 1, y {=} 3)
{\myxrightarrow{1}} (q_2, x {=} 0, y {=} 2)\\
&	&\myxrightarrow[(x, 2)]{\timeout{x}/o} (q_3, x {=} y {=} 2)
	\myxrightarrow{2} (q_3, x {=} y {=} 0)
	\myxrightarrow[\bot]{\timeout{y}/o} (q_0, \emptyset)
	\myxrightarrow{0} (q_0, \emptyset)
\end{eqnarray*}
has an associated timed input word $\tinpw{\rho} = 0.5 ~ i ~ 1 ~ i ~ 3$.
In fact, we may derive this timed input word of $\rho$ from its timed trace
\begin{eqnarray*}
	\timedtrace{\rho} & = & 0.5 ~ i/o ~ 1 ~ i/o' ~ 1 ~ \mathit{to}/o ~ 2 ~ \mathit{to}/o ~ 0
\end{eqnarray*}
by omitting outputs, $\mathit{to}$'s, and adding consecutive delays.
Note that $\tinpw{\rho} = \tinpw{\timedtrace{\rho}}$.

In general, for a given timed input word $w$, there may be multiple timed runs $\rho$ with $\tinpw{\rho} = w$. According to the next lemma, there will always at least one such timed run. So unlike timed automata~\cite{GomezB07}, MMTs have no Zeno runs or time deadlocks.

\begin{lemma}\label{lemma: at least one timed run}
	Let $w$ be a timed input word, let $\M$ be a complete MMT, and let $(q, \kappa)$ be a configuration of $\M$.
	Then there exists a timed run $\rho$ starting from $(q, \kappa)$ such that $\tinpw{\rho} = w$.
\end{lemma}
\begin{proof}
	We start by observing that, for any configuration $(q, \kappa)$ of $\M$ and for any nonnegative real number $d \leq 1$, there exists a run $\rho$ starting from $(q, \kappa)$, in which all discrete transitions are timeouts and the accumulated delay equals $d$.
	The proof is by induction on the number $n$ of timers $x$ of $q$ with $\kappa(x) < d$:
	\begin{itemize}
		\item 
		Basis.  If $n=0$ then $\rho$ just consists of a single transition \((q, \valuation) \myxrightarrow{d} (q, \valuation - d)\).
		\item 
		Induction step. If $n>0$ then let $x$ be the timer of $q$ with $\kappa(x) = e < d$ minimal.  By completeness of $\M$, we have a run
		\(\rho' = (q, \valuation) \myxrightarrow{e} (q, \valuation - d) \myxrightarrow{\timeout{x}} (q', \kappa')\).
		If the $\timeout{x}$-transition in $\rho'$ resets timer $x$, $\kappa'(x)$ will be at least 1.
		This implies that the number $n$ of timers $x$ of $q'$ with $\kappa'(x) < d$ will be less than $n$.
		Therefore, by induction hypothesis, there exists a run $\rho''$ that starts in $(q', \kappa')$, in which all discrete transitions are timeouts and the accumulated delay equals $d-e$.  Now the concatenation of runs $\rho'$ and $\rho''$ gives the desired run $\rho$.
	\end{itemize}
	By repeated application of the above observation we may construct, starting from any configuration, a run in which all discrete transitions are timeouts and the accumulated delay equals any given $d \in\nnr$.
	In combination with the fact that, due to completeness of $\M$, any configuration enables any input from $I$, this allows us to construct a timed run $\rho$ that starts from $(q, \kappa)$ such that $\tinpw{\rho} = w$.
\end{proof}

Based on \Cref{lemma: at least one timed run} and \Cref{lemma: timed inequivalence implies transparent witness}, we may now define the two types of concrete queries as follows:

\begin{definition}[Concrete Queries]
	Let \(\M\) be the complete \MMT of the teacher.
	The \emph{concrete queries} the learner can use are:
	\begin{itemize}
		\item \(\outputQ(w)\) with \(w\) a timed input word: the teacher outputs a timed trace $\tau \in \ttlanguage{\M}$ with $\tinpw{\tau} = w$.
		\item \(\equivQ(\hypothesis)\) with \(\hypothesis\) a complete \MMT: the teacher replies \yes if \(\M \ttequivalent \hypothesis\); otherwise it answers \no
		and provides a transparent timed trace contained in the symmetric difference of $\ttlanguage{\M}$ and $\ttlanguage{\hypothesis}$.
	\end{itemize}
\end{definition}

We say that a timed input word $w$ is \emph{transparent} if, for all occurrences of symbols from $I$ in $w$, the fractional part of their starting time is different.
Note that whenever a run is transparent, its associated timed input word is also transparent.
This implies that if the learner poses a concrete output query \(\outputQ(w)\) for a transparent timed input word $w$, the timed trace $\tau$ that is provided by the teacher in return will also be transparent.
In practice, a concrete equivalence query will typically be approximated by some timed testing algorithm, i.e., through timed output queries.  If only transparent timed input words are used during testing, then any counterexample found during testing will also be transparent. It is therefore not unreasonable to assume that timed equivalence queries return transparent timed traces.  Since symbolic words can be extracted from transparent timed traces, such counterexamples are quite useful for the learner.

\subsection{Proof of \texorpdfstring{\Cref{lemma:queries}}{Lemma 2}}\label{proof:lemma:queries:queries}

To ease the explanation, let us assume that the returned counterexample \(w\) is
such that \(\lengthOf{\tOutputs^\M(w)} = 1\), as \(\M\) is \robust.
Moreover, let us assume we have some \MMT \(\tree\) that contains all the
(untimed) runs that we already learned.
That is, \(\tree\) holds partial knowledge about \(\M\).
We now describe how to obtain each symbolic query from these concrete queries,
using \(\tree\), i.e., we prove \Cref{lemma:queries}.

\subsubsection{Symbolic output query.}
Let us start with symbolic output queries.
We recall the definition.
For a \siw (symbolic word) \(\symbolic{w}\) such that
\(\pi = q_0^\M \xrightarrow{\symbolic{w}} {} \in \runs{\M}\),
\(\symOutputQ(\symbolic{w})\) returns the sequence of outputs seen along the run
\(\pi\).

So, let \(\symbolic{w}\) be a \siw for which we want to ask
\(\symOutputQ(\symbolic{w})\).
During the learning process, such a query is always used to define a new
transition reading \(i \in \actions{\tree}\) from a state \(q\) that is already
present in \(\tree\), i.e., we extend the set of learned runs by extending a
run one symbol at a time.
Hence, we focus on this case.
We can assume that
\(q_0^\M \xrightarrow{\symbolic{w}} {} \in \runs{\M}\) and that, for each
proper prefix \(\symbolic{w}'\) of \(\symbolic{w}\), we already performed
\(\symWaitQ(\symbolic{w}')\).\footnote{Our learning algorithm, introduced in
\Cref{sec:learning:tree}, will ensure these two assumptions are satisfied.
Moreover, in all generality, the second assumption can easily be obtained by
simply calling \(\symWaitQ\) on every proper prefix.}
Let
\(
  \pi^\tree = p_0 \xrightarrow{i_1}
  p_1 \xrightarrow{i_2}
  \dotsb \xrightarrow{i_n}
  p_n \in \runs{\tree}
\)
such that \(p_0 = q_0^\tree\), \(p_n = q\), and
\(\toSymbolic{i_1 \dotsb i_n \cdot i} = \symbolic{w}\).
That is, we retrieve the unique run going from \(q_0^\tree\) to \(q\), convert
the actions into a symbolic word, alongside the action \(i\) (whose transition
is not necessarily already in the tree).
Since \(q_0^\M \xrightarrow{\symbolic{w}} {} \in \runs{\M}\),
\(\toSymbolic{i_1 \dotsb i_n \cdot i}\) is well-defined (the last symbol is either
an input, or \(\timeout{j}\) for some appropriate \(j \in \{1, \dotsc, n\}\)).
We now construct a \tiw that corresponds to the \siw \(\symbolic{w}\).

We leverage the system of linear constraints, noted \(\constraints(\pi)\),
introduced in the proof of \Cref{lemma: transparent}.
Recall that a solution for \(\constraints(\pi)\) gives us the delays that we can
use in a transparent, race-free timed run \(\rho\) such that \(\pi\) is the
untimed projection of \(\rho\).
Thus, let
\[
  \rho^\tree = (q_0^\tree, \emptyset) \xrightarrow{d_1}
  (q_0^\tree, \emptyset) \xrightarrow{i_1}
  (p_1, \valuation_1) \xrightarrow{d_2}
  \dotsb \xrightarrow{i_n}
  (p_n, \valuation_n)
  \xrightarrow{d_{n+1}}
  (p_n, \valuation_n - d_{n+1})
\]
be the timed run constructed from a solution of \(\constraints(\pi)\).
Deriving a timed input word \(w\) from \(\rho\) is easy: let
\(w' = d_1 i_1 \dotsb i_n d_{n+1}\) be the sequence of delays and actions of
\(\rho\).
Then, for any \(i_j = \timeout{x}\) in \(w'\), we remove \(i_j\) and replace
the delay \(d_j\) by \(d_j + d_{j+1}\). We repeat this until all timeouts are
removed from \(w'\). Observe that \(w\) contains at most as many symbols as
\(w'\) and the sums of the delays of both words are equal.
That is, one can construct a \tiw that can be passed to \outputQ from a run of
\(\tree\).

Since \(\rho^\tree\) is transparent and race-avoiding, there necessarily exists
some \(d > 0\) that is sufficiently small to guarantee that no timeout can occur
by waiting \(d\) units of time in the last configuration of \(\rho^\tree\).
That is, for every timer \(x \in \activeTimers^\tree(p_n)\),
\((\valuation_n - (d_{n+1} - d))(x) \neq 0\).

We thus call \(\outputQ(v \cdot d \cdot i \cdot 0)\).
Let us argue that the timed run \(\rho^\M\) of \(\M\) reading this word does not
contain any timeout that is \emph{unexpected} with regard to \(\rho^\tree\).
(An input-transition can never be unexpected).
In other words, any time a transition is triggered in \(\M\), we have a
corresponding transition in \(\tree\), except for the last action \(i\).
As we performed symbolic wait queries on every proper prefix of \(\symbolic{w}\),
it follows that we know the enabled timers of each traversed configuration in
\(\rho^\M\). Thus, in \(\constraints(\pi)\), we have constraints over the delays
that ensure that no unexpected timeout can occur in \(\rho^\M\), by the
definition of symbolic wait queries. Moreover, the delay \(d\) is sufficiently
small to guarantee the absence of timeouts before the last action \(i\). Hence,
the run \(\rho^M\) cannot contain any unexpected action. Let
\begin{multline*}
  \rho^\M = (q_0^\M, \emptyset)
  \xrightarrow{d_1} (q_0^\M, \emptyset)
  \xrightarrow{i'_1} (p'_1, \kappa'_1)
  \xrightarrow{d_2} \dotsb
  \xrightarrow{i'_n} (p'_n, \kappa'_n)
  \xrightarrow{d_{n+1} + d} (p'_n, \kappa'_n - (d_{n+1} + d))
  \\
  \xrightarrow{i'/o} (p'_{n+1}, \kappa'_{n+1})
  \xrightarrow{0} (p'_{n+1}, \kappa'_{n+1}).
\end{multline*}
(The actions \(i'_1, \dotsc, i'_n, i'\) are not necessarily those appearing on
\(\rho^\tree\), as \(\tree\) and \(\M\) do not necessarily use the same timers.)
Since we cannot have any unexpected transition, we can simply return the output
symbol \(o\) produced by the last discrete transition of \(\rho^\M\).

To conclude, the complexity of a symbolic output query depends solely on the
complexity of asking a symbolic wait query on every proper prefix of
\(\symbolic{w}\) in the worst case. As we will show in the following, each
symbolic wait query has a quadratic complexity, in the length of the word.
Note that, if we already asked symbolic wait queries on every proper prefix,
then we only need a single concrete output query to implement a symbolic output
query.
\begin{proposition}
  We need at most \(n^3\) concrete output queries to perform one symbolic
  output query on a symbolic word of length $n$.
\end{proposition}

\subsubsection{Symbolic wait query.}
Let us proceed with symbolic wait queries.
We recall the definition.
For a \siw (symbolic word) \(\symbolic{w}\) inducing a concrete run
\(\pi = q_0^\M \xrightarrow{i_1} \dotsb \xrightarrow{i_n} q_n \in \runs{\M}\)
such that \(\toSymbolic{i_1 \dotsb i_n} = \symbolic{w}\),
\(\symWaitQ(\symbolic{w})\) returns  the set of all pairs \((j, c)\)
such that \(q_{j-1} \xrightarrow[(x, c)]{i_j} \dotsb
\xrightarrow{i_n} q_n \xrightarrow{\timeout{x}}\) is \(x\)-spanning.

Let \(\symbolic{w}\) be a \siw for which we want to call \(\symWaitQ(\symbolic{w})\).
As for symbolic output queries, such a query is performed to know the set of
enabled timers of a state \(q\) that is already present in \(\tree\).
As in the previous part, we assume that we already asked a symbolic wait query
on every proper prefix of \(\symbolic{w}\). That is, there is no unexpected
timeout when performing a concrete output query.

Let \((q_0^\tree, \emptyset) \xrightarrow{v} (q, \valuation)\) be the transparent,
race-free timed run constructed from a solution of \(\constraints(\pi)\), with
\(v\) a \tiw as above. Moreover, let
\((q_0^\M, \emptyset) \xrightarrow{v} (q', \valuation')\) be the run reading the
same \tiw in \(\M\). Recall that each input in \(v\) is such that the
fractional part of the delays up to the input is unique, as the timed run is
transparent.
Hence, it is sufficient to wait \enquote{long enough} in \((q', \valuation')\)
to identify one potential enabled timer.
For now, assume the learner knows a constant \(\maxTime\) that is at least as large
as the largest constant appearing on any update of \(\M\).
That is, if we wait \(\maxTime\) units of time in a configuration and no timeout
occurs, then we are sure that \(\enabled{q'}[\M] = \emptyset\) (i.e., we
add \(\maxTime\) to the last delay of \(v\)).
We discuss below how to deduce \(\maxTime\) during the learning process.
Moreover, by the uniqueness of the fractional parts, it is easy to identify which
transition (re)started the timer that times out.

So, we have to explain how to ensure that we eventually observe every enabled
timer of \(q'\).
Recall that a timer must be started on a transition before \(q\) to be potentially
active in \(q\).
Thus, we define a set
\(\mathit{Potential}(q) = \{x_p \mid \text{\(p\) is an ancestor of \(q\)}\}\)
that contains every timer that \emph{may} be enabled in \(q\).
Our idea is to check each timer one by one to determine whether it is enabled.

We select any timer $x$ from \(\mathit{Potential}(q)\) and refine the constraints 
of \(\constraints^\tree(q)\) to enforce that the last delay is equal to \(\maxTime\),
and the input that initially starts \(x\) is triggered as soon as possible while
still satisfying the other constraints of \(\constraints^\tree(q)\).
It may be that the resulting constraints for that \(x\) are not satisfiable,
in which case it is not hard to deduce that a run ending with \(\timeout{x}\)
is not feasible (see the proof of \Cref{lemma: transparent}) and
\(x\) can not be enabled in \(q\).
If there exists a solution, i.e., a \tiw \(v\), we can ask \(\outputQ(w)\)
to obtain a \tow \(\omega\).
By the uniqueness of the fractional parts, it is thus easy to check whether \(x\)
times out by waiting in \(q\).
So, we can easily deduce which transition restarts \(x\).
Moreover, from the delays in \(\omega\), the constant of the update restarting
\(x\) can be computed.

We repeat this procedure for every timer in \(\mathit{Potential}(q)\).
Once this is done, we know the enabled timers of \(q\).
Let \(n\) be the number of states in the path from \(q_0\) to \(q\).
The size of \(\mathit{Potential}(q)\) is at most \(n\), and, for every timer \(x\)
in this set, we need at most \(n\) concrete output queries.

\begin{proposition}
  We need at most $n^2$ concrete output queries to perform one symbolic wait query (correct up to our guess of $\Delta$) on a symbolic word of length $n$.
\end{proposition}

\paragraph{Guessing \(\maxTime\).}
Let us quickly explain how the learner can infer \(\maxTime\) during the learning
process.
At first, \(\maxTime\) can be assumed to be any integer (preferably small when
interacting with real-world systems).
At some point, an update \((x, c)\) may be learned by performing a wait query
(or processing the counterexample of an equivalence query) with \(c > \maxTime\).
That is, we now know that \(\maxTime\) is not the largest constant appearing in
\(\M\).
We thus set \(\maxTime\) to be \(c\).
This implies that a new wait query must be performed in every explored state, in
order to discover potentially missing enabled timers.
That is, throughout the learning algorithm, the set of enabled timers in \(\tree\)
may be an under-approximation of the set of enabled timers of the corresponding
state in \(\M\) (cf.\ seismic events introduced in \Cref{sec:learning:algo}
which require rebuilding the MMT holding the learned runs from the root).
However, we will eventually learn the correct value of \(\maxTime\)
(cf.\ \Cref{app:learning:termination} for a bound relative to the
unknown $\M$ on how many times seismic events and, more generally, the discovery
of new timers, can occur).

\subsubsection{Symbolic equivalence query.}
Finally, let us explain how to implement a symbolic equivalence query \(\symEquivQ(\hypothesis)\).
Assume $\hypothesis$ is \complete and race-avoiding.
In order to implement the symbolic query, we go through the following sequence of steps:
\begin{enumerate}
\item 
First pose a concrete equivalence query \(\equivQ(\hypothesis)\).
If the teacher replies \yes then \(\M \ttequivalent \hypothesis\).
Therefore, and because both $\M$ and $\hypothesis$ are race-avoiding, \(\M \symEquivalent \hypothesis\) by \Cref{lemma: timed equivalence implies symbolic equivalence}.  Thus the response to the symbolic equivalence query is also \yes, and we are done.
\item
Otherwise, suppose the response to the concrete equivalence query is \no, together with a transparent timed trace $\tau$ contained in the symmetric difference of $\ttlanguage{\M}$ and $\ttlanguage{\hypothesis}$.
Then, by \Cref{Cy:symEquivalent:equivalent} and \Cref{lemma:symEquivalent:equivalent}, \(\M \not\symEquivalent \N\).
Thus the response to the symbolic equivalence query will also be \no.
However, we still need to determine a symbolic counterexample.
\item
Let $\w = \sinpw{\tau}$.
If $\w \not\in \silanguage{\hypothesis}$ then $\tau \not\in\ttlanguage{\hypothesis}$ by \Cref{lemma:commuting2}.
Therefore, since $\tau$ is a concrete counterexample,
$\tau\in\ttlanguage{\M}$ and $\w \in \silanguage{\M}$.
Thus $\w$ is a symbolic counterexample, and we are done.
Otherwise, we conclude $\w \in \silanguage{\hypothesis}$. 
\item 
Let $\omega = \outw{\tau}$.
If $\ofunction{\hypothesis}{\w} \neq \omega$ then $\tau \not\in\ttlanguage{\hypothesis}$ by \Cref{lemma:commuting1}.
Therefore, since $\tau$ is a concrete counterexample,
$\tau\in\ttlanguage{\M}$, $\w \in \silanguage{\M}$ and $\ofunction{\M}{\w} = \omega$ by \Cref{lemma:commuting1}.
Then $\w$ is a symbolic counterexample, and we are done.
Otherwise, we conclude $\ofunction{\hypothesis}{\w} = \omega$.
\item 
If $\tau \not\in\ttlanguage{\hypothesis}$ then, since $\tau$ is a counterexample, $\tau \in \ttlanguage{\M}$.
Then $\w \in \silanguage{\M}$ and $\ofunction{\M}{\w} = \omega$, so
$\w$ is not a counterexample anymore with this assumption.
Let $\pi$ be the unique run of $\hypothesis$ with $\sinpw{\pi} = \w$, and let $w = \tinpw{\tau}$.
If we apply transparent timed input word $w$ to $\hypothesis$, we obtain a unique timed trace of $\hypothesis$ that initially follows $\pi$ but then at some point deviates (because otherwise $\tau$ would be a timed trace of $\hypothesis$).  Since $\pi$ has exactly the same inputs, timeouts, and outputs as $\tau$, the only possibility is that at some point a timeout occurs in $\hypothesis$ that does not appear in $\tau$.
Let $\pi'$ be the untimed run of $\hypothesis$ up until this timeout,
and let $\w'$ be the symbolic input word of $\pi'$.
Then $\w'$ is not a symbolic input word of $\M$, otherwise the same timeout should have occurred in $\M$ as well. 
So also for this case, we have found a symbolic counterexample.
Otherwise, we conclude $\tau\in\ttlanguage{\hypothesis}$ and $\tau\not\in\ttlanguage{\M}$.
\item 
We perform a concrete output query on $\M$ for the timed input word $w$.  Let $\tau' \in\ttlanguage{\M}$ be the resulting transparent timed trace returned by the teacher.  Then $\tau \neq \tau'$.
\item
Now we compare the timing of the events in $\tau$ and $\tau'$ step by step, starting from the beginning.  Since $\tau$ and $\tau'$ have the same transparent timed input word, all the input events from $I$ occur at exactly the same time. However, we may have situation where at some (global) time $t$ there is a timeout event in one sequence, but not in the other.  In such a case, let $\tau''$ be the prefix of the timed input sequence up to and including that timeout, followed by a $0$ delay. Then $\tau''$ is a transparent timed trace, and
$\sinpw{\tau''}$ is a symbolic word in the symmetric difference of
$\silanguage{\hypothesis}$ and $\silanguage{\M}$, and therefore a symbolic counterexample.
Otherwise, if $\tau$ and $\tau'$ agree on the timing of all events, we conclude that $\w$ and $\w' = \sinpw{\tau'}$ are equal.
\item 
Let $\omega'= \outw{\tau'}$.
If two transparent timed words share the same timed input words, the same symbolic input words, and the same outputs, then they are equal.
Since $\tau$ and $\tau'$ are different,
this implies that $\omega$ and $\omega'$ are different.
By \Cref{lemma:commuting1}, $\ofunction{\hypothesis}{\w} = \omega$ and $\ofunction{\M}{\w'} = \omega'$.
This implies that $\w$ constitutes a counterexample.
\end{enumerate}
The next proposition is immediate.

\begin{proposition}
	We need one concrete equivalence query and at most one concrete output query to perform one symbolic equivalence query.
\end{proposition} 
\section{More details on observation trees}\label{app:tree}

In this section, we give more details on observation trees and their properties.
First, in \Cref{app:tree:simulation}, we formalize how one can map each state (resp.\ timer) of an observation tree \(\tree\)
to a state (resp.\ timer) of the \good \MMT \(\M\) of the teacher.
Then, in \Cref{app:tree:extension_soundness}, we prove
\Cref{thm:extension-n-soundness}.
In \Cref{app:tree:cotransitivity}, we show that if \(w\) is a witness of
\(p_0 \apart^m p'_0\) and \(w\) can be read from a state \(r\)
(i.e., \(\copyRun[m]{r}[p_0 \xrightarrow{w}]\) is defined), then
\(r\) must be apart from \(p_0\) or \(p'_0\) under some matchings derived from
\(m\).
This property is used in the learning algorithm to reduce the number of possible
hypotheses that can be constructed.
See \Cref{app:learning:termination}.

\subsection{Functional simulation}\label{app:tree:simulation}

In order to link (the finitely many) runs of \(\tree\) to runs of \(\M\), we
formalize the notion of functional simulation \(\funcSim\) introduced in
\Cref{sec:learning:tree}.
Recall that, since \(\tree\) is tree shaped, there is a unique run from
\(q_0^\tree\) to a state \(q\), noted \(\run{\tree}{q}\).
Furthermore, any run of \(\tree\) is feasible.
Hence, for every transition \(q \xrightarrow{\timeout{x}} q'\), there
exists an earlier transition that (re)starts \(x\).
That is, \(\cause{\run{\tree}{q}}{x}\) is well-defined.
Finally, if we assume that \(f(q) \xrightarrow{\timeout{g(x)}} f(q')\) exists
(see~\eqref{eq:simulation:transition} below), then \(g(x)\) is active in
\(f(q)\) and the timeout-transition has a cause.
That is, \(\cause{\funcSim(\run{\tree}{q})}{g(x)}\) is also well-defined, under
that assumption.

\begin{definition}[Functional simulation]
  Let \(\tree\) be an observation tree and \(\M\) be a \good \MMT.
  A \emph{functional simulation} \(\funcSim : \tree \to \M\) is a pair of two maps
  \(f : Q^\tree \to Q^\M\) and
  \(g : \bigcup_{q \in Q^\tree} \activeTimers^\tree(q) \to X^\M\).
  Let \(g\) be lifted to actions such that \(g(i) = i\) for every \(i \in I\), and
  \(g(\timeout{x}) = \timeout{g(x)}\) for every \(x \in \dom{g}\).
  We require that $\funcSim$ preserves initial states, active timers, and
  transitions:
  \begin{gather}
    f(q_0^\tree) = q_0^\M
    \tag{FS0}\label{eq:simulation:initial}
    \\
    \forall q \in Q^\tree, \forall x \in \activeTimers^\tree(q) :
    g(x) \in \activeTimers^\M(f(q))
    \tag{FS1}\label{eq:simulation:active:implies}
    \\
    \forall q \in Q^\tree, \forall x, y \in \activeTimers^\tree(q) :
    x \neq y \implies g(x) \neq g(y)
    \tag{FS2}\label{eq:simulation:different_timers}
    \\
    \forall q \xrightarrow{i/o} q' : f(q) \xrightarrow{g(i) / o} f(q')
    \tag{FS3}\label{eq:simulation:transition}
  \end{gather}
  Thanks to~\eqref{eq:simulation:initial} and~\eqref{eq:simulation:transition},
  we lift $\funcSim$ to runs in a straightforward manner.
  We require that timeout-transitions have the same causes in \(\tree\) and
  \(\M\):
  \begin{equation}
    \forall q \xrightarrow{\timeout{x}} {} :
    \cause{\run{\tree}{q}}{x} = \cause{\funcSim(\run{\tree}{q})}{g(x)}
    \tag{FS4}\label{eq:simulation:cause}
  \end{equation}
  We say that \(\tree\) is an \emph{observation tree for \(\M\)} if there exists
  \(\funcSim : \tree \to \M\).
\end{definition}

\begin{example}\label{ex:tree:forM}
  Let \(\M\) be the \MMT of \Cref{fig:ex:MMT:good}.
  Then, the observation tree \(\tree\) of \Cref{fig:ex:tree} is an observation
  tree for \(\M\) with the functional simulation \(\funcSim\) such that
  \begin{align*}
    f(t_0) = f(t_8) = f(t_{10}) &= q_0
    &
    f(t_1) = f(t_2) = f(t_4) &= q_1
    &
    f(t_3) &= q_2
    &
    f(t_5) = f(t_6) &= q_3
    \\
    f(t_7) = f(t_9) &= q_5
    &
    g(x_1) = g(x_6) &= x
    &
    g(x_3) &= y.
  \end{align*}

  Let \(\pi = \run{\tree}{t_5} = t_0 \xrightarrow[(x_1, 2)]{i} t_1
  \xrightarrow[(x_3, 3)]{i} t_3 \xrightarrow[(x_1, 2)]{\timeout{x_1}} t_5\).
  Let us check that~\eqref{eq:simulation:cause} holds.
  First, the corresponding run in \(\M\) is \(\funcSim(\pi) =
  q_0 \xrightarrow[(x, 2)]{i} q_1 \xrightarrow[(y, 3)]{i} q_2
  \xrightarrow[(x, 2)]{\timeout{x}} q_3\).
  From \(t_5\), there are two timeout-transitions: one for \(x_1\), and one for
  \(x_3\).
  \begin{itemize}
    \item
    The last transition that updated \(x_1\) in \(\tree\) is
    the third transition of \(\pi\), i.e.,
    \(t_3 \xrightarrow[(x_1, 2)]{\timeout{x_1}} t_5\).
    Moreover, the last transition that updated \(g(x_1) = x\) in \(\M\) is
    the third transition of \(\funcSim(\pi)\), i.e.,
    \(q_2 \xrightarrow[(x, 2)]{\timeout{x}} q_3\).
    So, \(\cause{\pi}{x_1} = (2, 3) = \cause{\funcSim(\pi)}{g(x_1)}\).

    \item
    The last transition that updated \(x_3\) in \(\tree\) is
    \(t_1 \xrightarrow[(x_3, 3)]{i} t_3\).
    Moreover, the last transition that updated \(g(x_3) = y\) in \(\M\) is
    \(q_1 \xrightarrow[(y, 3)]{i}\).
    So, \(\cause{\pi}{x_3} = (3, 2) = \cause{\funcSim(\pi)}{g(x_1)}\).
  \end{itemize}
  Thus,~\eqref{eq:simulation:cause} holds.
\end{example}

Observe that for fixed \(\tree\) and \(\M\) there exists at most one functional
simulation.
Further properties can be deduced from the definition of \(\funcSim\):
\begin{itemize}
  \item
  For any transition whose update is \((x, c)\) in \(\tree\), the update of the
  corresponding transition in \(\M\) is \((g(x), c)\), i.e., both work on the
  same timer (up to a renaming) and set it at the same value.

  However, if the update of the transition in \(\tree\) is \(\bot\), we cannot
  conclude anything about the update in \(\M\) (it can be any update in
  \(\updates{\M}\)).
  Indeed, we may have not yet discovered that an update must take place.

  \item
  Any \(x\)-spanning run of \(\tree\) has a corresponding \(g(x)\)-spanning run
  in \(\M\), and vice-versa.

  \item
  The number of active timers in \(q\) is bounded by the number of active timers
  in \(f(q)\).

  \item
  For any enabled timer \(x\) of \(q\), \(g(x)\) is enabled in \(f(q)\).
\end{itemize}

\begin{lemma}
  Let \(\tree\) be an observation tree, \(\M\) be an \good \MMT, and
  \(\funcSim : \tree \to \M\) be a functional simulation.
  Then,
  \begin{gather}
    \forall q \xrightarrow[(x, c)]{i / o} q' :
    f(q) \xrightarrow[(g(x), c)]{g(i) / o} f(q')
    \tag{FS5}\label{eq:simulation:update:known}
    \\
    \forall q \xrightarrow[\bot]{i / o} q' :
    f(q) \xrightarrow{g(i) / o} f(q')
    \tag{FS6}\label{eq:simulation:update:bot}
    \\
    \forall \pi \in \runs{\tree} : \text{\(\pi\) is \(x\)-spanning}
    \implies
    \text{\(\funcSim(\pi)\) is \(g(x)\)-spanning}
    \tag{FS7}\label{eq:simulation:spanning:fromTree}
    \\
    \forall \pi \in \runs{\tree} : \text{\(\funcSim(\pi)\) is \(y\)-spanning}
    \implies
    \text{\(\pi\) is \(x\)-spanning} \land g(x) = y
    \tag{FS8}\label{eq:simulation:spanning}
    \\
    \lengthOf{\activeTimers^\tree(q)} \leq \lengthOf{\activeTimers^\M(f(q))}
    \tag{FS9}\label{eq:simulation:active:size}
    \\
    \forall x \in \enabled{q}[\tree] : g(x) \in \enabled{f(q)}[\M].
    \tag{FS10}\label{eq:simulation:enabled}
  \end{gather}
\end{lemma}
\begin{proof}
  We show each property one by one.

  \begin{description}
    \item[\eqref{eq:simulation:update:known}]
    Let \(q \xrightarrow[(x, c)]{i/o} q'\) be a transition of \(\tree\).
    Recall that, by definition of an observation tree, \(x\) is active in \(q'\)
    if and only if there exists a \(x\)-spanning run traversing \(q'\).
    Therefore, there exists some state \(p\) such that
    \[
      \pi = \run{\tree}{p} =
      q_0^\tree \xrightarrow[u_1]{i_1}
      q_1 \xrightarrow[u_2]{i_2}
      \dotsb \xrightarrow[u_n]{i_n}
      q_n,
    \]
    \(i_n = \timeout{x}\), and there exists some \(j\) such that \(q_{j-1} = q\),
    \(q_j = q'\), and \(u_j = (x, c)\).
    Let us assume \(\pi\) is the shortest such path (i.e.,
    \(q_{n-1} \xrightarrow{\timeout{x}}\) is the first \(\timeout{x}\)-transition
    after \(q'\)).
    Thus, \(\cause{\pi}{x} = (c, j)\).
    By~\eqref{eq:simulation:cause}, it holds that
    \(\cause{\funcSim(\pi)}{g(x)} = (c, j)\).
    Hence, we have
    \(\funcSim(q_{j-1})
      \xrightarrow[(g(x), c)]{g(i)}
      \funcSim(q_{j+1})\).
    Since \(q_{j-1} = q\) and \(q_j = q'\), we conclude
    that~\eqref{eq:simulation:update:known} holds.

    \item[\eqref{eq:simulation:update:bot}]
    We immediately get~\eqref{eq:simulation:update:bot}
    from~\eqref{eq:simulation:transition}.

    \item[\eqref{eq:simulation:spanning:fromTree}]
    Let \(x\) be a timer of \(\tree\), and
    \(
      \pi = p_0 \xrightarrow[(x, c)]{i_1}
      p_1 \xrightarrow[u_2]{i_2}
      \dotsb \xrightarrow[u_n]{\timeout{x}}
      p_n
    \)
    be a \(x\)-spanning run of \(\tree\).
    By~\eqref{eq:simulation:update:known} and~\eqref{eq:simulation:transition},
    we obtain
    \[
      \funcSim(\pi) = f(p_0) \xrightarrow[(g(x), c)]{g(i_1)}
      f(p_1) \xrightarrow{g(i_2)}
      \dotsb \xrightarrow{\timeout{g(x)}}
      f(p_n).
    \]
    By~\eqref{eq:simulation:cause}, we know that
    \(\cause{\run{\tree}{p_0} \cdot \pi}{x} =
    \cause{\funcSim(\run{\tree}{p_0} \cdot \pi)}{g(x)}\),
    where \(\run{\tree}{p_0} \cdot \pi\) denotes the run obtained by
    concatenating \(\run{\tree}{p_0}\) (which ends in \(p_0\)) and
    \(\pi\) (which starts in \(p_0\)).
    Therefore, none of the intermediate transitions restarts \(g(x)\), and
    \(\funcSim(\pi)\) is \(g(x)\)-spanning.

    \item[\eqref{eq:simulation:spanning}]
    Let \(\pi = p_0 \xrightarrow[u_1]{i_1} \dotsb \xrightarrow[u_n]{i_n} p_n\)
    be a run of \(\tree\) such that
    \(\funcSim(\pi) = f(p_0) \xrightarrow[u'_1]{g(i_1)} \dotsb
    \xrightarrow[u'_n]{g(i_n)} f(p_n)\) is \(y\)-spanning, with \(y \in X^\M\).
    That is, \(u'_1 = (y, c)\) for some \(c \in \natplus\),
    \(g(i_n) = \timeout{y}\), and none of the intermediate transitions restarts
    \(y\).
    Hence, there exists some \(x \in X^\tree\) such that \(g(x) = y\)
    and \(i_n = \timeout{x}\), as \(\pi\) is a run of \(\tree\).
    By~\eqref{eq:simulation:cause}, we have
    \(\cause{\run{\tree}{p_0} \cdot \pi}{x} =
    \cause{\funcSim(\run{\tree}{p_0} \cdot \pi)}{y}\) and, so,
    \(u_1 = (x, c)\) and for all \(j \in \{2, \dotsc, n-1\}\), \(u_j \neq (x, d)\)
    with \(d \in \natplus\).
    Thus, \(\pi\) is \(x\)-spanning, with \(g(x) = y\), as wanted.

    \item[\eqref{eq:simulation:active:size}]
    By~\eqref{eq:simulation:active:implies}, we have that any timer \(x\) that is
    active in \(q\) is such that \(g(x)\) is active in \(f(x)\).
    Moreover, by~\eqref{eq:simulation:different_timers}, \(g(x) \neq g(y)\) for
    any \(x \neq y \in \activeTimers^\tree(q)\).
    So, it is not possible for \(q\) to have more active timers than \(f(q)\).

    \item[\eqref{eq:simulation:enabled}]
    Let \(x \in \enabled{q}[\tree]\).
    By definition of \(\tree\), it follows that \(q \xrightarrow{\timeout{x}}\)
    is defined.
    So, by~\eqref{eq:simulation:update:known} and~\eqref{eq:simulation:update:bot}, we have
    \(f(q) \xrightarrow{\timeout{g(x)}}\), meaning that
    \(g(x) \in \enabled{f(q)}[\M]\), as \(\M\) is \complete.\qed
  \end{description}
\end{proof}

\begin{figure}[t]
  \centering
  \begin{tikzpicture}[
  automaton,
  node distance = 90pt,
  every loop/.append style = {
    min distance = 4mm
  },
]
  \node [state, initial]      (q0)  {\(q_0\)};
  \node [state, right=of q0]  (q1)  {\(q_1\)};
  \node [state, right=of q1]  (q2)  {\(q_2\)};

  \path
    (q0)  edge              node [above] {\(i/o, \updateFig{x}{2}\)}
                            (q1)
    (q1)  edge [loop above] node {\(\timeout{x}/o, \updateFig{x}{2}\)} (q1)
          edge              node [above] {\(i/o', \updateFig{y}{3}\)}
                            (q2)
    (q2)  edge [loop above] node {\(i/o', \updateFig{x}{2}\)} (q2)
          edge [loop right] node {\(\timeout{x}/o, \updateFig{x}{2}\)} (q2)
  ;

  \draw [rounded corners = 10pt]
    let
      \p{s} = (q2.-135),
      \p{e} = (q0.-45),
    in
    (\p{s}) -- ($(\x{s}, \y{e}) + (-0.4, -0.25)$)
            -- ($(\p{e}) + (0.4, -0.25)$)
            node [very near start, above=-2pt] {\(\timeout{y}/o\)}
            -- (\p{e})
  ;
\end{tikzpicture}
  \caption{An \MMT with \(\activeTimers(q_0) = \emptyset\),
  \(\activeTimers(q_1) = \{x\}\), and \(\activeTimers(q_2) = \{x, y\}\).}\label{fig:app:ex:MMT}
\end{figure}
Finally, let us highlight that the notion of explored states only makes
sense when \(\M\) is \good.
Let \(\tree\) be the observation tree of \Cref{fig:ex:tree} and \(\N\) be the
not-\good \MMT of \Cref{fig:app:ex:MMT}.
We can still define the maps \(f : Q^\tree \to Q^\N\) and \(g : X^\tree \to X^\N\):
\begin{align*}
  f(t_0) = f(t_8) = f(t_{10}) &= q_0
  &
  f(t_1) = f(t_2) = f(t_4) &= q_1
  &
  f(t_3) = f(t_5) = f(t_6) = f(t_7) = f(t_9) &= q_2
  \\
  g(x_1) = g(x_6) &= x
  &
  g(x_3) &= y.
\end{align*}
We have \(\lengthOf{\enabled{t_3}[\tree]} = 1\)
but \(\lengthOf{\enabled{f(t_3)}[\N]} = \lengthOf{\enabled{q_2}[\N]} = 2\).
However, as every run of \(\tree\) must be feasible and \(x_3\) cannot time out
in \(t_3\), it is impossible to get the equality.
Therefore, in order to define explored states, we must require that \(\M\) is
\good.

\subsection{Proof of \texorpdfstring{\Cref{thm:extension-n-soundness}}{Theorem 2}}\label{app:tree:extension_soundness}

We now show both parts of \Cref{thm:extension-n-soundness} separately.

\subsubsection*{Apartness of timers.}
First, we introduce a notion of \emph{apartness of timers}, which will be useful
to simplify the proofs.
Recall that if two distinct timers \(x\) and \(y\) are both active in
the same state of $\tree$, it must hold that \(g(x) \neq g(y)\), i.e., they correspond to
different timers in \(\M\).
Hence, we say that \(x, y\) are \emph{apart}, noted \(x \timerApart y\) if
there exists \(q \in Q\) such that \(x, y \in \activeTimers(q)\).
In the context of an observation tree \(\tree\), we can rephrase this as follows:
two timers \(x_q, x_p\) are apart if and only if \(x_q\) and \(x_p\) are both
active in \(q\) or in \(p\). That is, we do not need to iterate over the whole
observation tree to decide whether \(x_q \timerApart x_p\).

The definition of apartness for states is parametrized by the notion of
matching, which can be seen as a hypothesis on the equivalence of timers.
Recall that a matching is meant to indicate that the timers
\(x\) and \(m(x)\) could represent the same timer in \(\M\) and that
we say  \(m\) is \emph{valid} if for all \(x \in \dom{m}\),
\(\lnot (x \timerApart m(x))\).
For notational convenience, we lift \(m\) to actions:
\(m(i) = i\) for all \(i \in I\), and \(m(\timeout{x}) = \timeout{m(x)}\) for
every \(x \in \dom{m}\).

\begin{thmSoundnessPart}
  \(w \witness p \apart^m p' \land m \subseteq m'
  \implies
  w \witness p \apart^{m'} p'\).
\end{thmSoundnessPart}
\begin{proof}
  Let \(w \witness p \apart^m p'\) and \(m \subseteq m'\).
  Moreover, let $p_0 = p$, \(p'_0 = p'\),
  \begin{gather*}
    \pi = p_0 \xrightarrow{i_1} p_1 \xrightarrow{i_2} \dotsb
    \xrightarrow[u]{i_n/o} p_n, \text{ and }
    \pi' = \copyRun{p'}[p \xrightarrow{w}] = p'_0 \xrightarrow{i'_1} p'_1
    \xrightarrow{i'_2} \dotsb \xrightarrow[u']{i'_n/o'} p'_n
  \end{gather*}
  with $\matchingRun{m}{\pi}{\pi'} : \pi \leftrightarrow \pi'$.
  By definition, each \(i_j\) is either an input, or \(\timeout{x}\) with
  \(x \in \dom{\matchingRun{m}{\pi}{\pi'}}\).
  Thus, since \(m \subseteq m'\), it follows that
  \(\copyRun[m']{p'}[p \xrightarrow{w}]\) uses exactly the same actions and takes
  the same transitions as \(\copyRun{p'}[p \xrightarrow{w}]\).
  That is, \(\copyRun[m']{p'}[p \xrightarrow{w}] = \copyRun[m]{p'}[p \xrightarrow{w}]\).
  There are five cases:
  \begin{itemize}
    \item
    There exists \(x \in \dom{\matchingRun{m}{\pi}{\pi'}}\) such that \(x \timerApart \matchingRun{m}{\pi}{\pi'}(x)\). If $x \timerApart m(x)$, then \(x \timerApart m'(x)\) since \(m \subseteq m'\).
    If there exists \(k \in \{1, \dotsc, n\}\) such that
    \(x_{p_k} \timerApart x_{p'_k}\), this does not change when extending \(m\). Hence, \(x \in \dom{\matchingRun{m'}{\pi}{\pi'}}\) and \(x \timerApart \matchingRun{m'}{\pi}{\pi'}(x)\), i.e.,
    we have \(w \witness p \apart^{m'} p'\).
    \item \(o \neq o'\), which, clearly, does not depend on \(m\).
    So, \(w \witness p \apart^{m'} p'\).
    \item Likewise if \(u = (x, c)\) and \(u' = (x', c')\) with \(c \neq c'\).
    \item \(p_n, p'_n \in \explored\) and \(\lengthOf{\enabled{p_n}} \neq
    \lengthOf{\enabled{p'_n}}\), which, again, does not depend on \(m\).
    So, \(w \witness p \apart^{m'} p'\).
    \item 
    \(p_n, p'_n \in \explored\) and there is \(x \in \dom{\matchingRun{m}{\pi}{\pi'}}\) such that
    \(x \in \enabled{p_n} \iff \matchingRun{m}{\pi}{\pi'}(x) \notin \enabled{p'_n}\).
    If $x \in \dom{m}$ and as \(m \subseteq m'\), we still have
    \(x \in \enabled{p_n} \iff m'(x) \notin \enabled{p'_n}\).
    Likewise if
    there is a \(k \in \{1, \dotsc, n\}\) such that
    \(x_{p_k} \in \enabled{p_n} \iff x_{p'_k} \notin \enabled{p'_n}\).
    Therefore, we have again \(w \witness p \apart^{m'} p'\).
  \end{itemize}
  In every case, we obtain that \(w \witness p \apart^{m'} p'\).
\qed\end{proof}

Hence, the first part of \Cref{thm:extension-n-soundness} holds.
We now focus on the second part, which requires an intermediate result.
Recall that, given two runs \(\pi = p_0 \xrightarrow{i_1} \dotsb
\xrightarrow{i_n} p_n\) and \(\pi' = p'_0 \xrightarrow{i'_1} \dotsb
\xrightarrow{i'_n} p'_n\), \(\matchingRun{m}{\pi}{\pi'} :
\pi \leftrightarrow \pi'\) denotes the matching such that
\(\matchingRun{m}{\pi}{\pi'} = m \cup \{(x_{p_j}, x_{p'_j}) \mid
j \in \{1, \dotsc, n\}\}\) and \(i'_j = \matchingRun{m}{\pi}{\pi'}(i_j)\) for
all \(j \in \{1, \dotsc, n\}\).
Given such a matching \(\matchingRun{m}{\pi}{\pi'} :
\pi \leftrightarrow \pi'\), the next lemma states that if $f(p_0) = f(p'_0)$ and $m$ agrees with $g$, then \(\funcSim(\pi) = \funcSim(\pi')\) and $\matchingRun{m}{\pi}{\pi'}$ (restricted to the started timers, ensuring that \(g\) is defined over
those timers) also agrees with $g$.

\begin{lemma}\label{lemma:soudness:runsInM}
    Let \(p_0, p'_0 \in Q^\tree\) and a matching \(m : p_0 \leftrightarrow p'_0\)
    such that \(f(p_0) = f(p'_0)\) and \(g(x) = g(m(x))\) for all \(x \in \dom{m}\).
    Moreover, let \(w = i_1 \dotsb i_n\) be a word such that
    \begin{align*}
      \pi &=
        p_0 \xrightarrow{i_1}
        p_1 \xrightarrow{i_2}
        \dotsb \xrightarrow{i_n}
        p_n \in \runs{\tree}, \text{ and }
      \pi' = \copyRun{p'_0}[\pi] =
        p'_0 \xrightarrow{i'_1}
        p'_1 \xrightarrow{i'_2}
        \dotsb \xrightarrow{i'_n}
        p'_n \in \runs{\tree}.
    \end{align*}
    Then, \(\funcSim(\pi) = \funcSim(\pi')\) and
    \(g(x) = g(\matchingRun{m}{\pi}{\pi'}(x))\) for all
    \(x \in \dom{\matchingRun{m}{\pi}{\pi'}}\)
    with \(x \in \dom{m}\) or \(x = x_{p_k}\), $k \in \{1,\dotsb,n\}$,
    such that $x_{p_k}$ is started along \(\pi\) and
    \(x'_{p_k}\) is started along \(\pi'\).
\end{lemma}
\begin{proof}
  We prove the lemma by induction over \(n\), the length of \(w\).

  \textbf{Base case:} \(\lengthOf{w} = 0\), i.e., \(w = \emptyword\).
  We thus have
  \begin{align*}
    \pi = p_0 &\xrightarrow{\emptyword} p_0
    &\text{and}&&
    \pi' = p'_0 &\xrightarrow{\emptyword} p'_0
  \shortintertext{which means that we have the following runs in \(\M\)}
    f(p_0) &\xrightarrow{\emptyword} f(p_0)
    &\text{and}&& f(p'_0) &\xrightarrow{\emptyword} f(p'_0).
  \end{align*}
  As \(f(p_0) = f(p'_0)\), these runs of $\M$ are equal.
  Moreover, as \(\matchingRun{m}{\pi}{\pi'} = m\), the second part of the lemma
  holds.

  \textbf{Induction step:}
  Let \(\ell \in \nat\) and assume the lemma holds for length \(\ell\).
  Let \(v = i_1 \dotsb i_{\ell+1} = w \cdot i_{\ell+1}\) be a word of length
  \(\ell + 1\) such that
  \(p_0 \xrightarrow{i_1} \dotsb \xrightarrow{i_\ell} p_\ell
    \xrightarrow{i_{\ell+1}} p_{\ell+1} \in \runs{\tree}\) and
  \(\copyRun{p'_0}[p_0 \xrightarrow{w \cdot i_{\ell+1}} p_{\ell+1}] =
    p'_0 \xrightarrow{i'_1} \dotsb
    \xrightarrow{i'_\ell} p'_n \xrightarrow{i'_{\ell+1}} p'_{\ell+1} \in \runs{\tree}\).
  Let $\pi = p_0 \xrightarrow{w} p_\ell$ and $\pi' = \copyRun{p'_0}[\pi] = p'_0 \xrightarrow{w'} p'_\ell$.
  By the induction hypothesis with $w$, it holds that
  \begin{itemize}
    \item
    the runs $\funcSim(\pi)$ and $\funcSim(\pi')$ are equal, and 
    \item
    \(g(x) = g(\matchingRun{m}{\pi}{\pi'}(x))\) for all started timers
    \(x \in \dom{\matchingRun{m}{\pi}{\pi'}}\).
  \end{itemize}
  It is thus sufficient to show that
  \begin{itemize}
    \item 
    $\funcSim(p_\ell \xrightarrow{i} p_{\ell+1}) = \funcSim(p'_\ell \xrightarrow{i'} p'_{\ell+1})$, and
    \item \(g(x_{p_{\ell+1}}) = g(x_{p'_{\ell+1}})\) if both $x_{p_{\ell+1}}$
    and $x_{p'_{\ell+1}}$ are started, i.e., if
    $x_{p_{\ell+1}} \in \activeTimers(p_{\ell+1})$
    and $x_{p'_{\ell+1}} \in \activeTimers(p'_{\ell+1})$.
  \end{itemize}

  By definition of
  \(\copyRun{p'_0}[p_0 \xrightarrow{w \cdot i_{\ell+1}} p_{\ell+1}]\), we have
  \[
    i'_{\ell+1} = \begin{cases}
      i_{\ell+1} & \text{if \(i_{\ell+1} \in I\)}\\
      \timeout{m(x)} & \text{if \(i_{\ell+1} = \timeout{x}\) with \(x \in \dom{m}\)}\\
      \timeout{x_{p'_k}} & \text{if \(i_{\ell+1} = \timeout{x_{p_k}}\)
        with \(k \in \{1, \dotsc, \ell\}\)}
    \end{cases}
  \]
  We can be more precise for the last case, i.e., when
  \(i_{\ell+1} = \timeout{x_{p_k}}\) with \(k \in \{1, \dotsc, \ell\}\).
  As \(p_{\ell} \xrightarrow{\timeout{x_{p_k}}} {} \in \runs{\tree}\), it must be
  that \(x_{p_k} \in \activeTimers(p_\ell)\).
  Hence, by definition of an observation tree,
  \(x_{p_k} \in \activeTimers(p_k)\).
  Likewise, as \(p'_\ell \xrightarrow{\timeout{x_{p'_k}}} {} \in \runs{\tree}\),
  it follows that \(x_{p'_k} \in \activeTimers(x_{p'_k})\).
  Hence, \(g(x_{p_k})\) and \(g(x_{p'_k})\) are both defined when the third case
  holds.

  By definition of \(g\), it holds that
  \begin{align*}
    g(i_{\ell+1}) &= \begin{cases}
      i_{\ell+1} & \text{if \(i_{\ell+1} \in I\)}\\
      \timeout{g(x)} & \text{if \(i_{\ell+1} = \timeout{x}\) with \(x \in \dom{m}\)}\\
      \timeout{g(x_{p_k})} & \text{if \(i_{\ell+1} = \timeout{x_{p_k}}\) with
        \(k \in \{1, \dotsc, \ell\}\)}
    \end{cases}
    \shortintertext{and}
    g(i'_{\ell+1}) &= \begin{cases}
      i'_{\ell+1} & \text{if \(i'_{\ell+1} \in I\)}\\
      \timeout{g(m(x))} &
        \text{if \(i'_{\ell+1} = \timeout{m(x)}\) with \(x \in \dom{m}\)}\\
      \timeout{g(x_{p'_k})} &
        \text{if \(i'_{\ell+1} = \timeout{x_{p'_k}}\) with \(k \in \{1, \dotsc, \ell\}\).}
    \end{cases}
  \end{align*}
  We have
  \begin{itemize}
      \item
      \(i'_{\ell+1} = i_{\ell+1}\) if \(i_{\ell+1} \in I\),
      \item
      \(g(m(x)) = g(x)\) for all \(x \in \dom{m}\) by the lemma statement, and
      \item
      by the induction hypothesis, \(g(x_{p'_k}) = g(x_{p_k})\) for all
      \(k \in \{1, \dotsc, \ell\}\)
      such that \(x_{p_k} \in \activeTimers(p_k)\) and \(x_{p'_k} \in \activeTimers(p'_k)\).
  \end{itemize}
  It follows that \(g(i'_{\ell+1}) = g(i_{\ell+1})\).

  As \(f(p_\ell) = f(p'_\ell)\) by induction hypothesis and \(g(i_{\ell+1}) = g(i'_{\ell+1})\),
  it holds by determinism of \(\M\) that
  $\funcSim(p_\ell \xrightarrow{i} p_{\ell+1}) = \funcSim(p'_\ell \xrightarrow{i'} p'_{\ell+1})$.

  To complete the proof of the lemma, it remains to prove that if
  $x_{p_{\ell+1}} \in \activeTimers(p_{\ell+1})$ and
  $x_{p'_{\ell+1}} \in \activeTimers(p'_{\ell+1})$, then
  \(g(x_{p_{\ell+1}}) = g(x_{p'_{\ell+1}})\).
  We have $x_{p_{\ell+1}} \in \activeTimers(p_{\ell+1})$ (resp.
  $x_{p'_{\ell+1}} \in \activeTimers(p'_{\ell+1})$) if the update of the
  transition $p_\ell \xrightarrow{i} p_{\ell+1}$
  (resp. $p'_\ell \xrightarrow{i'} p'_{\ell+1}$) is equal to $(x_{p_{\ell+1}},c)$
  for some $c$ (resp. $(x_{p'_{\ell+1}},c')$ for some $c'$).
  As $\funcSim(p_\ell \xrightarrow{i} p_{\ell+1}) =
  \funcSim(p'_\ell \xrightarrow{i'} p'_{\ell+1})$ and $\funcSim$ is a functional
  simulation, by~\eqref{eq:simulation:update:known}, we get that
  \(g(x_{p_{\ell+1}}) = g(x_{p'_{\ell+1}})\) and $c = c'$.
\qed\end{proof}

We are now ready the second part of \Cref{thm:extension-n-soundness}.

\begin{thmSoundnessPart}
  Let \(\tree\) be an observation tree for a \good \MMT \(\M\) with the
  functional simulation \(\funcSim\), \(p, p' \in Q^\tree\), and
  \(m : p \leftrightarrow p'\) be a matching.
  Then,
  \(p \apart^m p' \implies f(p) \neq f(p') \lor \exists x \in \dom{m} :
  g(x) \neq g(m(x))\).
\end{thmSoundnessPart}
\begin{proof}
  Towards a contradiction, assume
  \begin{itemize}
    \item
    \(w = i_1 \dotsb i_n \witness p \apart^m p'\),
    \item
    \(f(p) = f(p')\), and
    \item
    \(\forall x \in \dom{m} : g(x) = g(m(x))\).
  \end{itemize}
  Let us consider the two runs of \(\tree\) that showcase
  \(w \witness p \apart^m p'\):
  \begin{align*}
    \pi = p_0 \xrightarrow[u_1]{i_1/o_1} \dotsb \xrightarrow[u_n]{i_n/o_n} p_n
    &&\text{and}&&
    \pi' = \copyRun{p'_0}[\pi] =
      p'_0 \xrightarrow[u'_1]{i'_1/o'_1} \dotsb \xrightarrow[u'_n]{i'_n/o'_n} p'_n
  \end{align*}
  with \(p_0 = p\) and \(p'_0 = p'\).
  By \Cref{lemma:soudness:runsInM}, we thus have
  \(\funcSim(\pi) = \funcSim(\pi')\) and
  \(g(x) = g(\matchingRun{m}{\pi}{\pi'}(x))\) for all started timers
  \(x \in \dom{\matchingRun{m}{\pi}{\pi'}}\).
  In particular, the equality of the runs holds for the last transition:
  \begin{equation}\label{eq:soundness:equalRuns}
  f(p_{n-1}) \xrightarrow[u]{g(i_n)/o} f(p_n)
  =
  f(p'_{n-1}) \xrightarrow[u]{g(i'_n)/o} f(p'_n).
  \end{equation}
  Notice the same output and update, by determinism of \(\M\).

  First, if \(w \witness p_0 \apart^m p'_0\) is structural, then there must
  exist a timer \(x \in \dom{\matchingRun{m}{\pi}{\pi'}}\)
  such that \(x \timerApart \matchingRun{m}{\pi}{\pi'}(x)\).
  By definition of the timer apartness,
  \(x \neq \matchingRun{m}{\pi}{\pi'}(x)\) and there must
  exist a state \(q\) such that
  \(x, \matchingRun{m}{\pi}{\pi'}(x) \in \activeTimers(q)\).
  By~\eqref{eq:simulation:different_timers},
  \(g(x) \neq g(\matchingRun{m}{\pi}{\pi'}(x))\), which is a contradiction.

  Hence, assume \(w \witness p_0 \apart^m p'_0\) is behavioral.
  Let us show that each possibility leads to a contradiction.

  \paragraph{\eqref{eq:apartness:outputs}.}
  Assume~\eqref{eq:apartness:outputs} holds, i.e., \(o_n \neq o'_n\).
  By~\eqref{eq:simulation:transition} and~\eqref{eq:soundness:equalRuns},
  we get \(o_n = o = o'_n\), which is a contradiction.

  \paragraph{\eqref{eq:apartness:constants}.}
  Assume~\eqref{eq:apartness:constants} holds, i.e., \(u_n = (x, c)\)
  and \(u'_n = (x', c')\) with \(c \neq c'\).
  By~\eqref{eq:simulation:update:known},
  we get
  \begin{align*}
    f(p_n{n-1}) \xrightarrow[(g(x), c)]{g(i_n)} f(p_n)
    &&\text{and}&&
    f(p'_n) \xrightarrow[(g(x'), c')]{g(i_n)} f(p'_n).
  \end{align*}
  Thus, by~\eqref{eq:soundness:equalRuns}, \((g(x), c) = (g(x'), c')\), i.e.,
  \(c = c'\), which is a contradiction.

  \paragraph{\eqref{eq:apartness:sizesEnabled}.}
  Assume~\eqref{eq:apartness:sizesEnabled} holds, i.e.,
  \(p_n, p'_n \in \explored\) and
  \(\lengthOf{\enabled{p_n}[\tree]} \neq \lengthOf{\enabled{p'_n}[\tree]}\).
  By definition of \(\explored\), it follows that
  \begin{align*}
    \lengthOf{\enabled{p_n}[\tree]} = \lengthOf{\enabled{f(p_n)}[\M]}
    &&\text{and}&&
    \lengthOf{\enabled{p'_n}[\tree]} = \lengthOf{\enabled{f(p'_n)}[\M]}.
  \end{align*}
  That is,
  \(\lengthOf{\enabled{f(p_n)}[\M]} \neq \lengthOf{\enabled{f(p'_n)}[\M]}\),
  which is in contradiction with \(f(p_n) = f(p'_n)\).

  \paragraph{\eqref{eq:apartness:enabled}.}
  Assume~\eqref{eq:apartness:enabled} holds, i.e.,
  \(p_n, p'_n \in \explored\) and there exists
  \(x \in \dom{\matchingRun{m}{\pi}{\pi'}}\) such that
  \(x \in \enabled{p_n}[\tree]\) if and only if
  \(\matchingRun{m}{\pi}{\pi'}(x) \notin \enabled{p'_n}[\tree]\).

  It is not obvious that \(\matchingRun{m}{\pi}{\pi'}(x)\) belongs to the domain
  of \(g\), as \(\matchingRun{m}{\pi}{\pi'}(x)\) may never be started. We will
  argue that this timer is necessarily started. For now, assume that it belongs
  to the domain of \(g\).
  This leads to a contradiction since by the following derivation:
  \begin{eqnarray*}
  	x \in \enabled{p_n}[\tree] & \Leftrightarrow &
  	  (\mbox{by~\eqref{eq:simulation:enabled} and } p_n \in \explored)\\
  	g(x) \in \enabled{f(p_n)}[\M]  & \Leftrightarrow &
  	(\mbox{by \Cref{lemma:soudness:runsInM}})\\
  	(g(\matchingRun{m}{\pi}{\pi'}(x)) \in \enabled{f(p_n)}[\M])  & \Leftrightarrow & (\mbox{by~\eqref{eq:soundness:equalRuns}})\\
  	(g(\matchingRun{m}{\pi}{\pi'}(x)) \in \enabled{f(p'_n)}[\M])  & \Leftrightarrow & (\mbox{by~\eqref{eq:simulation:enabled} and } p'_n \in \explored)\\
  	\matchingRun{m}{\pi}{\pi'}(x) \in \enabled{p'_n}[\tree]. & &
  \end{eqnarray*}

  The last item is a contradiction with
  \(\matchingRun{m}{\pi}{\pi'}(x) \notin \enabled{p'_n}[\tree]\).
  That is, it suffices to show that \(\matchingRun{m}{\pi}{\pi'}(x) \in \dom{g}\),
  i.e., that \(\matchingRun{m}{\pi}{\pi'}(x)\) is active in some state of \(\tree\).
  We have two cases to consider.

  \begin{itemize}
    \item
    If \(x \in \dom{m}\), then, by definition, \(x \in \activeTimers^\tree(p_0)\)
    and \(\matchingRun{m}{\pi}{\pi'}(x) = m(x) \in \activeTimers^\tree(p'_0)\).
    Thus, \(\matchingRun{m}{\pi}{\pi'}(x) \in \dom{g}\).

    \item
    If \(x \notin \dom{m}\), then, there must exist some
    \(k \in \{0, \dotsc, n - 1\}\) such that \(x = x_{p_k} \in \enabled{p_n}[\tree]\).
    That is, \(u_k = (x_{p_k}, c)\) (with \(c \in \natplus\)).
    This necessarily means that \(i_k \in I\) (as we start a fresh timer in
    \(\tree\)).
    Hence, \(i'_k = i_k \in I\).
    Moreover, by definition of \(\matchingRun{m}{\pi}{\pi'}\), it must be that
    \(\matchingRun{m}{\pi}{\pi'}(x_{p_k}) = x_{p'_k}\)
    (recall that \(x = x_{p_k} \in \dom{\matchingRun{m}{\pi}{\pi'}}\) and
    \(x \notin \dom{m}\)).
    It remains to prove that \(x_{p'_k}\) is started.
    To do so, we will rely on the fact that a sub-run of \(\pi\) is
    \(x_{p_k}\)-spanning and deduce from there that the corresponding sub-run
    of \(\pi'\) must be \(x_{p'_k}\)-spanning, which can only hold when
    \(x_{p'_k}\) is effectively started.

    First, let us extend the run \(\pi\)
    by adding the transition reading \(\timeout{x_{p_k}}\) from \(p_n\):
    \(
      p_0
      \xrightarrow{i_1} \dotsb
      \xrightarrow[(x_{p_k}, c)]{i_k} p_k
      \xrightarrow{i_{k+1}} \dotsb
      \xrightarrow{i_n} p_n
      \xrightarrow{i_{n+1}} p_{n+1}
    \),
    with \(i_{n+1} = \timeout{x_{p_k}}\).
    This run necessarily exists, as \(x_{p_k} \in \enabled{p_n}[\tree]\).

    Second, we consider the sub-run of \(\pi\) from \(p_{k-1}\) that is
    \(x_{p_k}\)-spanning:
    observe that there exists some index \(\ell \in \{k + 1, \dotsc, n+1\}\)
    such that \(i_{\ell} = \timeout{x_{p_k}}\) and the run
    \[
      \sigma = p_{k-1}
      \xrightarrow[(x_{p_k}, c)]{i_k} p_k
      \xrightarrow{i_{k+1}} \dotsb
      \xrightarrow{i_\ell} p_{\ell}
    \]
    is \(x_{p_k}\)-spanning.
    (Notice that \(\ell\) is the index of the first timeout of \(x_{p_k}\) in
    \(\sigma\), by definition of a spanning run.
    Since \(i_\ell = \timeout{x_{p_k}}\), we may have \(\ell = n + 1\).)
    By~\eqref{eq:simulation:spanning:fromTree}, it follows that
    \[
      \funcSim(\sigma) = f(p_{k-1})
      \xrightarrow[(g(x_{p_k}), c)]{g(i_k)} f(p_k)
      \xrightarrow{g(i_{k+1})} \dotsb
      \xrightarrow{g(i_\ell)} f(p_\ell)
    \]
    is \(g(x_{p_k})\)-spanning.

    Third, by~\eqref{eq:soundness:equalRuns}, there must exist some timer
    \(y \in \enabled{f(p'_{\ell - 1})}[\M]\) such that
    \[
      \lambda = f(p'_{k-1})
      \xrightarrow[(y, c)]{g(i'_k)} f(p'_k)
      \xrightarrow{g(i'_{k+1})} \dotsb
      \xrightarrow{g(i'_\ell)} f(p'_\ell)
    \]
    is \(y\)-spanning.
    Note that \(y = g(x_{p_k})\).

    Fourth, notice that \(p_{\ell - 1}, p'_{\ell - 1} \in \explored\), as
    \(p_n, p'_n \in \explored\), \(\ell - 1 \leq n\), and \(\explored\) forms
    a subtree of \(\tree\).
    Therefore, there exists \(z \in \enabled{p'_{\ell - 1}}[\tree]\) such that
    \(g(z) = y = g(x_{p_k})\).
    Moreover, as \(\pi' = \copyRun{p'_0}[\pi]\) is already known to be an
    existing run of \(\tree\), it follows that
    \[
      \sigma' =
      p'_{k-1}
      \xrightarrow[(z, c)]{i'_k} p'_k
      \xrightarrow{i'_{k+1}} \dotsb
      \xrightarrow{i'_\ell} p'_\ell
    \]
    is a run of \(\tree\), and \(\funcSim(\sigma') = \lambda\).

    Finally, by~\eqref{eq:simulation:spanning},
    since \(\funcSim(\sigma')\) is \(g(z)\)-spanning,
    \(\sigma'\) is \(z\)-spanning, i.e., \(i'_\ell = \timeout{z}\).
    As \(i'_k = i_k \in I\), we necessarily have \(z = x_{p'_k}\), i.e.,
    \(x_{p'_k}\) is started and, thus, belong to \(\dom{m}\), as we wanted.
  \end{itemize}
  Every case leads to a contradiction.
  So, \(f(p) \neq f(p')\) or \(g(x) \neq g(m(x))\) for some \(x \in \dom{m}\).
\qed\end{proof}

\subsection{Weak co-transitivity}\label{app:tree:cotransitivity}

In \lsharp~\cite{VaandragerGRW22}, the learner can reduce the
number of hypotheses that can be constructed from a tree by exploiting the
\emph{weak co-transitivity} lemma, stating that, if we can read a witness
\(w\) of the apartness (defined for classical Mealy machines) \(p \apart p'\)
from some state \(r\), then \(p \apart r\) or \(p' \apart r\) (or both).
Hence, before folding the tree to obtain a hypothesis, it is possible to ensure
that each frontier state can be mapped to a single basis state.

For \MMTs, this is trickier, as we have to take into account the timers
and the mappings.
That is, our version of
the weak co-transitivity lemma states that if we can read a witness \(w\) of the
\emph{\behavioral} apartness \(p_0 \apart^{m} p'_0\) from a third state \(r_0\)
via some matching \(\mu : p_0 \leftrightarrow r_0\), then we can conclude that
\(p_0\) and \(r_0\) are
\(\mu\)-apart or that \(p'_0\) and \(r_0\) are ($\mu \circ m^{-1}$)-apart.
However, when \(p_0 \apart^m p'_0\) is due
to~\eqref{eq:apartness:constants}, we need to extend the witness.
In this case, since \(x\) is active in \(p_n\) (as $u = (x,c)$),
there must exist an \(x\)-spanning run
\(p_{n-1} \xrightarrow[u]{i_n/o} p_n \xrightarrow{w^x}\)
(by definition of an observation tree, see \Cref{def:tree}).
Hence, we actually \enquote{read} \(w \cdot w^x\) from \(r_0\),
in order to ensure that an update \((x', c')\) is present on the last
transition of \(\copyRun[\mu]{r_0}[p_0 \xrightarrow{w}]\).

Notice that the lemma requires that \(\dom{m} \subseteq \dom{\mu}\) for the
matching \(\mu \circ m^{-1}\).
See \Cref{fig:matching:composition} for
illustrations of a well- and an ill-defined \(\mu \circ m^{-1}\).

\begin{figure}[t]
  \centering
  \begin{subfigure}{.45\textwidth}
    \centering
    \begin{tikzpicture}[
  timer graph,
  set of timers/.append style = {
    inner xsep = 4pt,
    inner ysep = 0pt,
  },
  node distance = 10pt,
]
  \node [timer]                 (r x1) {};
  \node [timer, below=of r x1]  (r x2) {};
  \node [timer, below=of r x2]  (r x3) {};

  \path
    let
      \p{dist} = (1.3, 0),
      \p1 = ($(r x1)!.5!(r x2) + (\p{dist})$),
      \p2 = ($(r x1) + 2*(\p{dist})$)
    in
      node [timer]                        (p x1)  at (\p1) {}
      node [timer]                        (pp x1) at (\p2) {}
  ;
  \node [timer, below=of p x1]  (p x2)  {};

  \node [timer, below=of pp x1] (pp x2) {};
  \node [timer, below=of pp x2] (pp x3) {};

  \node [
    set of timers,
    label=left:$\activeTimers(r_0)$,
    fit=(r x1) (r x2) (r x3)
  ] {};
  \node [
    set of timers,
    label=above:$\activeTimers(p_0)$,
    fit=(p x1) (p x2)
  ] {};
  \node [
    set of timers,
    label=right:$\activeTimers(p'_0)$,
    fit=(pp x1) (pp x2) (pp x3)
  ] {};

  \path [mapping]
    (p x1)  edge (pp x2)
  ;
  \path [mapping, dashed]
    (p x1)  edge (r x1)
    (p x2)  edge (r x3)
  ;
  \path [mapping, densely dotted]
    (pp x2) edge [out=-150, in=-40, looseness=0.3] (r x1)
  ;
\end{tikzpicture}
  \caption{Well-defined.}
  \end{subfigure}
  \hfill
  \begin{subfigure}{.5\textwidth}
    \centering
    \begin{tikzpicture}[
  timer graph,
  set of timers/.append style = {
    inner xsep = 4pt,
    inner ysep = 0pt,
  },
  node distance = 10pt,
]
  \node [timer]                 (r x1) {};
  \node [timer, below=of r x1]  (r x2) {};
  \node [timer, below=of r x2]  (r x3) {};

  \path
    let
      \p{dist} = (1.3, 0),
      \p1 = ($(r x1)!.5!(r x2) + (\p{dist})$),
      \p2 = ($(r x1) + 2*(\p{dist})$),
      \p3 = ($(r x3) + .5*(\p{dist})$)
    in
      node [timer]                        (p x1)  at (\p1) {}
      node [timer]                        (pp x1) at (\p2) {}
      node [inner sep=0pt, outer sep=0pt] (no)    at (\p3) {?}
  ;
  \node [timer, below=of p x1]  (p x2)  {};

  \node [timer, below=of pp x1] (pp x2) {};
  \node [timer, below=of pp x2] (pp x3) [label=right:$x$] {};

  \node [
    set of timers,
    label=left:$\activeTimers(r_0)$,
    fit=(r x1) (r x2) (r x3)
  ] {};
  \node [
    set of timers,
    label=above:$\activeTimers(p_0)$,
    fit=(p x1) (p x2)
  ] {};
  \node [
    set of timers,
    inner xsep = 12pt,
    label=right:$\activeTimers(p'_0)$,
    fit=(pp x1.east) (pp x2.east) (pp x3.east)
  ] {};

  \path [mapping]
    (p x1)  edge (pp x2)
    (p x2)  edge (pp x3)
  ;
  \path [mapping, dashed]
    (p x1)  edge (r x1)
;
  \path [mapping, densely dotted]
    (pp x2) edge [out=-150, in=-40, looseness=0.3] (r x1)
    (pp x3) edge [-, bend left=20] (no)
  ;
\end{tikzpicture}
  \caption{Ill-defined: \((\mu \circ m^{-1})(x)\) has no value.}
  \end{subfigure}
  \caption{Visualizations of compositions \(\mu \circ m^{-1}\) where \(m\) is
  drawn with solid lines, \(\mu\) with dashed lines, and \(\mu \circ m^{-1}\)
  with dotted lines.}\label{fig:matching:composition}
\end{figure}

\begin{lemma}[Weak co-transitivity]\label{lemma:coTransitivity}
    Let \(p_0, p'_0, r_0 \in Q^\tree\), \(m : p_0 \leftrightarrow p'_0\) and
  \(\mu : p_0 \leftrightarrow r_0\) be two matchings such that
  \(\dom{m} \subseteq \dom{\mu}\).
  Let \(w = i_1 \ldots i_n\) be a witness of the \behavioral apartness
  \(p_0 \apart^m p'_0\) and
  $\copyRun[m]{p'_0}[p_0 \xrightarrow{w} p_n] = p'_0 \xrightarrow{w'} p'_n$.
  Let \(w^x\) be defined as follows:
  \begin{itemize}
    \item if \(p_0 \apart^m p'_0\) due to~\eqref{eq:apartness:constants}, $w^x$
    is a word such that \(p_{n-1} \xrightarrow{i_n} p_n \xrightarrow{w^x}\)
    is $x$-spanning,
    \item otherwise, \(w^x = \emptyword\). 
  \end{itemize}
  If $\copyRun[\mu]{r_0}[p_0 \xrightarrow{w \cdot w^x}] \in \runs{\tree}$
  with $r_n \in \explored$, then
  \(p_0 \apart^{\mu} r_0\) or \(p'_0 \apart^{\mu \circ m^{-1}} r_0\).
\end{lemma}
\begin{proof}
  Let \(p_0, p'_0, r_0 \in Q^\tree\), and \(m : p_0 \leftrightarrow p'_0\) and
  \(\mu : p_0 \leftrightarrow r_0\) be two matchings such that
  \(\dom{m} \subseteq \dom{\mu}\).
  Let \(w \cdot w^x\), \(w'\), and the runs as described in the statement,
  \(n = \lengthOf{w}\), and \(\ell = \lengthOf{w \cdot w^x}\).
  Moreover, let \(v\) be the word labeling the run from \(p_0\), i.e., such that
  \(\copyRun[\mu]{r_0}[p_0 \xrightarrow{w \cdot w^x}] =
  r_0 \xrightarrow{v} r_\ell\).
  We write $w_j$ (resp.\ $w'_j$, $v_j$) for a symbol of $w \cdot w^x$
  (resp.\ $w'$, $v$).
  (So, \(n = \ell\) whenever \(w \witness p_0 \apart^m p'_0\) due to a condition
  that is not~\eqref{eq:apartness:constants}, and \(n < \ell\) otherwise.)
  We then have
  \begin{align*}
    &p_0 \xrightarrow{w_1} p_1 \xrightarrow{w_2} \dotsb
    \xrightarrow{w_n} p_n
    \xrightarrow{w_{n+1}} \dotsb \xrightarrow{w_\ell} p_\ell,
    \\
    \copyRun{p'_0} ={}
    &p'_0 \xrightarrow{w'_1} p'_1 \xrightarrow{w'_2} \dotsb
    \xrightarrow{w'_n} p'_n,
    \\
    \copyRun[\mu]{r_0}[p_0 \xrightarrow{w \cdot w^x} p_\ell] ={}
    &r_0 \xrightarrow{v_1} r_1 \xrightarrow{v_2} \dotsb
    \xrightarrow{v_n} r_n
    \xrightarrow{v_{n+1}} \dotsb \xrightarrow{v_\ell} r_\ell,
    \\
    \copyRun[\mu \circ m^{-1}]{r_0}[p'_0 \xrightarrow{w'} p'_n] =
    \copyRun[\mu]{r_0} ={}
    & r_0 \xrightarrow{v_1} r_1 \xrightarrow{v_2} \dotsb
    \xrightarrow{v_n} r_n
  \end{align*}
  with \(r_n \in \explored\), by hypothesis.
  The run from \(p'_0\) does not read \(w^x\) after $p'_n$.
  
  A first possibility is that $p_0 \apart^{\mu} r_0$ or $p'_0 \apart^{\mu \circ m^{-1}} r_0$ due to \structural apartness (with $w \cdot w^x$ or $w'$ as witness). If this does not happen, from $w$ being a witness of the \behavioral apartness \(p_0 \apart^m p'_0\), we have to show that \(p_0 \apart^\mu r_0\) or \(p'_0 \apart^{\mu \circ m^{-1}} r_0\) for 
  one case among~\eqref{eq:apartness:outputs},~\eqref{eq:apartness:constants},~\eqref{eq:apartness:sizesEnabled}, or~\eqref{eq:apartness:enabled}.
  We do it by a case analysis.
  Let \(o, o', \omega \in O\) such that
  \(p_{n-1} \xrightarrow{w_n/o} p_n, p'_{n-1} \xrightarrow{w'_n/o'} p'_n\),
  and \(r_{n-1} \xrightarrow{v_n/\omega} r_n\).
  \begin{itemize}
    \item
    If \(o \neq o'\), then, necessarily, \(\omega \neq o\) or \(\omega \neq o'\)
    and we can apply~\eqref{eq:apartness:outputs} to obtain \(p_0 \apart^\mu r_0\)
    or \(p'_0 \apart^{\mu \circ m^{-1}} r_0\).
    \item
    If \(\lengthOf{\enabled{p_n}} \neq \lengthOf{\enabled{p'_n}}\), then,
    necessarily, \(\lengthOf{\enabled{r_n}} \neq \lengthOf{\enabled{p_n}}\)
    or \(\lengthOf{\enabled{r_n}} \neq \lengthOf{\enabled{p'_n}}\).
    As \(p_n, p'_n, r_n \in \explored\), we can
    apply~\eqref{eq:apartness:sizesEnabled} and get \(p_0 \apart^\mu r_0\)
    or \(p'_0 \apart^{\mu \circ m^{-1}} r_0\).
    \item
    Suppose now that \(p_0 \apart^m p'_0\) is due to~\eqref{eq:apartness:enabled}. 
    \begin{itemize}
    \item If \(x \in \dom{m}\) and
    \(x \in \enabled{p_n} \iff m(x) \notin \enabled{p'_n}\), then
     \(x \in \dom{\mu}\) and, necessarily, depending on whether
     $\mu(x) \in \enabled{r_n}$ or $\mu(x) \notin \enabled{r_n}$, we have either
    \(x \in \enabled{p_n} \iff \mu(x) \notin \enabled{r_n}\)
    or
    \(m(x) \in \enabled{p'_n} \iff \mu(m^{-1}(m(x))) = \mu(x) \notin \enabled{r_n}\).
    Hence,~\eqref{eq:apartness:enabled} applies
    (as \(p_n, p'_n, r_n \in \explored\)).
    \item
    If \(x_{p_k} \in \enabled{p_n} \iff x_{p'_k} \notin \enabled{p'_n}\) for some
    \(k \in \{1, \dotsc, n\}\), we conclude with arguments similar to the
    previous case that~\eqref{eq:apartness:enabled} is also satisfied.
    \end{itemize}
    \item
    Finally, if none of the above holds, \(p_0 \apart^m p'_0\) is
    due to~\eqref{eq:apartness:constants}.
    We thus have
    \(p_{n-1} \xrightarrow[(x, c)]{w_n} p_n
      \xrightarrow{w_{n+1}} \dotsb \xrightarrow{w_\ell} p_\ell\),
    \(p'_{n-1} \xrightarrow[(x', c')]{w'_n} p'_n\), and
    \(r_{n-1} \xrightarrow[u]{v_n} r_n
      \xrightarrow{v_{n+1}} \dotsb \xrightarrow{v_\ell} r_\ell\)
    with \(c \neq c'\) and \(w_\ell = \timeout{x}\).
    Finally, let
    \begin{align*}
      y &= \begin{cases}
        \mu(x) & \text{if \(x \in \dom{m}\)}\\
        x_{r_k} & \text{if \(x = x_{p_k}\) for some \(k \in \{1,\ldots,n\}\)}
      \end{cases}
    \shortintertext{which means that}
      v_\ell = \timeout{y} &= \begin{cases}
        \timeout{\mu(x)} & \text{if \(x \in \dom{m}\)}\\
        \timeout{x_{r_k}} & \text{if \(x = x_{p_k}\) for some
          \(k \in \{1, \dotsc, n\}\).}
      \end{cases}
    \end{align*}

    We argue that \(u = (y, d)\) for some constant \(d\) that is distinct from
    either \(c\) or \(c'\).
    Once we have this,~\eqref{eq:apartness:constants} applies and we obtain the
    desired result.
    We have three cases:
    \begin{itemize}
      \item
      If \(w_n \in I\), it must be that \(x = x_{p_n}\) as an input transition
      can only start a fresh timer in \(\tree\).
      Then, \(v_n = w_n\) as \(w_n \in I\), \(w_\ell = \timeout{x_{p_n}}\), and
      \(y = x_{r_n}\).
      So, \(v_\ell = \timeout{x_{r_n}}\).
      Moreover, as the only transition that can start \(x_{r_n}\) for the first
      time is \(r_{n-1} \xrightarrow[u]{v_{n}} r_n\), we conclude that
      \(u = (y, d) = (x_{r_n}, d)\).
      \item
      If \(w_n = \timeout{x}\) with \(x \in \dom{m}\), then \(x \in \dom{\mu}\),
      \(w_n = w_\ell = \timeout{x}\), and
      \(v_n = v_\ell = \timeout{y} = \timeout{\mu(x)}\).
      Assume \(u = \bot\), i.e., we do not restart \(y\) from \(r_{n-1}\) to
      \(r_n\).
      In other words, \(y\) is not active in \(r_n\).
      Recall that, in an observation tree, it is impossible to start again a timer
      that was previously active (as, for every timer \(z\), there is a unique
      transition that can start \(z\) for the first time).
      So, \(y\) can not be active in \(r_{\ell-1}\).
      But, then, \(v_\ell\) can not be \(\timeout{y}\), which is a contradiction.
      Hence, \(u = (y, d) = (\mu(x), d)\).
      \item
      If \(w_n = \timeout{x_{p_k}}\) with \(k \in \{1,\ldots,n-1\}\), then
      \(w_n = w_\ell = \timeout{x_{p_k}}\) and \(v_n = v_\ell = \timeout{x_{r_k}}\).
      With arguments similar to the previous case, we conclude that \(u = (y, d)
      = (x_{r_k}, d)\).
    \end{itemize}
  \end{itemize}
\qed\end{proof}
 
\section{More details on the learning algorithm}\label{app:learning}

In this section, we give further details on the ideas introduced in
\Cref{sec:learning:hypo,sec:learning:algo}.
We first clarify the notion of \emph{replaying} a run from \Cref{ex:learning:basis}
and state a useful property.
We then introduce generalized \MMTs and explain how to construct one from \(\tree\)
in \Cref{app:learning:hypo}.
Finally, we prove \Cref{thm:learning:termination} in \Cref{app:learning:termination}.

\subsection{Replaying a run}\label{app:learning:replay}

Recall that~\ref{item:tree:basis:timers} requires that, for every
\(r \in \frontier\) and \((p, m) \in \compatible(r)\),
\(\lengthOf{\activeTimers^\tree(p)} = \lengthOf{\activeTimers^\tree(r)}\) holds.
In \Cref{ex:learning:basis}, we introduced the idea of replaying a run to
ensure that~\ref{item:tree:basis:timers} is satisfied: if \(p\) and \(r\)
do not have the same number of active timers, we extend the tree by mimicking
the run showing that some timer is active in \(p\) (i.e., a run starting from \(p\)
and ending in the timeout of that timer) from \(r\), using and extending
some matching.

Let us first formalize this algorithm.
Let \(p_0, p'_0 \in Q^\tree\), \(m : p_0 \leftrightarrow p'_0\) be a
matching, and \(w = i_1 \dotsb i_n\) be a word such that
\(p_0 \xrightarrow{i_1} p_1 \xrightarrow{i_2} \dotsb \xrightarrow{i_n} p_n
\in \runs{\tree}\).
We provide a function \(\replay{p'_0}\) that extends the tree by replaying the
run $p_0 \xrightarrow{w} p_n$ from $p'_0$ as much as possible, 
or we discover a new apartness pair
\(p_0 \apart^m p'_0\), or we discover a new active timer.
Intuitively, we replay the run transition by transition while performing
symbolic wait queries in every reached state in order to determine the enabled
timers (which extends $\explored$).
This may modify the number of active timers of \(p'_0\), meaning that \(m\) may
become non-maximal.
As we are only interested in maximal matchings, we stop early.
This may also induce a new apartness pair \(p_0 \apart^m p'_0\), and we also 
stop early (notice that this may already hold
without adding any state in \(\tree\)).
If the number of active timers of \(p'_0\) remains unchanged and no new apartness pair is discovered,
we consider the next symbol \(i\) of \(w\) and try to replay it.
Determining the next symbol \(i'\) to use in the run from \(p'_0\) follows the
same idea as for \(\copyRun{p'_0}\).
If \(i \in I\), then \(i' = i\) (recall that it is always possible to replay \(i\)
as \(\M\) is complete, since it is \good).
If \(i = \timeout{x}\), we have three cases:
\begin{itemize}
  \item \(x \in \dom{m}\), in which case \(i' = \timeout{m(x)}\);
  \item \(x = x_{p_k}\) is a fresh timer, i.e., \(p_k\) appears on the run from
  \(p_0\), in which case we consider the timer started on the corresponding
  transition from \(p'_0\): \(i' = \timeout{x_{p'_k}}\);
  \item none of the previous case holds:
  \(x \in \activeTimers^\tree(p_0) \setminus \dom{m}\) and we cannot replay \(i\).
\end{itemize}
To avoid this last case, we consider the longest prefix \(v\) of \(w\)
where each action \(i\) of \(v\) is an input or is such that \(m(i)\) is defined
or \(i = \timeout{x_{p_k}}\) for some state \(p_k\).

Formally, assume that we already replayed
\(p_0 \xrightarrow{i_1} p_1 \xrightarrow{i_2}
\dotsb \xrightarrow{i_{j-1}} p_{j-1}\) and obtained the run
\(p'_0 \xrightarrow{i'_1} p'_1 \xrightarrow{i'_2} \dotsb \xrightarrow{i'_{j-1}}
p'_{j-1}\), and we try to replay $i_j$ from $p'_{j-1}$.
We extend the tree with a symbolic output query when \(i_j \in I\)
and a symbolic wait query in every case.
If the wait query leads to a discovery of new active timers of \(p'_0\), we stop
and return \(\ACTIVE\).
If we can already deduce \(p_0 \apart^m p'_0\) from the replayed part, we also
stop and return \(\APART\).
Since \(\lnot (p_0 \apart^m p'_0)\) and by the output and wait queries, there
must exist \(p_{j-1} \xrightarrow{i'_j}\) such that
\begin{itemize}
  \item
  \(i'_j = i_j\) if \(i_j \in I\),
  \item
  \(i'_j = \timeout{m(x)}\) if \(i_j = \timeout{x}\) (\(m(x)\) is well-defined
  by the considered prefix \(v\) of \(w\)), or
  \item \(i'_j = \timeout{x_{p'_k}}\).
\end{itemize}
Indeed, if the timeout-transition is not defined, then
\(p_0 \apart^m p'_0\) by~\eqref{eq:apartness:enabled}.
Hence, we continue the procedure with the next symbol of \(w\).
If we completely replayed \(w\) and did not discover any new timer or apartness
pair, we return \DONE.
Otherwise, we perform one last wait query and check whether we obtain apartness
(by the following lemma, we return \ACTIVE otherwise).

We now give a lemma stating some properties of the replay function.
Namely, when \(\replay[m]{p_0}[\pi]\) successfully replays the complete run, it
follows that \(\copyRun[m]{p_0}[\pi]\) is well-defined and yields a run in
\(\tree\).
From that, we can deduce that, when \(\pi\) ends with a transition reading
the timeout of \(x\) and cannot be completely reproduced from \(p'_0\) (that has
less active timers than \(p_0\)), we either
have a new apartness pair or a new timer in \(p'_0\).
These properties are enough to conclude that one can
obtain~\ref{item:tree:basis:timers} (from page~\pageref{item:tree:basis:timers})
by successively replaying some runs, as the
set of compatible states of any frontier state \(r\) will eventually only contain
basis states with the same number of active timers as \(r\).

\begin{lemma}\label{lemma:replay:done}
  Let \(\pi = p_0 \xrightarrow{w} {} \in \runs{\tree}\), \(p'_0 \in Q^\tree\),
  and \(m : p_0 \leftrightarrow p'_0\) be a maximal matching:
  \begin{itemize}
    \item
    \(\replay{p'_0}[\pi] = \DONE\) implies that \(\copyRun{p'_0}[\pi]\)
    is now a run of \(\tree\).
    \item
    \(\replay{p'_0}[\pi]\) is \APART or \ACTIVE when
    \(\lengthOf{\activeTimers^\tree(p_0)} > \lengthOf{\activeTimers^\tree(p'_0)}\)
    and \(w\) ends with \(\timeout{x}\) for some
    \(x \in \activeTimers^\tree(p_0) \setminus \dom{m}\).
  \end{itemize}
\end{lemma}
\begin{proof}
  Observe that the second item follows immediately from the first, given the
  fact that we process a proper prefix of \(w\) in that case.
  That is, it is sufficient to show the first item.

  Let \(w = i_1 \dotsb i_n\) and
  \(
    \pi = p_0 \xrightarrow{i_1}
    p_1 \xrightarrow{i_2} \dotsb
    \xrightarrow{i_n} p_n
  \).
  Towards a contradiction, assume that \(\replay{p'_0}[\pi] = \DONE\) but
  \(\copyRun{p'_0}[\pi]\) is not a run of \(\tree\).
  Then, let \(\ell \in \{1, \dotsc, n - 1\}\) be the largest index such that
  \(\copyRun{p'_0}[p_0 \xrightarrow{i_1 \dotsb i_\ell}] =
    p'_0 \xrightarrow{i'_1}
    p'_1 \xrightarrow{i'_2}
    \dotsb \xrightarrow{i'_\ell}
    p'_\ell
    \in \runs{\tree}\).
  Hence,
  \(p_0 \xrightarrow{i_1} p_1 \xrightarrow{i_2} \dotsb \xrightarrow{i_\ell}
  p_\ell \xrightarrow{i_{\ell + 1}} {} \in \runs{\tree}\) and
  \(\copyRun{p'_0}[p_0 \xrightarrow{i_1 \dotsb i_\ell \cdot i_{\ell + 1}}] =
  p'_0 \xrightarrow{i'_1} p'_1 \xrightarrow{i'_2} \dotsb \xrightarrow{i'_\ell}
  p'_\ell \xrightarrow{i'_{\ell + 1}} {} \notin \runs{\tree}\).
  First, if \(i_{\ell + 1} \in I\), then we must have performed a symbolic
  output query in \(p'_\ell\), i.e.,
  \(p'_\ell \xrightarrow{i'_{\ell + 1}} {} \in \runs{\tree}\).
  Second, if \(i_{\ell + 1} = \timeout{x_{p_k}}\) for some
  \(k \in \{1, \dotsc, \ell\}\), then we have that \(p_0 \apart^m p'_0\)
  by~\eqref{eq:apartness:enabled}.
  Likewise when \(i_{\ell + 1} = \timeout{x}\) with \(x \in \dom{m}\).

  Therefore, assume \(i_{\ell + 1}\) is the timeout of some timer in
  \(\activeTimers^\tree(p_0) \notin \dom{m}\).
  Since \(\replay{p'_0}[\pi] = \DONE\), we have that
  \(\lnot (p_0 \apart^m p'_0)\)
  and we did not discover a new active timer in \(p'_0\).
  Hence,
  \begin{gather}
    \lengthOf{\enabled{p_{\ell}}[\tree]} = \lengthOf{\enabled{p'_{\ell}}[\tree]}
    \label{eq:proof:lemma:replay:done:enabled}
    \\
    \forall y \in \dom{m} :
    y \in \enabled{p_{\ell}}[\tree] \iff
    m(y) \in \enabled{p'_{\ell}}[\tree],
    \label{eq:proof:lemma:replay:done:timer}
    \\
    \forall k \in \{1, \dotsc, \ell\} :
    x_{p_k} \in \enabled{p_{\ell}}[\tree] \iff
    x_{p'_k} \in \enabled{p'_{\ell}}[\tree],
    \label{eq:proof:lemma:replay:done:fresh}
  \end{gather}
  As $m$ is maximal, we deduce from~\eqref{eq:proof:lemma:replay:done:timer}
  and~\eqref{eq:proof:lemma:replay:done:fresh} that all enabled timers in
  $p'_\ell$ have their corresponding enabled timer in $p_\ell$.
  However, $x$ is an enabled timer in $p_\ell$ that does not appear among those
  corresponding timers as $x \not\in \dom{m}$.
  This is in contradiction with~\eqref{eq:proof:lemma:replay:done:enabled}.
  We thus conclude that \(\replay{p'_0}[\pi] \neq \DONE\).
\qed\end{proof}

\subsection{Generalized MMTs and hypothesis construction}\label{app:learning:hypo}

In this section, we introduce generalized \MMTs (that allow timer renamings alongside
the transitions), and show that a symbolically equivalent \MMT always exists.
This \MMT suffers a factorial blowup, in general.
We then give the construction of a generalized \MMT from \(\tree\), and give
an example where the construction of an \MMT as explained in
\Cref{ex:learning:hypothesis} fails, as the equivalence relation groups
together timers that are known to be apart.

\subsubsection{Generalized MMTs.}

In short, a generalized \MMT is similar to an \MMT, except that the update of
a transition \(q \xrightarrow{i} q'\) is now a function instead of a value in
\((X \times \natplus) \cup \{\bot\}\).
For the \gMMT to be well-formed, we request that the
domain of such a function is exactly the set of active timers of \(q'\).
Moreover, its range must be the set of active timers of \(q\) or a natural
constant.
That is, each timer \(x'\) of \(q'\) must either come from an active timer \(x\)
of \(q\) (we rename \(x\) into \(x'\)), or be (re)started with a constant.
We also require that at most one timer is started per transition, as in
\MMTs.
Finally, if \(i = \timeout{x}\), we forbid to rename \(x\) into \(x'\), i.e.,
\(x'\) cannot be obtained from \(x\): it must be the renaming of some other timer
or be explicitly started by the transition.
An example is given below.

\begin{definition}[gMMT]\label{def:gMMT}
  A \emph{generalized Mealy machine with timers} (gMMT) is a 5-tuple
  \(\M = (X, Q, q_0, \activeTimers, \delta)\) where:
  \begin{itemize}
    \item
    \(X\) is a finite set of timers (we assume $X \cap \natplus = \emptyset$),
    \item
    \(Q\) is a finite set of states, with \(q_0 \in Q\) the \emph{initial state},
    \item
    \(\activeTimers : Q \to \subsets{X}\) is a total function that assigns a
    finite set of active timers to each state,
    and
    \item
    \(\delta : Q \times \actions{\M} \partto Q \times O \times (X \to (X \cup \natplus)) \)
    is a partial transition function that assigns a state-output-update triple to a
    state-action pair.
  \end{itemize}
  We write \(q \xrightarrow[\updateFunction]{i/o} q'\) if
  \(\delta(q, i) =(q', o, \updateFunction)\).
  We require the following:
  \begin{itemize}
    \item
    In the initial state, no timer is active, i.e.,
    \(\activeTimers(q_0) = \emptyset\).
    \item
    For any transition \(q \xrightarrow[\updateFunction]{} q'\),
    \(\updateFunction\) must be an injective function whose domain is exactly
    the set of active timers of \(q'\) (i.e.,
    \(\dom{\updateFunction} = \activeTimers\)), and whose range is composed
    of timers that were active in \(q\) or constants from \(\natplus\)
    (i.e., \(\ran{\updateFunction} \subset \activeTimers(q) \cup \natplus\)).
    Finally, there is at most one (re)started timer, i.e., there is at most
    one \(x \in \dom{\updateFunction}\) such that
    \(\updateFunction(x) \in \natplus\).
    \item
    For any transition \(q \xrightarrow[\updateFunction]{\timeout{x}} q'\),
    it must be that \(x\) was active in \(q\) and \(x\) cannot be used as a value
    of \(\updateFunction\), i.e., \(x \in \activeTimers(q)\) and
    \(x \notin \ran{\updateFunction}\).
  \end{itemize}
\end{definition}
Observe that an \MMT is in fact a \gMMT where all renaming maps on transitions
coincide with the identity function
(except for those mapping to an integer, which are regular updates).

We now adapt the timed semantics of the model via the following rules.
Again, they are similar to the rules for \MMTs, except that we use
\(\updateFunction\)
to rename and start timers.
Let \((q, \valuation), (q', \valuation')\) be two configurations of a \gMMT:
\begin{gather*}
\infer{(q, \valuation) \xrightarrow{d} (q, \valuation - d)}
{\forall x \colon \valuation(x) \geq d}
\\
\infer{(q, \valuation) \xrightarrow[\updateFunction]{i/o} (q', \valuation')}
{q \xrightarrow[\updateFunction]{i/o} q', & i = \timeout{x} \Rightarrow \valuation(x)
  = 0, & \forall x \in \chi(q') \colon \valuation'(x) = \begin{cases}
  \updateFunction(x) & \mbox{if } \updateFunction(x) \in\natplus\\
  \valuation(\updateFunction(x)) & \mbox{otherwise}
  \end{cases}
}
\end{gather*}
We immediately obtain the definitions of enabled timers and \complete \gMMT.
Moreover, it is clear that the notion of timed equivalence from
\Cref{sec:MMT:equivalence} can be applied to two \complete \gMMTs,
or a \gMMT and an \MMT, both \complete.

\begin{figure}[t]
  \centering
  \begin{tikzpicture}[
  automaton,
  node distance = 35pt and 80pt,
]
  \node [state, initial]      (q0)  {\(q_0\)};
  \node [state, right=of q0]  (q1)  {\(q_1\)};
  \node [state, below=of q1]  (q2)  {\(q_2\)};
  \node [state, right=of q1]  (q3)  {\(q_3\)};
  \node [state, right=of q3]  (q4)  {\(q_4\)};

  \path
    (q0)  edge  node [above] {\(i/o\)}
                node [below] {\(\updateFig{x}{2}\)} (q1)
    (q1)  edge  node [right, pos=0.6] {\(\timeout{x}/o, \updateFig{y}{2}\)} (q2)
          edge  node [above] {\(i/o\)}
                node [below] {\(\updateFig{y}{1}\)} (q3)
    (q2)  edge [loop left] node {\(i/o, \updateFig{y}{2}\)} ()
          edge [loop right] node {\(\timeout{y}/o, \updateFig{y}{2}\)} ()
    (q3)  edge [loop above] node {\(i/o, \updateFig{y}{x}, \updateFig{x}{2}\)} ()
          edge [bend left=10] node {\(\timeout{x}/o, \updateFig{y}{1}\)} (q4)
          edge [bend right=10] node ['] {\(\timeout{y}/o, \updateFig{y}{1}\)} (q4)
    (q4)  edge [loop above] node {\(i/o, \updateFig{y}{1}\)} ()
          edge [loop below] node {\(\timeout{y}/o, \updateFig{y}{1}\)} ()
  ;
\end{tikzpicture}
  \caption{A generalized \MMT with \(\activeTimers(q_0) = \emptyset\),
  \(\activeTimers(q_1) = \{x\}, \activeTimers(q_2) = \activeTimers(q_4) = \{y\}\),
  and \(\activeTimers(q_3) = \{x, y\}\).}\label{fig:gMMT}
\end{figure}
\begin{example}
  Let \(\M\) be the \gMMT of \Cref{fig:gMMT} with timers \(X = \{x, y\}\).
  Update functions are shown along each transition.
  For instance, \(x\) is started to \(2\) by the transition from \(q_0\) to
  \(q_1\), while the self-loop over \(q_3\) renames \(x\) into \(y\) (i.e.,
  \(y\) copies the current value of \(x\)) and then restarts \(x\) to \(2\).
  Let us illustrate this with the following timed run:
  \begin{align*}
    (q_0, \emptyset) &\xrightarrow{1}
    (q_0, \emptyset) \xrightarrow{i}
    (q_1, x = 2) \xrightarrow{0.5}
    (q_1, x = 1.5) \xrightarrow{i}
    (q_3, x = 1.5, y = 1)
    \\
    &\xrightarrow{1}
    (q_3, x = 0.5, y = 0) \xrightarrow{i}
    (q_3, x = 2, y = 0.5) \xrightarrow{0.5}
    (q_4, y = 1) \xrightarrow{0}
    (q_4, y = 1).
  \end{align*}
  Observe that \(y\) takes the value of \(x\) when taking the \(i\)-loop of
  \(q_3\).

  We also highlight that some of the timeout transitions restart a timer that is
  not the one timing out, as illustrated by the timed run
  \((q_0, \emptyset) \xrightarrow{1}
    (q_0, \emptyset) \xrightarrow{i}
    (q_1, x = 2) \xrightarrow{2}
    (q_1, x = 0) \xrightarrow{\timeout{x}}
    (q_2, y = 2) \xrightarrow{0}
    (q_2, y = 2)\).

  It is not hard to see that we have the following enabled timers per state:
  \(\enabled{q_0} = \emptyset\), \(\enabled{q_1} = \{x\}\),
  \(\enabled{q_2} = \enabled{q_4} = \{y\}\), and \(\enabled{q_3} = \{x, y\}\).
  From there, we conclude that \(\M\) is complete.
\end{example}

We now provide a definition of symbolic equivalence
(\Cref{def:equivalence:symbolic}) between a \gMMT and an \MMT.
First, we adapt the notion of \(x\)-spanning runs.
Recall that a run \(\pi\) of an \MMT is said \(x\)-spanning if
\(\pi = p_0 \xrightarrow[u_1]{i_1} p_1 \xrightarrow[u_2]{i_2}
\dotsb \xrightarrow{i_n} p_n\) with
\begin{itemize}
  \item
  \(u_1 = (x, c)\) for some \(c \in \natplus\),

  \item
  \(u_j \neq (x, c')\) for every \(j \in \{2, \dotsc, n - 1\}\) and
  \(c' \in \natplus\),

  \item
  \(x \in \activeTimers(p_j)\) for all \(j \in \{2, \dotsc, n - 1\}\),
  and

  \item
  \(i_n = \timeout{x}\).
\end{itemize}
The adaptation to \gMMTs is direct: the first transition must start
a timer \(x\), each update function renames the timer (in a way, \(x\) remains
active but under a different name), and the final transition reads the
corresponding timeout.
A run \(\pi\) of a \gMMT is said \emph{spanning} if
\(\pi = p_0 \xrightarrow[\updateFunction_1]{i_1}
  p_1 \xrightarrow[\updateFunction_2]{i_2}
  \dotsb \xrightarrow{i_n}
  p_n\)
and there exist timers \(x_1, \dotsc, x_n\) such that
\begin{itemize}
  \item
  \(\updateFunction_1(x_1) = c\) for some \(c \in \natplus\),

  \item
  \(\updateFunction_j(x_j) = x_{j-1}\) for every \(j \in \{2, \dotsc, n - 1\}\)
  (this implies that \(x_j \in \activeTimers(p_j)\)),
  and

  \item
  \(i_n = \timeout{x_{n-1}}\).
\end{itemize}
Observe that the notion of symbolic words (\Cref{sec:MMT:equivalence}) still
holds using this definition of spanning runs.

We can thus obtain a definition of symbolic equivalence between a
\complete \gMMT and a \complete \MMT.
In short, we impose the same constraints as in \Cref{def:equivalence:symbolic}:
\begin{itemize}
  \item 
  a run reading a symbolic word exists in the \gMMT if and only if one exists in
  the \MMT,

  \item
  if they both exist, we must see the same outputs and for transitions starting
  a timer that eventually times out during the run (i.e., the sub-run is spanning),
  we must have the same constants.
\end{itemize}
\begin{definition}[Symbolic equivalence between \gMMT and \MMT]
  Let \(\M\) be a \complete \gMMT and \(\N\) be a \complete \MMT.
  We say that \(\M\) and \(\N\) are \emph{symbolically equivalent}, also noted
  \(\M \symEquivalent \N\), if for every symbolic word
  \(\symbolic{w} = \symbolic{i_1} \dotsb \symbolic{i_n}\) over $I \cup \toevents{\natplus}$:
  \begin{itemize}
    \item  $q_0^\M \xrightarrow[\updateFunction_1]{\symbolic{i_1}/o_1} q_1
    \dotsb \xrightarrow[\updateFunction_n]{\symbolic{i_n}/o_n} q_n$ is a feasible
    run in $\M$ if and only if
    $q_0^{\N} \xrightarrow[u'_1]{\symbolic{i_1}/o'_1} q'_1
    \dotsb \xrightarrow[u'_n]{\symbolic{i_n}/o'_n} q'_n$ is a feasible run in $\N$.
    \item Moreover,
    \begin{itemize}
        \item
        $o_j = o'_j$ for all $j \in \{1, \dotsc, n\}$, and
        \item
        $q_{j-1} \xrightarrow{\symbolic{i_j} \dotsb \symbolic{i_k}} q_k$
        is spanning implies that there is timer \(x\) such that
        \(\updateFunction_j(x) = c\), \(u'_j = (x',c')\), and \(c = c'\).
    \end{itemize}
  \end{itemize}
\end{definition}
We then obtain that \(\M \symEquivalent \N\) implies that \(\M \equivalent \N\),
with arguments similar to those presented in \Cref{app:equivalence}.

\subsubsection{Existence of a symbolically equivalent \MMT.}

Let \(\M\) be a \complete \gMMT.
We give a construction of a \complete MMT \(\N\) such that
\(\M \symEquivalent \N\).
Intuitively, we rename the timers of \(\M\), on the fly, into the timers of \(\N\)
and keep track in the states of \(\N\) of the current renaming.
As each transition of \(\M\) can freely rename timers, it is
possible that \(x\) is mapped to \(x_j\) in a state of \(\N\), but mapped to
\(x_k\) (with \(k \neq j\)) in some other state.
That is, we sometimes need to split states of \(\M\) into multiple states in
\(\N\), accordingly to the update functions.
An example is given below.

Formally, we define \(\N = (X^\N, Q^\N, q_0^\N, \activeTimers^\N, \delta^\N)\) with
\begin{itemize}
  \item
  \(X^\N = \{x_1, \dotsc, x_n\}\) with
  \(n = \max_{q \in Q^\M}\lengthOf{\activeTimers^\M(q)}\).

  \item
  \(Q^\N = \{
    (q, \renamingMMT) \in Q^\M \times (\activeTimers^\M(q) \leftrightarrow X^\N)
    \mid
    \lengthOf{\ran{\renamingMMT}} = \lengthOf{\activeTimers^\M(q)}
  \}\).
  The idea is that \(\renamingMMT\) dictates how to rename a timer from
  \(\M\) into a timer of \(\N\), for this specific state.
  As said above, the renaming may change transition by transition.

  \item
  \(q_0^\N = (q_0^\M, \emptyset)\).

  \item
  \(\activeTimers^\M((q, \renamingMMT)) = \ran{\mu}\) for all
  \((q, \renamingMMT) \in Q^\N\).
  Then, \(\lengthOf{\activeTimers^\N((q, \renamingMMT))} =
  \lengthOf{\activeTimers^\M(q)}\).

  \item
  The function \(\delta^\N : Q^\N \times \actions{\N} \to
  Q^\N \times O \times \updates{\N}\) is defined as follows.
  Let \(q \xrightarrow[\updateFunction]{i/o} q'\) be a run of \(\M\) and
  \((q, \renamingMMT) \in Q^\N\).
  \begin{itemize}
    \item
    If \(i \in I\), we have two different cases depending on whether
    \(\updateFunction\) (re)starts a fresh timer or does not (re)start anything.
    That is, let \(\delta^\N((q, \renamingMMT), i) =
    ((q', \renamingMMT'), o, u)\) with \(u\) and \(\renamingMMT'\) defined
    as follows.
    \begin{itemize}
      \item
      If \(\updateFunction(x) = c \in \natplus\) for a timer \(x\),
      i.e., the transition of \(\M\) (re)starts a fresh timer,
      then, in \(\N\), we want to start a timer that is not already tied to
      some timer.
      Let \(\nu = \renamingMMT \circ (\updateFunction \setminus \{(x, c)\})\),
      i.e., the matching telling us how to rename every timer, except \(x\),
      after taking the transition.
      As \(\lengthOf{X^\N} = \max_{p \in Q^\M} \lengthOf{\activeTimers^\M(q)} =
      \max_{(p, \nu) \in Q^\N} \lengthOf{\activeTimers^\N((q, \nu))}\), it
      follows that \(\lengthOf{X^\N} > \lengthOf{\ran{\nu}}\).
      Hence, there exists a timer \(x_j \in X^\N\) such that
      \(x_j \notin \ran{\nu}\).
      We then say that \(x\) is mapped to \(x_j\) and follow \(\nu\) for the
      other timers.
      That is,
      \(u = (x_j, c)\) and
      \(\renamingMMT' = \nu \cup \{(x, x_j)\}\).

      \item
      If \(\updateFunction(x) \notin \natplus\) for any timer \(x\), i.e., the
      transition does not (re)start anything,
      then, in \(\N\), we also do not restart anything.
      Hence,
      \(u = \bot\) and
      \(\renamingMMT' = \renamingMMT \circ \updateFunction\).
    \end{itemize}

    \item
    If \(i = \timeout{x}\), we have two cases depending on whether the transition
    start a timer, or not.
    That is, we define \(\delta^\N((q, \renamingMMT), \timeout{\renamingMMT(x)})
    = ((q', \renamingMMT'), o, u)\) with \(u\) and \(\renamingMMT'\) defined
    as follows.
    \begin{itemize}
      \item
      If \(\updateFunction(y) = c \in \natplus\) for some timer \(y\),
      then, in \(\N\), we want to restart \(x\).
      That is, we restart the timer that times out.
      Again, the remaining timers simply follow \(\updateFunction\).
      Hence,
      \(u = (\renamingMMT(x), c)\) and
      \(\renamingMMT' = \left(
          \renamingMMT \circ \left(
            \updateFunction \setminus \{(y, c)\}
          \right)
        \right)
        \cup \{(y, \renamingMMT(x))\}\).

      \item
      If \(\updateFunction(y) \notin \natplus\) for any timer \(y\),
      then, in \(\N\), we do not restart anything.
      Hence,
      \(u = \bot\) and
      \(\renamingMMT' = \renamingMMT \circ \updateFunction\).
    \end{itemize}
  \end{itemize}
\end{itemize}
In order to obtain a deterministic procedure, let us assume that the fresh timer
\(x_j\) is picked with the smallest possible \(j\).
\Cref{fig:gMMT:MMT} gives the \MMT constructed from the \gMMT of \Cref{fig:gMMT}.

\begin{figure}[t]
  \centering
  \begin{tikzpicture}[
  automaton,
  node distance = 70pt and 150pt,
  state/.append style = {
    rectangle,
    minimum size = 15pt,
  }
]
  \node [state, initial]      (q0)  {\((q_0, \emptyset)\)};
  \node [state, right=of q0]  (q1)  {\((q_1, \{(x, x_1)\})\)};
  \node [state, right=of q1]  (q2)  {\((q_2, \{(y, x_1)\})\)};
  \node [state, below=of q1]  (q3 1)  {\((q_3, \{(x, x_1), (y, x_2)\})\)};
  \node [state, below=40pt of q3 1]  (q3 2)  {\((q_3, \{(x, x_2), (y, x_1)\})\)};
  \node [state, left=of q3 1]  (q4 x1)  {\((q_4, \{(y, x_1)\})\)};
  \node [state, right=of q3 1]  (q4 x2)  {\((q_4, \{(y, x_2)\})\)};

  \path
    (q0)  edge                node [above] {\(i/o\)}
                              node [below] {\(\updateFig{x_1}{2}\)} (q1)
    (q1)  edge                node [above] {\(\timeout{x_1}/o\)}
                              node [below] {\(\updateFig{x_1}{2}\)} (q2)
          edge                node {\(i/o, \updateFig{x_2}{1}\)} (q3 1)
    (q2)  edge [loop above]   node {\(i/o, \updateFig{x_1}{2}\)} ()
          edge [loop below]   node {\(\timeout{x_1}/o, \updateFig{x_1}{2}\)} ()
    (q3 1)edge [bend right=15]node ['] {\(i/o, \updateFig{x_2}{2}\)} (q3 2)
          edge                node [above] {\(\timeout{x_1}/o\)}
                              node [below] {\(\updateFig{x_1}{1}\)} (q4 x1)
          edge                node [above] {\(\timeout{x_2}/o\)}
                              node [below] {\(\updateFig{x_2}{1}\)} (q4 x2)
    (q3 2)edge [bend right=15]node ['] {\(i/o, \updateFig{x_1}{1}\)} (q3 1)
          edge [in=-20, out=180]    node [near start] {\(\timeout{x_1}/o, \updateFig{x_1}{1}\)} (q4 x1)
          edge [in=-160, out=0]   node [near start, '] {\(\timeout{x_2}/o, \updateFig{x_2}{1}\)} (q4 x2)
    (q4 x1)edge [loop above]  node {\(i/o, \updateFig{x_1}{1}\)} (q4 x1)
          edge [loop below]    node {\(\timeout{x_1}/o, \updateFig{x_1}{1}\)} (q4 x1)
    (q4 x2)edge [loop above]  node {\(i/o, \updateFig{x_2}{1}\)} ()
          edge [loop below]  node {\(\timeout{x_2}/o, \updateFig{x_2}{1}\)} (q4 x2)
  ;
\end{tikzpicture}
  \caption{The \MMT obtained from the \gMMT of \Cref{fig:gMMT}.}\label{fig:gMMT:MMT}
\end{figure}

It should be clear that \(\M\) is \complete since \(\N\) is \complete (and one
can obtain \(\M\) back from \(\N\) as a sort of homomorphic image of \(\N\)).

\begin{lemma}\label{lem:gmmt-to-mmt}
  Let \(\M\) be a \complete \gMMT and
  \(\N\) be the \MMT constructed as explained above.
  Then, \(\N\) is \complete and its number of states
  is in \(\complexity{n! \cdot \lengthOf{Q^\M}}\) with \(n = \max_{q \in Q^\M}
  \lengthOf{\activeTimers^\M(q)}\).
\end{lemma}

The following lemma highlights the relation between a transition of \(\M\)
and a corresponding transition in \(\N\).
It holds by construction of \(\N\).
\begin{lemma}
  Let \((q, \renamingMMT) \in Q^\N\) and
  \(q \xrightarrow[\updateFunction]{i/o} q' \in \runs{\M}\).
  Then, we have the transition
  \((q, \renamingMMT) \xrightarrow[u]{i'/o'} (q', \renamingMMT') \in \runs{\N}\)
  with \(o = o'\), and \(i'\) and \(u\) as follows.
  \begin{itemize}
    \item
    If \(i \in I\), then \(i' = i\) and
    \[
      u = \begin{cases*}
        (x, c) & for some \(x \notin \ran{\renamingMMT \circ \updateFunction}\),
          if \(\updateFunction(y) = c \in \natplus\) for some \(y\)
        \\
        \bot & if for all \(x\), \(\updateFunction(x) \notin \natplus\).
      \end{cases*}
    \]

    \item
    If \(\symbolic{i_j} = \timeout{x}\), then \(i' = \timeout{\renamingMMT(x)}\)
    and
    \[
      u_j = \begin{cases*}
        (\renamingMMT(x), c) &
        if there exists a timer \(y\) such that \(\updateFunction(y) = c\)
        \\
        \bot & if for all \(x\), \(\updateFunction_j(x) \notin \natplus\).
      \end{cases*}
    \]
  \end{itemize}
\end{lemma}

Then, by applying this lemma over and over on each transition along a run, we
obtain that we always see the same outputs and the same updates in both machines.
\begin{corollary}
  For every symbolic word \(\symbolic{w} = \symbolic{i_1} \dotsb \symbolic{i_n}\),
  the following are equivalent
  \begin{itemize}
    \item
    \(q_0^\M \xrightarrow[\updateFunction_1]{\symbolic{i_1}/o_1}
    q_1 \xrightarrow[\updateFunction_2]{\symbolic{i_2}/o_2}
    \dotsb \xrightarrow[\updateFunction_n]{\symbolic{i_n}/o_n}
    q_n\) is a run of \(\M\),

    \item
    \((q_0^\M, \emptyset) \xrightarrow[u_1]{\symbolic{i_1}/o_1}
    (q_1, \renamingMMT_1) \xrightarrow[u_2]{\symbolic{i_2}/o_2}
    \dotsb \xrightarrow[u_n]{\symbolic{i_n}/o_n}
    (q_n, \renamingMMT_n)\) is a run of \(\N\),
  \end{itemize}
  with \(u_j\) defined as follows for every \(j\):
  \begin{itemize}
    \item
    If \(\symbolic{i_j} \in I\), then
    \[
      u_j = \begin{cases*}
        (x_k, c) & \parbox[t]{.7\textwidth}{for some \(x_k \notin
          \ran{\renamingMMT_{j-1} \circ \updateFunction_j}\),
          if there exists \(x\) such that \(\updateFunction_j(x) = c \in \natplus\)}
        \\
        \bot & if for all \(x\), \(\updateFunction_j(x) \notin \natplus\).
      \end{cases*}
    \]

    \item
    If \(\symbolic{i_j} = \timeout{k}\), then
    \[
      u_j = \begin{cases*}
        (x, c) & \parbox[t]{.7\textwidth}{if there exists \(x\) such that
        \(\updateFunction_j(x) = c \in \natplus\) and
        \((q_{k-1}, \renamingMMT_{k-1}) \xrightarrow[(x, c')]{i_k}\)}
        \\
        \bot & if for all \(x\), \(\updateFunction_j(x) \notin \natplus\).
      \end{cases*}
    \]
  \end{itemize}
\end{corollary}

This directly implies that if a run is feasible in \(\M\), the corresponding run
is also feasible in \(\N\), and vice-versa.
\begin{corollary}
  For any symbolic word \(\symbolic{w}\),
  \(q_0^\M \xrightarrow{\symbolic{w}} {} \in \runs{\M}\) is feasible if and only
  if \(q_0^\N \xrightarrow{\symbolic{w}} {} \in \runs{\N}\) is feasible.
\end{corollary}
This is enough to obtain the desired result: \(\M \symEquivalent \N\), as any
run in one can be reproduced in the other, and we see the same
outputs and updates (for spanning sub-runs) along these runs.
\begin{corollary}
  \(\M \symEquivalent \N\).
\end{corollary}

\subsubsection{Construction of a \gMMT hypothesis.}

Let us now describe how a \complete \gMMT \(\hypothesis\) is
constructed from \(\tree\).
We assume that~\ref{item:tree:basis:timers} holds, i.e.,
\begin{itemize}
  \item
  \(\compatible(r) \neq \emptyset\) for every frontier state \(r\),
  and

  \item
  \(\lengthOf{\activeTimers^\tree(r)} = \lengthOf{\activeTimers^\tree(p)}\)
  for each \(r \in \frontier\) and \((p, m) \in \compatible(r)\).
\end{itemize}
Hence, \(m\) is bijective.
The idea is to use the set of basis states as the states of \(\hypothesis\), with
exactly the same active timers per state.
For any transition \(q \xrightarrow[u]{i/o} q' \in \runs{\tree}\) with
\(q, q' \in \basis\), we do not rename anything in \(\hypothesis\).
If \(u = (x, c)\), then the update function of \(\hypothesis\) also (re)starts
\(x\) to \(c\).
Finally, for a transition \(q \xrightarrow[u]{i/o} r \in \runs{\tree}\) with
\(r \in \basis\), we arbitrarily select a pair \((p, m) \in \compatible(r)\) and
define a transition \(q \xrightarrow[\updateFunction]{i/o} p\) where
\(\updateFunction\) renames every timer according to \(m\).
In other words, the only functions that actually rename timers come from folding
the tree, i.e., when we exit the basis in \(\tree\).

\begin{definition}[Generalized \MMT hypothesis]
  We define the \gMMT
  \(\hypothesis = (X^\hypothesis, Q^\hypothesis, q_0^\hypothesis,
  \activeTimers^\hypothesis, \delta^\hypothesis)\) where:
  \begin{itemize}
    \item
    \(X^\hypothesis = \bigcup_{q \in \basis} \activeTimers^\tree(q)\),

    \item
    \(Q^\hypothesis = \basis\), with \(q_0^\hypothesis = q_0^\tree\),

    \item
    \(\activeTimers^\hypothesis(q) = \activeTimers^\tree(q)\) for each
    \(q \in \basis\), and

    \item
    \(\delta^\hypothesis\) is constructed as follows.
    Let \(q \xrightarrow[u]{i/o} q'\) be a transition in \(\tree\) with
    \(q \in \basis\).
    We have four cases:
    \begin{itemize}
      \item
      If \(q' \in \basis\) (i.e., the transition remains within the basis)
      and \(u = \bot\),
      then we define \(\delta^\hypothesis(q, i) = (q', o, \updateFunction)\)
      with, for all \(x \in \activeTimers^\tree(q')\), \(\updateFunction(x) = x\).

      \item
      If \(q' \in \basis\) and \(u = (y, c)\),
      then we define \(\delta^\hypothesis(q, i) = (q', o, \updateFunction)\)
      with \(\updateFunction(y) = c\) and, for all
      \(x \in \activeTimers^\tree(q') \setminus \{y\}\), \(\updateFunction(x) = x\).

      \item
      If \(q' \in \frontier\) (i.e., the transition leaves the basis) and
      \(u = \bot\),
      then we select an arbitrary \((p, m) \in \compatible(r)\) and define
      \(\delta^\hypothesis(q, i) = (p, o, \updateFunction)\) with, for all
      \(x \in \activeTimers^\tree(q')\), \(\updateFunction(m^{-1}(x)) = x\).

      \item
      If \(q' \in \frontier\) and \(u = (y, c)\),
      then we select an arbitrary \((p, m) \in \compatible(r)\) and define
      \(\delta^\hypothesis(q, i) = (p, o, \updateFunction)\) with
      \(\updateFunction(m^{-1}(y)) = c\) and, for all
      \(x \in \activeTimers^\tree(q') \setminus \{y\}\),
      \(\updateFunction(m^{-1}(x)) = x\)
    \end{itemize}
  \end{itemize}
\end{definition}
It is not hard to see that \(\hypothesis\) is well-formed and \complete, as it is
constructed from \(\tree\).
Indeed, recall that \(q \xrightarrow{i} {} \in \runs{\tree}\) is defined for
each \(q \in \basis\) if and only if \(i \in I \cup \toevents{\enabled{q}[\tree]}\),
i.e., the basis is \complete.

\begin{example}
  Let \(\tree\) be the observation tree of \Cref{fig:ex:learning:hypothesis}.
  We can observe that
  \begin{align*}
    \compatible[6](t_2) &= \{(t_1, x_1 \mapsto x_1)\}
    &
    \compatible[6](t_5) &= \{(t_6, x_6 \mapsto x_1, x_3 \mapsto x_3)\}
    \\
    \compatible[6](t_{10}) &= \{(t_0, \emptyset)\}
    &
    \compatible[6](t_{11}) &= \{(t_6, x_6 \mapsto x_{11}, x_3 \mapsto x_3)\}
    \\
    \compatible[6](t_{12}) &= \{(t_0, \emptyset)\}
    &
    \compatible[6](t_{15}) &= \{(t_9, x_3 \mapsto x_3)\}.
  \end{align*}
  We thus construct a \gMMT \(\hypothesis\) as follows.
  \begin{itemize}
    \item
    The states of \(\hypothesis\) are \(t_0, t_1, t_3, t_6\), and \(t_9\).
    \item
    The \(i\)-transition from \(t_0\) to \(t_1 \in \basis\) starts the timer
    \(x_1\) to 2.
    \item
    The \(i\)-transition from \(t_1\) to \(t_3 \in \basis\) keeps \(x_1\)
    as-is and starts a new timer \(x_3\).
    \item
    The \(i\)-transition from \(t_3\) to \(t_6 \in \basis\) stops \(x_1\) and
    starts \(x_6\).
    The timer \(x_3\) is unchanged.
    Observe that, so far, we followed exactly the transitions defined within the
    basis of \(\tree\).
    \item
    We consider the \(\timeout{x_1}\)-transition from \(t_3\) to \(t_5\) in
    \(\tree\).
    As \(t_5 \in \frontier\), we select a pair \((p, m)\) in \(\compatible(t_5)\).
    Here, the only possibility is \((t_6, x_6 \mapsto x_1, x_3 \mapsto x_3)\).
    So, we define the renaming function \(\updateFunction_{5 \mapsto 6}\) such that
    \(\updateFunction_{5 \mapsto 6}(x_6) = 2\) and
    \(\updateFunction_{5 \mapsto 6}(x_3) = x_3\).
    Hence, we have the transition
    \(t_3 \xrightarrow[\updateFunction_{5 \mapsto 6}]{\timeout{x_1}/o} t_6\).
  \end{itemize}
  And so on for the remaining transitions.
  The resulting \gMMT is given in \Cref{fig:ex:gMMT:running}.
  A more complex example is provided in the next section.
\end{example}
\begin{figure}[t]
  \centering
  \begin{tikzpicture}[
  automaton,
]
  \node [state, initial]            (t0)  {\(t_0\)};
  \node [state, right=50pt of t0]   (t1)  {\(t_1\)};
  \node [state, right=50pt of t1]   (t3)  {\(t_3\)};
  \node [state, right=100pt of t3]  (t6)  {\(t_6\)};
  \node [state, right=70pt of t6]   (t9)  {\(t_9\)};

  \path
    (t0)  edge              node [above] {\(i/o\)}
                            node [below] {\(\updateFig{x_1}{2}\)}           (t1)
    (t1)  edge              node [above] {\(i/o'\)}
                            node [below, align=left]
                                        {\(\updateFig{x_1}{x_1}\),\\
                                          \(\updateFig{x_3}{3}\)}           (t3)
          edge [loop above] node {\(\timeout{x_1}/o, \updateFig{x_1}{2}\)}  (t1)
    (t3)  edge [out=-40, in=-140] node [above] {\(i/o'\)}
                            node [below] {\(\updateFig{x_6}{2},
                                            \updateFig{x_3}{x_3}\)}         (t6)
          edge [out=40, in=140]  node [above] {\(\timeout{x_1}/o\)}
                            node [below, align=right] {\(\updateFig{x_6}{2}\),\\
                                            \(\updateFig{x_3}{x_3}\)}       (t6)
    (t6)  edge              node [above] {\(\timeout{x_6}/o\)}
                            node [below] {\(\updateFig{x_3}{x_3}\)}         (t9)
          edge [loop above] node [align=right] {\(i/o', \updateFig{x_6}{2}\),\\
                                        \(\updateFig{x_3}{x_3}\)}           (t6)
    (t9)  edge [loop above] node {\(i/o', \updateFig{x_3}{x_3}\)}           (t9)
  ;

  \draw [rounded corners=10pt]
    let
      \p{e} = ($(t0.south east) + (0, -1.1)$),
      \p{s} = ($(t6)$)
    in
    (t6)  -- (\x{s}, \y{e})
          node [midway, right] {\(\timeout{x_3}/o, \bot\)}
          -- (\p{e})
          -- (t0.south east)
  ;
  \draw [rounded corners=10pt]
    let
      \p{e} = ($(t0.south west) + (0, -1.3)$),
      \p{s} = ($(t9)$)
    in
    (t9)  -- (\x{s}, \y{e})
          node [midway, right] {\(\timeout{x_3}/o, \bot\)}
          -- (\p{e})
          -- (t0.south west)
  ;
\end{tikzpicture}
  \caption{A \gMMT constructed from the observation tree of
  \Cref{fig:ex:learning:hypothesis}.}\label{fig:ex:gMMT:running}
\end{figure}

Finally, while one can convert the \gMMT hypothesis \(\hypothesis\) into an
\MMT hypothesis, it is not required.
In order to avoid the factorial blowup, one can simply give \(\hypothesis\)
to the teacher who has to check whether \(\hypothesis\) and its hidden \MMT
\(\M\) are symbolically equivalent.
For instance, the teacher may construct the zone \gMMT of \(\hypothesis\) and
the zone \MMT of \(\M\) (see \Cref{proof:lemma:gootMMT}) and then check the
equivalence between them.
This does not necessitate to construct an \MMT.

\subsubsection{Example of a case where \gMMTs are required.}

\begin{figure}
  \centering
  \begin{tikzpicture}[
  automaton,
]
  \node [state, initial]        (q0)  {\(q_0\)};
  \node [state, right=of q0]    (q1)  {\(q_1\)};
  \node [state, right=of q1]    (q2)  {\(q_2\)};
  \node [state, right=of q2]    (q3)  {\(q_3\)};
  \node [state, below=30pt of q2]    (q4)  {\(q_4\)};

  \path
    (q0)  edge                    node [above] {\(i\)}
                                  node [below] {\(\updateFig{x}{2}\)}     (q1)
    (q1)  edge [loop above]       node [align=center]
                                        {\(j/\bot\)\\
                                        \(\timeout{x}/\updateFig{x}{2}\)} (q1)
          edge                    node [above] {\(i\)}
                                  node [below] {\(\updateFig{y}{1}\)}     (q2)
    (q2)  edge [loop above]       node [align=center]
                                        {\(i/\bot\)\\
                                        \(j/\bot\)}                       (q2)
          edge                    node [above] {\(\timeout{x}\)}
                                  node [below] {\(\updateFig{x}{1}\)}     (q3)
          edge                    node [left, near end] {\(\timeout{y},
                                                  \updateFig{y}{1}\)}     (q4)
    (q3)  edge [loop right]       node [align=center]
                                        {\(i/\bot\)\\
                                        \(j/\bot\)\\
                                        \(\timeout{x}/\updateFig{x}{1}\)} (q3)
    (q4)  edge [loop right]       node [align=center]
                                        {\(i/\bot\)\\
                                        \(j/\bot\)\\
                                        \(\timeout{y}/\updateFig{y}{1}\)} (q4)
  ;

  \draw [rounded corners = 10pt]
    let
      \p{0} = (q0.south),
      \p{4} = (q4.west),
    in
      (\p{0}) -- (\x{0}, \y{4})
              node [midway, left] {\(j/\updateFig{y}{1}\)}
              -- (\p{4})
  ;
\end{tikzpicture}
\begin{tikzpicture}[
  automaton,
  node distance = 33pt and 68pt,
]
  \node [state, initial above, basis]     (t0) {\(t_{0}\)};
  \node [state, basis, left=of t0]        (t1) {\(t_{1}\)};
  \node [state, above=of t1]              (t2) {\(t_{2}\)};
  \node [state, basis, right=of t0]       (t3) {\(t_{3}\)};
  \node [state, above=of t3]              (t4) {\(t_{4}\)};
  \node [state, right=of t3]              (t5) {\(t_{5}\)};
  \node [state, below=of t5]              (t6) {\(t_{6}\)};
  \node [state, above=of t4]              (t7) {\(t_{7}\)};
  \node [state, above=of t5]              (t8) {\(t_{8}\)};
  \node [state, below=of t6]              (t9) {\(t_{9}\)};
  \node [state, above=of t2]              (t10) {\(t_{10}\)};
  \node [state, basis, left=of t1]        (t11) {\(t_{11}\)};
  \node [state, left=of t11]              (t12) {\(t_{12}\)};
  \node [state, above=of t11]             (t13) {\(t_{13}\)};
  \node [state, below=of t0]              (t14) {\(t_{14}\)};
  \node [state, below=of t3]              (t15) {\(t_{15}\)};
  \node [state, above=of t12]             (t16) {\(t_{16}\)};
  \node [state, above=of t13]             (t17) {\(t_{17}\)};
  \node [state, below=of t12]             (t18) {\(t_{18}\)};
  \node [state, below=of t18]             (t19) {\(t_{19}\)};
  \node [state, left=46pt of t18]              (t20) {\(t_{20}\)};
  \node [state, below=of t11]             (t21) {\(t_{21}\)};
  \node [state, right=of t21]             (t22) {\(t_{22}\)};
  \node [state, below=of t21]             (t23) {\(t_{23}\)};
  \node [state, right=of t7]              (t24) {\(t_{24}\)};
  \node [state, right=44pt of t8]              (t25) {\(t_{25}\)};
  \node [state, right=44pt of t9]              (t26) {\(t_{26}\)};
  \node [state, right=of t10]             (t27) {\(t_{27}\)};
  \node [state, below=of t15]             (t28) {\(t_{28}\)};
  \node [state, above=of t16]             (t29) {\(t_{29}\)};

  \foreach \s/\t/\i/\u in {
    t0/t1/i/\updateFig{x_1}{2},
    t0/t3/j/\updateFig{x_3}{1},
    t1/t11/i/\updateFig{x_{11}}{1},
    t3/t5/i/\bot,
    t11/t12/\timeout{x_1}/\updateFig{x_1}{1},
    t14/t15/\timeout{x_{1}}/\updateFig{x_{1}}{2},
    t21/t22/\timeout{x_1}/\bot
  } {
    \path
      (\s) edge   node [above]  {\(\i, \u\)} (\t)
    ;
  }

  \foreach \s/\t/\i/\u/\p in {
    t1/t2/\timeout{x_1}/\updateFig{x_1}{2}/right,
    t1/t14/j/\bot/right,
    t3/t4/\timeout{x_2}/\updateFig{x_2}{1}/left,
    t3/t6/j/\bot/right,
    t4/t7/\timeout{x_2}/\updateFig{x_2}{1}/left,
    t5/t8/\timeout{x_2}/\updateFig{x_2}{1}/left,
    t6/t9/\timeout{x_2}/\updateFig{x_2}{1}/left,
    t2/t10/\timeout{x_1}/\updateFig{x_1}{2}/right,
    t12/t16/\timeout{x_1}/\updateFig{x_1}{2}/left,
    t13/t17/\timeout{x_{11}}/\bot/right,
    t18/t19/\timeout{x_1}/\bot/right,
    t18/t20/\timeout{x_{11}}/\bot/above,
    t11/t21/j/\bot/right,
    t11/t13/\timeout{x_{11}}/\updateFig{x_{11}}{1}/right,
    t11/t18/i/\bot/left,
    t7/t24/\timeout{x_2}/\bot/above,
    t8/t25/\timeout{x_2}/\bot/above,
    t9/t26/\timeout{x_2}/\bot/above,
    t10/t27/\timeout{x_1}/\bot/above,
    t15/t28/\timeout{x_1}/\bot/left,
    t16/t29/\timeout{x_1}/\bot/left,
    t21/t23/\timeout{x_{11}}/\bot/right} {
    \path
      (\s) edge node [\p] {\(\i/\u\)} (\t)
    ;
  }
\end{tikzpicture}
  \caption{An \MMT with
  \(\activeTimers(q_0) = \emptyset, \activeTimers(q_1) = \activeTimers(q_3) =
  \{x\}, \activeTimers(q_2) = \{x, y\}, \activeTimers(q_4) = \{y\}\), and an
  observation tree in which basis states are highlighted in gray.
  For simplicity, the output \(o\) of each transition is omitted.}\label{fig:app:hypothesis:needgMMT}\label{fig:app:hypothesis:needgMMT:tree}
\end{figure}

Finally, we give an example of an observation tree from which the construction of
\(\PER\) (as explained earlier) fails.
That is, we obtain \(x \PER y\) but \(x \timerApart y\).
Let \(\M\) be the \MMT of \Cref{fig:app:hypothesis:needgMMT}.
For simplicity, we omit all outputs in this section.

Observe that replacing the transition \(q_0 \xrightarrow[(y, 1)]{j} q_4\) by
\(q_0 \xrightarrow[(x, 1)]{j} q_3\) would yield an \MMT symbolically equivalent
to \(\M\).
Indeed, both \(q_3\) and \(q_4\) have the same behavior, up to a renaming of
the timer.
So, the \(j\)-transition from \(q_0\) can freely go to \(q_3\) or \(q_4\), under
the condition that it starts respectively \(x\) or \(y\).
Here, we fix that it goes to \(q_4\) and starts \(y\).
However, the learning algorithm may construct a hypothesis where it instead goes
to \(q_3\).
In short, this uncertainty will lead us to an invalid \(\PER\).

Let \(\tree\) be the observation tree of \Cref{fig:app:hypothesis:needgMMT:tree}.
One can check that \(\tree\) is an observation tree for \(\M\), i.e., there
exists a functional simulation \(\funcSim : \tree \to \M\).
We have the following compatible sets:
\begin{align*}
  \compatible(t_2) = \compatible(t_{14}) &= \{(t_1, x_1 \mapsto x_1)\}
  \\
  \compatible(t_4) = \compatible(t_5) = \compatible(t_6) &=
    \{(t_3, x_2 \mapsto x_2)\}
  \\
  \compatible(t_{12}) &= \{(t_3, x_2 \mapsto x_1)\}
  \\
  \compatible(t_{13}) &= \{(t_3, x_2 \mapsto x_{11})\}
  \\
  \compatible(t_{18}) &= \{(t_{11}, x_1 \mapsto x_1, x_{11} \mapsto x_{11})\}
\end{align*}
By constructing the equivalence relation
\(\PER {} \subseteq \{x_1, x_2, x_{11}\} \times \{x_1, x_2, x_{11}\}\), we
obtain that
\(x_1 \PER x_2\)
due to \((t_3, x_2 \mapsto x_1) \in \compatible(t_{12})\), and
\(x_2 \PER x_{11}\)
due to \((t_3, x_2 \mapsto x_{11}) \in \compatible(t_{13})\).
So, \(x_1 \PER x_{11}\).
However, notice that \(x_1 \timerApart x_{11}\), as both timers are active in
\(t_{11}\).
Since that relation is the only possibility, we conclude that it is not always
possible to construct a relation that does not put together two apart timers.

\subsection{Proof of \texorpdfstring{\Cref{thm:learning:termination}}{Theorem 1}}\label{app:learning:termination}

Let us now show \Cref{thm:learning:termination}, which gives the termination
and complexity of \lsharpMMT.
Before that, we introduce an optimization of our learning algorithm that
reduces the size of each compatible set, by leveraging weak co-transitivity
(\Cref{lemma:coTransitivity}).
In turn, this diminishes the number of hypotheses that can be constructed and,
thus, helps the learner to converge towards the target \MMT while reducing the
number of equivalence queries needed.

Let \(r \in \frontier\) and
assume we have \((p, \mu) \in \compatible(r)\) and
\((p', \mu') \in \compatible(r)\) with \(p \neq p'\) and maximal matchings
\(\mu : p \leftrightarrow r\) and \(\mu' : p' \leftrightarrow r\).
These matchings are necessarily valid by definition of \(\compatible\).
We also assume that
\(\lengthOf{\activeTimers^\tree(p)} = \lengthOf{\activeTimers^\tree(r)}
= \lengthOf{\activeTimers^\tree(p')}\) (by applying the above idea).
As \(p, p' \in \basis\), it must be that \(p \apart^m p'\) for any
maximal matching \(m : p \leftrightarrow p'\).
In particular, take \(m : p \leftrightarrow p'\) such that
\(m = \mu'^{-1} \circ \mu\).
Notice that \(\dom{m} \subseteq \dom{\mu}\) and \(\mu' = \mu \circ m^{-1}\)
(see \Cref{fig:matching:composition} for a visualization).
It always exists and is unique as the three states have the same number of
active timers.
There are two cases: either any witness \(w \witness p \apart^m p'\) is such that
the apartness is structural, in which case we cannot apply
\Cref{lemma:coTransitivity}, or there is a witness \(w \witness p \apart^m p'\)
where the apartness is behavioral.
In that case, let also 
\(w^x\) be as described in \Cref{lemma:coTransitivity}.
We then replay the run \(p \xrightarrow{w \cdot w^x}\) from \(r\) using \(\mu\).
We have three cases:
\begin{itemize}
  \item \(\replay[\mu]{r}[p \xrightarrow{w \cdot w^x}] = \APART\), meaning that \(p \apart^\mu r\).
  Then, \((p, \mu)\) is no longer in \(\compatible(r)\).
  \item \(\replay[\mu]{r}[p \xrightarrow{w \cdot w^x}] = \ACTIVE\), in which case we
  discovered a new active timer in \(r\).
  Hence, we now have that \(\lengthOf{\activeTimers^\tree(p)} \neq
  \lengthOf{\activeTimers^\tree(r)}\) and we can reapply the idea of replaying
  runs showcasing timeouts (see \Cref{ex:learning:basis})
  to obtain the equality again, or that \(p\) and
  \(r\) are not compatible anymore.
  \item \(\replay[\mu]{r}[p \xrightarrow{w \cdot w^x}] = \DONE\), meaning that we
  could fully replay \(p \xrightarrow{w \cdot w^x}\) and thus did not obtain
  \(p \apart^\mu r\).
  By \Cref{lemma:coTransitivity}, it follows that \(p' \apart^{\mu'} r\).
\end{itemize}
Hence, it is sufficient to call \(\replay[\mu]{r}[p \xrightarrow{w \cdot w^x}]\)
when \(w \witness p \apart^{\mu'^{-1} \circ \mu} p'\) is behavioral.

Unlike in~\cite{VaandragerGRW22},
we cannot always obtain \(\lengthOf{\compatible(r)} = 1\),
as \Cref{lemma:coTransitivity} cannot be applied when the considered apartness
pairs are structural.

\begin{figure}[t]
  \centering
  \begin{tikzpicture}[
  automaton,
  new/.style = {
    dashed,
  },
]
  \node [state, initial, initial distance=7pt]    (t0)  {\(t_0\)};
  \node [state, right=40pt of t0]                 (t1)  {\(t_1\)};
  \node [state, above right=7pt and 70pt of t1]   (t2)  {\(t_2\)};
  \node [state, below right=7pt and 70pt of t1]   (t3)  {\(t_3\)};
  \node [state, above right=8pt and 75pt of t2]   (t4)  {\(t_4\)};
  \node [state, above right=8pt and 75pt of t3]   (t5)  {\(t_5\)};
  \node [state, below right=7pt and 75pt of t3]   (t6)  {\(t_6\)};
  \node [state, right=75pt of t4]                 (t7)  {\(t_7\)};
  \node [state, right=75pt of t5]                 (t8)  {\(t_8\)};
  \node [state, right=75pt of t6]                 (t9)  {\(t_9\)};
  \node [state, below=15pt of t9]                 (t10) {\(t_{10}\)};
  \node [state, new, above=16pt of t4]            (t11) {\(t_{11}\)};
  \node [state, new, above=16pt of t7]            (t12) {\(t_{12}\)};
  \node [state, new, right=60pt of t12]           (t13) {\(t_{13}\)};
  \node [state, new, below=15pt of t13]           (t14) {\(t_{14}\)};

  \path
    (t0)  edge                  node [above] {\(i/o\)}
                                node [below] {\(\updateFig{x_1}{2}\)}           (t1)
    (t1)  edge                  node [sloped, '] {\(i/o', \updateFig{x_3}{3}\)} (t3)
          edge                  node [sloped] {\(\timeout{x_1}/o,
                                              \updateFig{x_1}{2}\)}             (t2)
    (t2)  edge                  node [sloped] {\(\timeout{x_1}/o, \bot\)}       (t4)
    (t3)  edge                  node [sloped] {\(\timeout{x_1}/o,
                                              \updateFig{x_1}{2}\)}             (t5)
          edge                  node [sloped, '] {\(i/o', \updateFig{x_6}{2}\)} (t6)
    (t5)  edge                  node {\(\timeout{x_3}/o, \bot\)}                (t8)
    (t6)  edge                  node {\(\timeout{x_6}/o, \bot\)}                (t9)
  ;

  \path [new]
    (t11) edge                  node [above=-2pt]
                                    {\(\timeout{x_1}/o, \updateFig{x_1}{2}\)}   (t12)
    (t12) edge                  node [sloped, ']{\(\timeout{x_{11}}/o, \bot\)}  (t14)
          edge                  node [above=-2pt] {\(\timeout{x_1}/o, \bot\)}   (t13)
  ;

  \draw [rounded corners = 10pt]
    let
      \p{s} = (t5.45),
      \p{t} = (t7.-180),
    in
      (\p{s}) -- ($(\x{s}, \y{t}) + (0.1, 0)$)
              -- (\p{t})
              node [above=-1pt, midway] {\(\timeout{x_1}/o, \bot\)}
  ;

  \draw [rounded corners = 10pt]
    let
      \p{s} = (t6.-43),
      \p{t} = (t10.180),
    in
      (\p{s}) -- ($(\x{s}, \y{t}) + (0.2, 0)$)
              -- (\p{t})
              node [above, midway] {\(\timeout{x_3}/o, \bot\)}
  ;

  \draw [rounded corners = 10pt, new]
    let
      \p{s} = (t2.90),
      \p{t} = (t11.180),
    in
      (\p{s}) -- ($(\x{s}, \y{t}) + (0.1, 0)$)
              -- (\p{t})
              node [above=-2pt, midway] {\(i/o', \updateFig{x_{11}}{3}\)}
  ;
\end{tikzpicture}
  \caption{Extension of the observation tree of \Cref{fig:ex:tree} obtained by
  calling \(\replay[x_1 \mapsto x_1]{t_2}[\pi] \),
  where \( \pi = t_1 \xrightarrow{i \cdot \timeout{x_1} \cdot \timeout{x_3}}\).
  New states and transitions are highlighted with dashed lines.}\label{fig:app:ex:learning:replay}
\end{figure}
\begin{example}\label{ex:learning:minimization}
  Let the \MMT of \Cref{fig:ex:MMT:good} be the \MMT of the teacher and
  \(\tree\) be the observation tree of \Cref{fig:ex:tree}.
  As explained in \Cref{ex:learning:basis}, we have
  \(\compatible(t_2) = \{(t_1, x_1 \mapsto x_1), (t_3, x_1 \mapsto x_1)\}\).
  Let us extend the tree in order to apply weak co-transitivity to deduce that
  \(t_1 \apart^{x_1 \mapsto x_1} t_2\) or \(t_1 \apart^{x_1 \mapsto x_1} t_3\).
  We have that \(i \witness t_1 \apart^{x_1 \mapsto x_1} t_3\) due
  to~\eqref{eq:apartness:constants}.
  Hence, we replay the run \(\pi = t_1 \xrightarrow{i \cdot \timeout{x_1} \cdot
  \timeout{x_3}}\) from \(t_2\) using the matching \(x_1 \mapsto x_1\) (i.e.,
  we have \(w^x = \timeout{x_1} \cdot \timeout{x_3}\)).
  That is, we call \(\replay[x_1 \mapsto x_1]{t_2}[\pi]\).
  The resulting tree is given in \Cref{fig:app:ex:learning:replay}.
  Recall that the function returned \DONE.
  So, \(\lnot (t_1 \apart^{m_1 \mapsto x_1} t_2)\).
  By \Cref{lemma:coTransitivity}, it must be that
  \(t_3 \apart^{x_1 \mapsto x_1} t_2\).
  It is indeed the case as \(i \witness t_3 \apart^{x_1 \mapsto x_1} t_2\)
  by~\eqref{eq:apartness:constants}.
  Hence, we now have \(\compatible(t_2) = \{(t_1, x_1 \mapsto x_1)\}\).
\end{example}

Using this idea, we then add a new step inside the refinement loop, taking place
after \minimizationActive:
\begin{description}
  \item[\minimizationCoTrans]
  where \minimizationCoTrans stands for Weak Co-Transitivity.
  As explained above, we minimize each compatible set by extending the tree to
  leverage \Cref{lemma:coTransitivity} as much as possible.
\end{description}

We require several intermediate results.
First, we argue that the refinement loop of our algorithm eventually
terminates, i.e., a hypothesis is eventually constructed.
Under the assumption that the basis is finite, observe that the frontier is then
finite (as \(I \cup \toevents{\cup_{p \in \basis} \activeTimers^\tree(p)}\) is
finite).
Hence, we can apply \completion, \minimizationActive, \minimizationCoTrans only
a finite number of times.
It is thus sufficient to show that the basis cannot grow forever, i.e., that
\promotion is applied a finite number of times, which implies that \seismic is
also applied a finite number of times.
The next lemma states an upper bound over the number of basis states.
In short, if this does not hold, one can construct a matching
\(m_g : p \leftrightarrow p'\) such that \(g(x) = g(m(x))\) for all
\(x \in \dom{m_g}\) for two states \(p, p'\) such that \(f(p) = f(p')\).
By the contrapositive of the second part of \Cref{thm:extension-n-soundness}, we conclude that
either \(p\) or \(p'\) cannot be in the basis.

\begin{lemma}\label{lemma:learning:basis:bound}
  \(\lengthOf{\basis} \leq \lengthOf{Q^\M} \cdot 2^{\lengthOf{X^\M}}\).
\end{lemma}
\begin{proof}
  In order to distinguish the theoretical definition of \(\basis\) and its
  computation, let us denote by \(B\) the basis as computed in the refinement
  loop of \lsharpMMT.
  We highlight that \(B\) may not always satisfy the definition of \(\basis\)
  \emph{during} the refinement loop.
  Indeed, as explained in \Cref{sec:learning:algo}, when a new active timer
  is found in a basis state, we may have \(\lnot (p \apart^m p')\) for some
  \(p \neq p' \in \basis\) and maximal matching \(m : p \leftrightarrow p'\).
  We thus need to perform \seismic and recompute \(B\).
  Below, we will establish that
  \(B \leq \lengthOf{Q^\M} \cdot 2^{\lengthOf{X^\M}}\).
  The same (or even simpler) arguments yield the bound for $\basis$.

  Let us assume we already treated the pending \seismic (if there is one),
  i.e.,
  \begin{equation}\label{proof:lemma:learning:basis:bound:invariant}
    \text{\(\forall p, p' \in B : p \neq p' \implies p \apart^m p'\)
    for all maximal matchings \(m : p \leftrightarrow p'\)}
  \end{equation}
  Towards a contradiction, assume \(\lengthOf{B} > \lengthOf{Q^\M} \cdot
  2^{\lengthOf{X^\M}}\).
  By Pigeonhole principle, there must exist two states \(p \neq p' \in B\) such
  that \(f(p) = f(p')\) (as \(\lengthOf{B} > \lengthOf{Q^\M}\)) and which furthermore satisfy
  \(g(\activeTimers^\tree(p)) = g(\activeTimers^\tree(p'))\).
  This implies that \(\lengthOf{\activeTimers^\tree(p)} =
  \lengthOf{\activeTimers^\tree(p')}\).
  Let \(m_g : p \leftrightarrow p'\) be the matching such that \(g(x) = g(m_g(x))\)
  for all \(x \in \dom{m_g}\).
  It necessarily exists as
  \(g(\activeTimers^\tree(p)) = g(\activeTimers^\tree(p'))\).
  Moreover, \(m_g\) is maximal.
  Hence, we have \(p \apart^{m_g} p'\)
  by~\eqref{proof:lemma:learning:basis:bound:invariant}.
  However, as \(f(p) = f(p')\) and \(g(x) = g(m_g(x))\) for all \(x \in \dom{m_g}\),
  and by the contrapositive of the second part of \Cref{thm:extension-n-soundness}, it follows that
  \(\lnot (p \apart^{m_g} p')\).
  We thus obtain a contradiction, meaning that both \(p\) and \(p'\) cannot be
  in \(B\) at the same time.
  Hence, \(\lengthOf{B} \leq \lengthOf{Q^\M} \cdot 2^{\lengthOf{X^\M}}\).
\qed\end{proof}

We immediately get bounds on the size of the frontier and the
compatible sets.
For the former, note that the number of immediate successors of each state is bounded by $\lengthOf{I} + \lengthOf{X^{\M}} = \lengthOf{\actions{\M}}$;
for the latter, the number of maximal matchings is at most
\(\lengthOf{X^{\M}}!\) since the number of active timers in any state from the observation tree is bounded by the same value from the hidden MMT $\M$.

\begin{corollary}\label{cor:learning:frontier_compat:bounds}
  \(\lengthOf{\frontier} \leq \lengthOf{Q^\M} \cdot
      2^\lengthOf{X^\M} \cdot (\lengthOf{\actions{\M}})\)
  and \(\max_{r \in \frontier} \lengthOf{\compatible(r)} \leq
      \lengthOf{Q^\M} \cdot 2^\lengthOf{X^\M} \cdot \lengthOf{X^\M}!\;\).
\end{corollary}

Finally, we give an upper bound
over the length of the minimal words ending in \(\timeout{x}\) for any
\(q \in Q^\tree\) and \(x \in \activeTimers^\tree(q)\).
That is, when applying \minimizationCoTrans, one can seek a short run to be
replayed (typically, via a BFS).

\begin{lemma}\label{lemma:bound:length:timeout}
  \(
    \max_{q \in Q^\tree}
      \max_{x \in \activeTimers^\tree(q)}
        \min_{q \xrightarrow{w \cdot \timeout{x}} {} \in \runs{\tree}}
          \lengthOf{w \cdot \timeout{x}}
    \leq \lengthOf{Q^\M}
  \).
\end{lemma}
\begin{proof}
  Towards a contradiction, assume that
  \[
    \max_{q \in Q^\tree}
      \max_{x \in \activeTimers^\tree(q)}
        \min_{q \xrightarrow{w \cdot \timeout{x}} {} \in \runs{\tree}}
          \lengthOf{w \cdot \timeout{x}}
    > \lengthOf{Q^\M}.
  \]
  Then, there must exist a state \(p_0\) and a timer
  \(x \in \activeTimers^\tree(p_0)\) such that the length of
  \(w \cdot \timeout{x}\) is strictly greater than the number of states of \(\M\),
  i.e., we have a run
  \(
    \pi = p_0 \xrightarrow{i_1}
    \dotsb \xrightarrow{i_\ell}
    p_\ell \xrightarrow{\timeout{x}}
    \in \runs{\tree}
  \)
  with \(\ell > \lengthOf{Q^\M}\).
  Observe that \(p_\ell \in \explored\), as \(p_\ell \xrightarrow{\timeout{x}}\)
  is defined.
  By Pigeonhole principle, we thus have \(f(p_j) = f(p_\ell)\) for some
  \(j \in \{1, \dotsc, \ell - 1\}\).
  As \(\explored\) is tree-shaped, it follows that \(p_j \in \explored\) (since
  \(j < \ell\)).

  We thus need to argue that \(p_j \xrightarrow{\timeout{x}} {} \in \runs{\tree}\)
  to obtain our contradiction.
  Since \(f(p_j) = f(p_\ell)\), it naturally follows that \(\enabled{f(p_j)}[\M] =
  \enabled{f(p_\ell)}[\M]\).
  Moreover, \(x \in \activeTimers^\tree(p_j)\) as \(x\) is active in both
  \(p_0\) and \(p_\ell\).
  Hence, \(g(x) \in \enabled{f(p_j)}[\M]\) as \(g(x) \in \enabled{f(p_\ell)}[\M]\).
  So, it must be that \(p_j \xrightarrow{\timeout{x}} {} \in \runs{\tree}\)
  since \(p_j \in \explored\).
  We thus have a contradiction as
  \(p_0 \xrightarrow{i_1 \dotsb i_j \cdot \timeout{x}} {} \in \runs{\tree}\)
  and \(j < \ell\).
\qed\end{proof}

Let us now prove \Cref{thm:learning:termination}, which we repeat.

\termination*

\begin{proof}
  Let us start with showing that the algorithm eventually terminates.
  First, we formally prove that the refinement loop always finishes, i.e., that
  the basis and the frontier always stabilize.
  By \Cref{lemma:learning:basis:bound} and
  \Cref{cor:learning:frontier_compat:bounds}, we have
  \begin{align}
    \lengthOf{\basis} &\leq \lengthOf{Q^\M} \cdot 2^{\lengthOf{X^\M}}
    \label{eq:proof:thm:learning:termination:basis}
    \\
    \lengthOf{\frontier} &\leq \lengthOf{Q^\M} \cdot 2^{\lengthOf{X^\M}}
      \cdot (\lengthOf{\actions{\M}})
    \label{eq:proof:thm:learning:termination:frontier}
    \\
    \max_{r \in \frontier} \lengthOf{\compatible(r)} &\leq \lengthOf{Q^\M}
      \cdot 2^{\lengthOf{X^\M}} \cdot \lengthOf{X^\M}!
    \label{eq:proof:thm:learning:termination:compat}
  \shortintertext{Furthermore, by~\eqref{eq:simulation:active:size}}
    \max_{q \in Q^\tree} \lengthOf{\activeTimers^\tree(q)} &\leq
      \lengthOf{\activeTimers^\tree(f(q))} \leq \lengthOf{X^\M}.
    \label{eq:proof:thm:learning:termination:active}
  \end{align}
  Let us argue that each part of the refinement loop is applied finitely many
  times.
  We write \(\lengthOf{\seismic}\) for the number of times \seismic
  (see \Cref{sec:learning:algo}) is applied.
  \begin{itemize}
    \item
    The maximal number of applied \seismic \emph{per basis state} is bounded
    by \(\lengthOf{X^\M}\), by~\eqref{eq:proof:thm:learning:termination:active}.
    Indeed, each \seismic event is due to the discovery of a new active timer
    in a basis state.
    Hence, by~\eqref{eq:proof:thm:learning:termination:basis},
    \begin{equation}
      \lengthOf{\seismic}
      \leq
      \lengthOf{\basis} \cdot \lengthOf{X^\M}
      \leq
      \lengthOf{Q^\M} \cdot \lengthOf{X^\M} \cdot 2^{\lengthOf{X^\M}}.
      \label{eq:proof:thm:learning:termination:seismic}
    \end{equation}
    \item
    By~\eqref{eq:proof:thm:learning:termination:basis}, the number
    of times \promotion is applied between two instances of \seismic is bounded
    by \(\lengthOf{Q^\M} \cdot 2^{\lengthOf{X^\M}}\).
    So,
    \begin{equation}
      \lengthOf{\promotion}
      \leq
      \lengthOf{\basis} \cdot \lengthOf{\seismic}
      \leq {\lengthOf{Q^\M}}^2 \cdot \lengthOf{X^\M} \cdot 2^{2\lengthOf{X^\M}}.
      \label{eq:proof:thm:learning:termination:promotion}
    \end{equation}
    \item
    Between two cases of \seismic, the number of \completion is bounded by
    \(\lengthOf{I} \cdot \lengthOf{\basis}\), as, in the
    worst case, each input-transition is missing from each basis state.
    In general, we may have multiple frontier states \(r_1, \dotsc, r_n\)
    such that \(f(r_1) = \dotsb = f(r_n)\).
    Thus, \(\compatible(r_1) = \dotsb = \compatible(r_n)\), after minimizing
    each set.
    Due to \seismic, \lsharpMMT potentially has to choose multiple times
    one of those states.
    So, in the worst case, we select a different \(r_j\) each time.
    As each new basis state may not have all of its outgoing transitions,
    \begin{equation}
      \lengthOf{\completion}
      \leq
      \lengthOf{I} \cdot \lengthOf{\basis} \cdot \lengthOf{\seismic}
      = \lengthOf{I} \cdot \lengthOf{\promotion}
      \leq
      \lengthOf{I} \cdot {\lengthOf{Q^\M}}^2 \cdot \lengthOf{X^\M}
      \cdot 2^{2\lengthOf{X^M}}.
      \label{eq:proof:thm:learning:termination:completion}
    \end{equation}
    \item
    Between two instances of \seismic, the number of pairs
    \((p, m) \in \compatible(r)\) such that
    \(\lengthOf{\activeTimers^\tree(p)} \neq \lengthOf{\activeTimers^\tree(r)}\)
    is directly given by~\eqref{eq:proof:thm:learning:termination:compat}
    for each frontier state \(r\):
    \begin{equation}
      \lengthOf{\minimizationActive}
      \!\leq
      \lengthOf{\frontier} \cdot \max_{r \in \frontier} \lengthOf{\compatible(r)}
        \cdot \lengthOf{\seismic}
      \leq
      \lengthOf{\actions{\M}} \cdot \lengthOf{Q^\M}^3
        \cdot \lengthOf{X^\M} \cdot 2^{3\lengthOf{X^\M}} \cdot \lengthOf{X^\M}!\,.
      \label{eq:proof:thm:learning:termination:minimization:active}
    \end{equation}
    \item
    With similar arguments,
    \begin{equation}
      \lengthOf{\minimizationCoTrans}
      \leq
      \lengthOf{\frontier} \cdot
        {\left(\max_{r \in \frontier} \lengthOf{\compatible(r)}\right)}^2
        \cdot \lengthOf{\seismic}
      \leq
      \lengthOf{\actions{\M}}
        \cdot \lengthOf{Q^\M}^4 \cdot \lengthOf{X^\M}^2 \cdot
        2^{4 \lengthOf{X^\M}} \cdot {\left(\lengthOf{X^\M}!\right)}^2.
      \label{eq:proof:thm:learning:termination:minimization:coTrans}
    \end{equation}
  \end{itemize}
  Since each part of the refinement loop can only be applied a finite number
  of times, it follows that the loop always terminates.
  Thus, it remains to prove that \lsharpMMT constructs finitely many
  hypotheses.
  Given how a counterexample is processed, it is now hard to see that
any counterexample results in a new timer being
  discovered for a state in the basis (leading to an occurrence of \seismic),
  or a compatibility set decreasing in size (potentially leading to a \promotion).
  We already know that \(\lengthOf{\seismic}\) and \(\lengthOf{\promotion}\) are
  bounded by a finite constant.
  Moreover, by~\eqref{eq:proof:thm:learning:termination:compat}, each compatible
  set contains finitely many pairs.
  So, there can only be finitely many counterexamples, and, thus, hypotheses.
  More precisely, the number of hypotheses is bounded by
  \begin{equation}
    \lengthOf{\seismic} \cdot \lengthOf{\promotion} \cdot
    \max_{r \in \frontier} \lengthOf{\compatible(r)}
    \leq
    \lengthOf{Q^\M}^4 \cdot \lengthOf{X^\M}^2 \cdot 2^{4 \lengthOf{X^\M}}
      \cdot \lengthOf{X^\M}!\;.
    \label{eq:proof:thm:learning:termination:hypotheses}
  \end{equation}

  To establish that $\N$ is equivalent to $\M$, we observe
  that the last equivalence query to the teacher confirmed they are
  symbolically equivalent.
  Hence, by \Cref{lemma:symEquivalent:equivalent}, they are also timed
  equivalent, i.e., \(\N \equivalent \M\).
  From our construction of an \MMT hypothesis based on an \gMMT
  (see \Cref{app:learning:hypo}), we get that the intermediate \gMMT has at most
  \(\lengthOf{Q^\M} \cdot 2^\lengthOf{X^\M}\) states
  (by~\eqref{eq:proof:thm:learning:termination:basis}).
  Hence, the final \MMT has at most
  \(\lengthOf{Q^\M} \cdot 2^\lengthOf{X^\M} \cdot \lengthOf{X^\M}!\) states by \Cref{lem:gmmt-to-mmt}, i.e., a number that is polynomial in
  \(\lengthOf{Q^\M}\) and factorial in \(\lengthOf{X^\M}\), as announced.

  We now prove the claimed number of queries.
  We start with the number of queries per step of the refinement loop and
  to process a counterexample.
  \begin{description}
    \item[\seismic]
    Applying \seismic does not require any symbolic queries.

    \item[\promotion]
    Let \(r\) be the state newly added to \(\basis\).
    We thus need to do a wait query in each \(r'\) such that
    \(r \xrightarrow{i} r'\) for some \(i \in \actions{\tree}\).
    By~\eqref{eq:proof:thm:learning:termination:active}, there are at most
    \(\lengthOf{\actions{\M}}\) wait queries.

    \item[\completion]
    A single application of \completion requires a single symbolic output query
    and a single wait query.

    \item[\minimizationActive]
    Each occurrence of \minimizationActive necessitates to replay a run \(\pi\)
    from state \(q\).
    Let \(n\) be the number of transitions in \(\pi\).
    In the worst case, we have to perform \(n\) symbolic output queries
    and \(n\) symbolic wait queries.
    By \Cref{lemma:bound:length:timeout}, \(n \leq \lengthOf{Q^\M}\), if we
    always select a minimal run ending in the timeout of the desired timer.

    \item[\minimizationCoTrans]
    Likewise, applying \minimizationCoTrans requires to replay a run of length
    \(n\), i.e., we do \(n\) symbolic output queries and \(n\) wait queries.
    This time, let us argue that we can always select a run such that:
    \(
      n \leq \lengthOf{\basis} + 1 + \lengthOf{Q^\M} + \ell +
        (\lengthOf{\basis} + 1) \cdot \lengthOf{\seismic}.
    \)
    (Recall that \(\ell\) is the length of the longest counterexample.)
    Let us decompose the summands appearing on the right piece by piece:
    \begin{itemize}
      \item
      \(\lengthOf{\basis} + 1\) denotes the worst possible depth for a frontier
      state.
      Indeed, it may be that all basis states are on a single branch.
      So, the frontier states of the last basis state of that branch are at
      depth \(\lengthOf{\basis} + 1\).
      \item
      \(\lengthOf{Q^\M}\) comes from \Cref{lemma:bound:length:timeout},
      as~\eqref{eq:apartness:enabled} requires to see the timeout of some timer.
      \item
      \(\ell\) comes from the counterexample processing (see next item).
      \item
      When we previously replayed a witness of apartness due to some occurrences
      of \minimizationCoTrans, we had to copy runs from a frontier state.
      Since the worst possible depth of a frontier state is
      \(\lengthOf{\basis} + 1\), this means we added (at most) that length of the copied run when counted from the root of the observation tree.
      Since these replays may have triggered some instances of \seismic, the basis
      must have been recomputed each time.
      As explained above, we may not obtain the same exact basis, but the bound
      over the number of states is
      still~\eqref{eq:proof:thm:learning:termination:basis}.
      So, in the worst case, we add $\lengthOf{\seismic}$ many times \((\lengthOf{\basis} + 1)\) to the longest branch of the tree.
    \end{itemize}

    \item[Processing a counterexample]
    First, in the worst case, we have to add the complete counterexample to
    observe what is needed, creating a new
    run in the tree, whose length is thus \(\ell\).
    Recall that each iteration of the counterexample processing splits a run
    \(p \xrightarrow{v}\) with \(p \in \basis\) into
    \(p \xrightarrow{v'} r \xrightarrow{v''}\) such that \(v = v' \cdot v''\)
    and \(r \in \frontier\).
    It may be that every \(v'\) is of length 1, meaning that we replay
    runs of lengths \(\ell, \ell - 1, \dotsc, 1\).
    So, we do
    \(\ell + \ell - 1 + \dotsb + 1 = \frac{\ell^2 + \ell}{2}\)
    symbolic output queries and the same number of wait queries.
  \end{description}
  By combining with the bounds
  of~\eqref{eq:proof:thm:learning:termination:seismic}
  to~\eqref{eq:proof:thm:learning:termination:minimization:coTrans}, we obtain
  the following bounds.
  \begin{itemize}
    \item
    The number of symbolic output queries is bounded by
    \begin{equation*}
      1 \cdot \lengthOf{\completion}
      + \lengthOf{Q^\M} \cdot \lengthOf{\minimizationActive}
      + \left(
        \lengthOf{\basis} + 1 + \lengthOf{Q^\M} + \ell
        + (\lengthOf{\basis} + 1) \cdot \lengthOf{\seismic}
      \right) \cdot \lengthOf{\minimizationCoTrans}
    \end{equation*}
    Clearly, \((\ell + \lengthOf{\basis} \cdot \lengthOf{\seismic}) \cdot
    \lengthOf{\minimizationCoTrans}\) is bigger than the other operands
    (observe that \(\ell\) is independent from \(\lengthOf{Q^\M}, \lengthOf{I}\),
    and \(\lengthOf{X^\M}\)).
    Hence, the number of symbolic output queries is in
    \[
      \complexity{(
        \ell
        +
        \lengthOf{Q^\M}^2 \cdot \lengthOf{X^\M} \cdot 2^{2\lengthOf{X^\M}}
      )
      \!\cdot\!
      \lengthOf{\actions{\M}}
      \!\cdot\!
      \lengthOf{Q^\M}^4
      \!\cdot\!
      \lengthOf{X^\M}^2
      \!\cdot\!
      2^{4\lengthOf{X^\M}}
      \!\cdot\!
      {(\lengthOf{X^\M}!)}^2
      }.
    \]
    \item
    The number of symbolic wait queries is bounded by
    \begin{multline*}
      \lengthOf{\actions{\M}} \cdot \lengthOf{\promotion}
      +
      1 \cdot \lengthOf{\completion}
      +
      \lengthOf{Q^\M} \cdot \lengthOf{\minimizationActive}
      \\
      + \left(
        \lengthOf{\basis} + 1 + \lengthOf{Q^\M} + \ell
        + (\lengthOf{\basis} + 1) \cdot \lengthOf{\seismic}
      \right) \cdot \lengthOf{\minimizationCoTrans}.
    \end{multline*}
    Again, \((\ell + \lengthOf{\basis} \cdot \lengthOf{\seismic}) \cdot
    \lengthOf{\minimizationCoTrans}\) is bigger than the other operands.
    That is, we obtain the same complexity results as for symbolic output queries.

    \item
    The number of symbolic equivalence queries is exactly the number of
    constructed hypothesis.
    It is bounded by~\eqref{eq:proof:thm:learning:termination:hypotheses}.
  \end{itemize}
  We thus obtain the announced complexity results.
\qed\end{proof}
 
\section{FDDI protocol\label{appendix:fddi_model}}

\begin{figure}[t]
  \centering
  \begin{tikzpicture}[
  automaton,
  node distance = 40pt and 160pt,
]
  \node [state, initial]      (Idle)  {Idle};
  \node [state, right=of Idle]  (Idlex)  {Idlex};
  \node [state, below=of Idlex]  (STxy)  {STxy};
  \node [state, below=of Idle]  (STy)  {STy};
  \node [state, right=of STxy]  (ATxy)  {ATxy};

  \path
    (Idle)  edge    node [left] {TT/BS, $\updateFig{y}{\mathtt{SA}}$}       (STy)
    (Idlex) edge    node [above] {\(\timeout{x}/o,\bot\)}                   (Idle)
            edge    node [right=-2pt, pos=0.65] {TT/BS, $\updateFig{y}{\mathtt{SA}}$}      (STxy)
    (STxy)  edge    node [above=-2pt] {\(\timeout{x}/o, \bot\)}             (STy)
            edge    node [above=-2pt] {$\timeout{y}/\text{ES+BA},
                                    \updateFig{y}{\mathtt{TTRT}}$}          (ATxy)
;

  \draw [rounded corners = 5pt]
    let
      \p{s} = (ATxy.north),
      \p{t} = (Idlex.east),
    in
      (\p{s}) -- (\x{s}, \y{t})
              -- (\p{t})
              node [near start, above=-2pt]
                {\(\timeout{x}/\text{EA+RT}, \updateFig{x}{y}\)}
  ;

  \draw [rounded corners = 5pt]
    let
      \p{s} = (ATxy.140),
      \p{t} = ($(Idlex.-40)$),
      \p{m1} = ($(\x{s}, \y{t}) + (0, -0.4)$),
      \p{m2} = ($(\x{t}, \y{t}) + (0, -0.4)$),
    in
      (\p{s}) -- (\p{m1})
              -- (\p{m2})
              node [near start, above=-2pt] {EA/RT, $\updateFig{x}{y}$}
              -- (\p{t})
  ;

  \draw [rounded corners = 5pt]
    let
      \p{s} = (STy.40),
      \p{t} = (Idlex.-140),
      \p{m1} = ($(\x{s}, \y{t}) + (0, -0.4)$),
      \p{m2} = ($(\x{t}, \y{t}) + (0, -0.4)$),
    in
      (\p{s}) -- (\p{m1})
              -- (\p{m2})
              node [pos=0.4, above=-2pt] {$\timeout{y}/\text{ES+RT},
                                \updateFig{x}{\mathtt{TTRT}}$}
              -- (\p{t})
  ;
\end{tikzpicture}
  \caption{A \gMMT model of one FDDI station. For the experiments we actually translated this \gMMT into an MMT (as per \Cref{lem:gmmt-to-mmt}), this streamlines the learning algorithm described in this work.}\label{fig:FDDI}
\end{figure}
By far the largest benchmark that is learned by Waga~\cite{Waga23} is a fragment of the FDDI communication protocol~\cite{Johnson87}, based on a timed automata model described in~\cite{DawsOTY95}.
FDDI (Fiber Distributed Data Interface) is a protocol for a token ring that is composed of $N$ identical
stations.
\Cref{fig:FDDI} shows a \gMMT-translation of the timed automata model for a
single station from~\cite{DawsOTY95}.
(See \Cref{def:gMMT} for the definition of \gMMTs. The FDDI model is denoted as a \gMMT to ease the explanations.)
In the initial state {\tt Idle}, the station is waiting for the token.
When the token arrives ({\tt TT}), the station begins with transmission of synchronous messages ({\tt BS}). A timer $y$ ensures that synchronous transmission ends ({\tt ES}) after exactly {\tt SA} time units,
for some constant {\tt SA} (= Synchronous Allocation).
The station also maintains a timer $x$, that expires exactly {\tt TTRT}+{\tt SA} time units after the previous receipt of the token, for some constant {\tt TTRT} (= Target Token Rotation Timer).
When synchronous transmission ends and timer $x$ has not expired yet, the station has the possibility to begin transmission of asynchronous messages ({\tt BA}).
Asynchronous transmission must end ({\tt EA}) and the token must be returned ({\tt RT}) at the latest when $x$ expires.\footnote{In location {\tt ATxy}, timer $x$ will expire before timer $y$ (we may formally prove this by computing the zone graph):
        (1) The value of $y$ in location {\tt ATxy} is at most {\tt TTRT}.
        (2) Hence, the value of $x$ in location {\tt Idlex} is at most {\tt TTRT}.
        (3) So $x$ is at most {\tt TTRT} upon arrival in location {\tt STxy}, and at most {\tt TTRT}$-${\tt SA} upon arrival in location {\tt ATxy}.
        (4) Thus $x$ is smaller than $y$ in location {\tt ATxy} and will expire first.}
Upon entering location {\tt Idlex}, we ensure that timer $x$ will expire exactly {\tt TTRT}+{\tt SA} time units after the previous {\tt TT} event.
In a FDDI token ring of size $N$, an {\tt RT} event of station $i$ will instantly trigger a {\tt TT} event of station $(i+1) \mod N$. Waga~\cite{Waga23} considered the instance with 2 stations, {\tt SA} $= 20$, and {\tt TTRT} $= 100$.
\end{document}